\documentclass[accepted=2026-02-01, letterpaper]{quantumarticle}
\pdfoutput=1
\usepackage{braket,amsmath,amsfonts,amsthm,amssymb,url,graphicx,url,xcolor,hyperref,subcaption,appendix,comment}

\usepackage[square, numbers]{natbib}
\newcommand{\id}{\hat{1}}

\newtheorem{lemma}{Lemma}[section]
\newtheorem{prop}[lemma]{Proposition}
\newtheorem{theorem}[lemma]{Theorem}

\newtheorem{cor}[lemma]{Corollary}

\newtheorem{rem}[lemma]{Remark}

\newtheorem{exam}[lemma]{Example}
\newtheorem{cexam}[lemma]{Counterexample}

\newcommand{\Hil}{\mathcal{H}}
\newcommand{\ketbra}[1]{|{#1}\rangle\langle{#1}|}

\newcommand{\NN}{\mathbb{N}}
\newcommand{\MM}{\mathbb{M}}
\newcommand{\RR}{\mathbb{R}}

\newcommand{\BB}{\mathbb{B}}

\newcommand{\E}{\mathcal{E}}
\renewcommand{\S}{\mathcal{S}}

\newcommand{\pl}{\hspace{.1cm}}

\renewcommand{\L}{\mathcal{L}}
\newcommand{\M}{\mathcal{M}}
\newcommand{\N}{\mathcal{N}}

\newcommand{\D}{\mathcal{D}}

\newcommand{\Id}{\mathrm{Id}}

\newcommand{\Ha}{H}
\newcommand{\tr}{\text{tr}}
\newcommand{\epow}[2]{F^{({#2})}_{{#1}}}
\newcommand{\A}{\mathcal{A}}
\newcommand{\B}{\mathcal{B}}
\newcommand{\C}{\mathcal{C}}
\newcommand{\X}{\mathcal{X}}

\newcommand{\mightbealign}{}
\newcommand{\mightbenewline}{}

\DeclareUnicodeCharacter{2009}{\,}

\begin{document}

\title[Self-Restricting Noise and Entropy Decay]{Self-restricting Noise and Exponential Relative Entropy Decay Under Unital Quantum Markov Semigroups}

\author[N. LaRacuente]{Nicholas LaRacuente}
\affiliation{Indiana University Bloomington, Bloomington, IN 47408, USA}



 \maketitle
 
 \begin{abstract}
States of open quantum systems often decay continuously under environmental interactions. Quantum Markov semigroups model such processes in dissipative environments. It is known that finite-dimensional quantum Markov semigroups with GNS detailed balance universally obey complete modified logarithmic Sobolev inequalities (CMLSIs), yielding exponential decay of relative entropy to a subspace of fixed point states. We analyze continuous processes that combine dissipative with Hamiltonian time-evolution, precluding this notion of detailed balance. First, we find counterexamples to CMLSI-like decay for these processes and determine conditions under which it fails. In contrast, we prove that despite its absence at early times, exponential decay re-appears for unital, finite-dimensional quantum Markov semigroups at finite timescales. Finally, we show that when dissipation is much stronger than Hamiltonian time-evolution, the rate of eventual, exponential decay toward the semigroup's decoherence-free subspace is bounded inversely in the decay rate of the dissipative part alone. Dubbed self-restricting noise, this inverse relationship arises when strong damping suppresses effects that would otherwise spread noise beyond its initial subspace.
\end{abstract}



 \section{Introduction}
 A primary challenge for quantum information theory is to understand noise-induced decay of quantum states, estimate rates, and invent strategies to improve control. Decay and decoherence are major challenges for quantum computing \cite{clerk_introduction_2010, preskill_quantum_2018} and ultimately constrain rates of quantum communication \cite{wilde_quantum_2017}. Related questions concern the emergence of classicality \cite{rebolledo_decoherence_2005} and the thermal equilibration rates of open quantum systems \cite{bardet_rapid_2023, bluhm_exponential_2022}. A common intuition is that without error correction or other protection, information stored in quantum systems decays exponentially. A broad task at the intersection of mathematical physics and quantum information is to show theoretically when exponential decay is inevitable, and when it is not.

Existing theory shows tensor-stable, exponential decay for a range of open quantum systems. Recall the notion of a quantum Markov semigroup (QMS): a family of parameterized quantum channels (completely positive, trace-preserving maps) $(\Phi^t)_{t \geq 0}$ for which $\Phi^t \circ \Phi^s = \Phi^{s + s}$ for all $t,s \geq 0$. A QMS has a Lindbladian generator given by $\L(\rho) = - i [\Ha, \rho] + \S(\rho)$ \cite{lindblad_generators_1976, gorini_completely_1976}, where $\Ha$ is a Hamiltonian (Hermitian operator) and $\S$ a linear map described as a dissipator, such that $\Phi^t = \exp( - \L t)$. We also recall Umegaki's relative entropy given in finite dimension by $D(\rho \| \omega) = \tr(\rho(\log \rho - \log \omega))$ \cite{umegaki_conditional_1962} for density matrices $\rho$ and $\omega$.  Earlier results \cite{gao_complete_2022, gao_geometric_2021, junge_stability_2022} show a particularly strong form of exponential relative entropy decay known as a complete modified logarithmic Sobolev inequality (CMLSI) for all QMSs obeying a detailed balance condition. We recall \cite{gao_fisher_2020} that a QMS has $\lambda$-CMLSI for $\lambda > 0$ toward subspace projection $\E$ if
\[ D((\Phi^t \otimes \Id^B)(\rho) \| ((\Phi^t \circ \E) \otimes \Id^B)(\rho)) \leq e^{-\lambda t} D(\rho \| (\E \otimes \Id^B)(\rho))  \]
for all $t \geq 0$, where $B$ is an arbitrary, finite-dimensional auxiliary system, and $\rho$ is any input density on the combined system. CMLSI is a tensor-stable strengthening of the longer-known modified logarithmic Sobolev inequality (MLSI) introduced in \cite{kastoryano_quantum_2013}. In the detailed balance setting, $\E$ is a uniquely defined conditional expectation projecting to the fixed point subspace of the QMS. In general, when a QMS has a fixed point subspace to which it induces decay, we define CMLSI with respect to the projection to that subspace. If a QMS induces persistent rotation without decay (for example, if it is unitary), then it may be more appropriate to regard $\E$ as a decoherence-free subspace projection such that $\Phi^t \circ \E = \E \circ \Phi^t =  R^t \circ \E = \E \circ R^t$ for some unitary rotation group $R^t$.

Detailed balance (as formally recalled Section \ref{sec:background}) applies primarily to processes that might be considered entirely dissipative, having no non-trivial Hamiltonian part. Here we aim to extend beyond the detailed balance setting by incorporating non-trivial, system-internal Hamiltonian time-evolution alongside dissipation. When Hamiltonian dynamics do not commute with the dissipator's fixed point subspace projection, they reduce the fixed point subspace. Early-time decay, however, is still toward the original fixed point subspace up to linear order in time. Hence some information that would be eventually destroyed may escape initial exponential decay.
\begin{prop}[Introduction Version of Proposition \ref{prop:decayhaml}] \label{prop:main1intro}
Let $(\Phi^t)$ be a unital, finite-dimensional QMS with generator $\L(\rho) = i [\Ha, \rho] + \S(\rho)$.
Assume $\Phi_0^t := \exp(-t \S)$ has $\lambda_0$-CMLSI and fixed point conditional expectation $\E_0$.
\begin{itemize}
    \item If time-evolution via $\Ha$ commutes with $\E_0$, then $\L$ has $\lambda_0$-CMLSI toward $\E_0$.
    \item If time-evolution generated by $\Ha$ does not commute with $\E_0$, then $(\Phi^t)$ has a unique decoherence-free subspace projection $\E \neq \E_0$, and $\L$ does not have CMLSI.
\end{itemize}
\end{prop}
Though CMLSI fails generically at short times, longer times may still effectively yield CMLSI-like behavior. The second main result herein is that for unital QMSs, exponential relative entropy decay re-appears with respect to any finite timescale, albeit with a timescale-dependent constant. 
\begin{theorem}[Introduction Version of Theorem \ref{thm:maindecay}] \label{thm:introdecay}
Let $(\Phi^t)_{t \geq 0}$ be a finite-dimensional, unital QMS. For any $\tau > 0$, there exists a constant $\epsilon_\tau < 1$ and unique decoherence-free subspace projection $\E$ for which
\begin{equation*}
\begin{split}
& D((\Phi^t \otimes \Id^B)(\rho) \| ((\E \circ \Phi^{t}) \otimes \Id^B)(\rho)) \mightbenewline \mightbealign
 \leq \epsilon_\tau^{\lfloor t / \tau \rfloor} D(\rho \| (\E \otimes \Id^B)(\rho))
\end{split}
\end{equation*}
for all $t > 0$ and input densities $\rho$, where $\lfloor \cdot \rfloor$ denotes the floor function, and $B$ is any finite-dimensional auxiliary system.
\end{theorem}
Theorem \ref{thm:introdecay} shows that even without true CMLSI, analogous forms of relative entropy decay are universal for unital, finite-dimensional QMSs. The constant $\epsilon_\tau$ is defined explicitly in the Theorem's technical version, Theorem \ref{thm:maindecay}. When decay toward the dissipator's fixed point subspace is however extremely fast, it induces Zeno dynamics, actually suppressing the spread of noise that would otherwise occur via Hamiltonian interactions. Since the channel $\E$ in Theorem \ref{thm:zvscmlsi} is a decoherence-free subspace projection, it is known that when the generator is written in form
\begin{equation} \label{eq:lindform}
    \L(\rho) = i [\Ha, \rho] + \frac{1}{2} \sum_{j \geq 1} (L_j^* L _j \rho - 2 L_j^* \rho L_j + \rho L_j^* L _j)
\end{equation}
for which the $(L_j)$ are traceless and orthonormal, a one-parameter unitary subgroup $R^t$ is generated by $\Ha$ \cite{dhahri_decoherence-free_2010, fagnola_role_2019} such that $\Phi^t \E = \E \Phi^t = R^t \E = \E R^t$.
\begin{theorem}[Introduction Version of Theorem \ref{thm:zvscmlsi}] \label{thm:zvscmlsiintro}
    Let $\L(\cdot) = i[\Ha, \cdot] + \S(\cdot)$ generate QMS $(\Phi^t)$ with decoherence-free subspace projection $\E$. Assume that $\S$ generates QMS $(\Phi_0^t)$ with $\lambda_0$-CMLSI toward fixed point conditional expectation $\E_0 \neq \E$.

    Then there are constants $c_0, c_1$, and $\delta$ scaling logarithmically with the system dimension or subalgebra indices of $\E$ and $\E_0$ such that for any given $\tau \in (0,1)$ and input $\rho$ maximizing $D(\E_0(\rho) \| \E(\rho))$,
    \[ D(\Phi^t(\E_0(\rho)) \| \Phi^t \circ \E(\rho)) \geq e^{-\lambda t} D(\E_0(\rho) \| \E(\rho)) \]
    at $t = \tau \times \lambda_0 c_0 / \| \Ha \|_{\infty}^2 \delta$ with
     \begin{equation*}
     \begin{split}
         \lambda \leq \frac{\| \Ha \|_{\infty}^2}{\lambda_0} \times \frac{\delta }{\tau c_0} \ln \Big ( \frac{c_1}{2 (1-\tau)^2 c_0^2}  \Big ) \pl.
     \end{split}
     \end{equation*}
\end{theorem}
Theorem \ref{thm:zvscmlsiintro} upper-bounds the rate of exponential decay. The title phenomenon of this paper, self-restricting noise, is that $\lambda$ is upper-bounded \textit{inversely} in $\lambda_0$. A dissipator inducing faster decay can actually slow CMLSI-like decay at intermediate timescales. The particular intermediate timescale considered by the Theorem depends on the free parameter $\tau$. Decay at times up to $\lambda_0 c_0 / \| \Ha \|_{\infty}^2 \delta$ as in the Theorem is controlled. In contrast, the bound on $\lambda$ induced by $\lambda_0$ also becomes weak at very short times. This is not because the decay rate becomes unconstrained at short times but because it is upper-bounded by \ref{prop:decayhaml} independently from $\lambda_0$ - here the restriction on linear-order decay arises not from self-restricting noise, but from the initial entropy decay having no linear-order component.  
The Theorem follows Proposition \ref{prop:normconverse}, which upper bounds trace distance decay rates toward $\E$ and might be of independent interest. Unlike previously noted results, Theorem \ref{thm:zvscmlsiintro} does not require unitality. Explicit constants are given in the technical version, \ref{thm:zvscmlsi}. In Section \ref{sec:examples}, examples illustrate Theorem \ref{thm:zvscmlsiintro}. Such an inverse relationship in CMLSI constants is needed to account for the intermediate-timescale regime of effective Zeno dynamics as noted in \cite{popkov_effective_2018}.

\section{Background and Notations} \label{sec:background}
By $\BB(\Hil)$ we denote the space of bounded operators on Hilbert space $\Hil$, and by $\D(\Hil)$ we denote the space of densities. By $\id$ we denote the identity matrix. For a unitary matrix $U$, we denote by $R_U$ the channel that applies that unitary matrix via conjugation such that $R_U(X) := U X U^\dagger$ for any $X \in \BB(\Hil)$. If there is a family of unitaries $(U_j)$, we may denote $R_j := R_{U_j}$ when it is clear from context. Similarly, we sometimes denote $R^t := R_{\exp(- i \Ha t)}$. A quantum channel is a completely positive, trace-preserving map on density matrices, usually denoted by $\Phi, \Psi,$ or $\Theta$. For products of quantum channels, we use concatenation to denote composition or the ``$\circ$" symbol to optimize readability, e.g. $\Phi \Psi (\rho) = \Phi(\Psi(\rho)) = \Phi \circ \Psi(\rho)$ for a density matrix $\rho \in \D(\Hil)$. By $\Id$ we denote the identity superoperator (which is a quantum channel).

For a normal, faithful state $\omega \in \BB(\Hil)_* \cong \D(\Hil)$, the GNS inner-product with respect to $\omega$ is defined for $x, y \in \BB(\Hil)$ as 
\begin{equation} \label{eq:tracegns}
    \braket{x,y}_{\omega} = \tr(\omega x^\dagger y)
\end{equation}
in a tracial setting. We say that
a map $\Phi$ has \textit{GNS detailed balance} when $\Phi$ is self-adjoint with respect to the GNS inner product for an invariant density $\omega$. 
We say that a map is trace-symmetric if it is GNS self-adjoint with respect to the completely mixed state. By $\Phi^*_\omega$ we denote the adjoint of $\Phi$ under the GNS inner product with state $\omega$. We denote by $\Phi^* : \BB(\Hil) \rightarrow \BB(\Hil)$ the adjoint of $\Phi$ with respect to the trace-induced inner product and by $\Phi_*$ the often equivalent pre-adjoint.

Conditional expectations are quantum channels that are idempotent and self-adjoint with respect to a GNS inner product. In particular, when a conditional expectation $\E$ is trace-symmetric,
\begin{equation}
    \tr(\E(X) Y) = \tr(\E(X) \E(Y))
\end{equation}
for any operators $X$ and $Y$. In general,
\begin{equation} \label{eq:weightblock}
\E(\rho) = \oplus_i (\tr_{i} (P_i \rho^{A_i} P_i) \otimes \sigma_i^{B_i}) \pl,
\end{equation}
where $P_i$ projects to a diagonal block of the original matrix, $\tr_i$ is a partial trace on the $i$th block and $\sigma_i$ is a positive matrix on the $i$th block. In keeping with the physics convention to study Schr\"odinger-picture channels on densities, we denote as the conditional expectation such a channel. In the more mathematical convention, our object of study would be the predual of the conditional expectation. For unital channels, and in most cases we consider, these coincide. See \cite{gao_relative_2020} and \cite[Section 2.1]{gao_complete_2022} for more on conditional expectations.

\subsection{Ways to Compare States and Channels}
To study decay of quantum states, we recall several quantitative notions of distance or similarity.

By $\rho \geq \sigma$, we mean that $\rho$ is greater than $\sigma$ in the Loewner order: $\rho - \sigma$ is a non-negative matrix. We say for a pair of superoperators $\Phi, \Psi : \D(\Hil) \rightarrow \D(\Hil)$ that $\Phi \geq_{cp} \Psi$ iff
\[ (\Phi \otimes \Id^B)(\rho) \geq (\Psi \otimes \Id^B)(\rho) \]
for all input densities $\rho$ and finite-dimensional auxiliary systems $B \cong \MM_n$ such that $n \in \NN$. We call this and the associated symbols $<_{cp}, \leq_{cp}$, and $>_{cp}$ cp-order relations. Via the Choi-Jamiolkowski isomorphism, a finite-dimensional quantum channel is fully defined by its action on one side of a maximally entangled state between its input space and an auxiliary space of the same dimension. Hence $\Phi \geq_{cp} (1-\epsilon) \Psi$ if and only if $\Phi = (1-\epsilon) \Psi + \epsilon \Theta$ for some $\epsilon \in (0,1)$ and other channel $\Theta$. Therefore, one may think of cp-order as an especially strong notion of comparison in which one process is a probabilistic combination involving the other. Physically, one could envision that if $\Phi \geq_{cp} (1-\epsilon) \Psi$, then $\Phi$ is a process that applies $\Psi$ with probability $1-\epsilon$ and some other channel with probability $\epsilon$. We also recall the Pimsner-Popa indices \cite{pimsner_entropy_1986}, which we denote for conditional expectations $\E_0$ and $\E$:
\begin{equation} \label{eq:pimsnerpopa}
\begin{split}
C(\E_0 \| \E) := & \inf \{c > 0 | \E_0(\rho) \leq c \E(\rho) \forall \rho \in \M_* \} \\ C_{cb}(\E_0 \| \E) := & \inf \{c > 0 | \E_0(\rho) \leq_{cp} c \E(\rho) \forall \rho \in \M_* \}
\end{split}
\end{equation}
as considered in \cite{gao_relative_2020, gao_complete_2022} and originally by Pimsner and Popa \cite{pimsner_entropy_1986} as a finite-dimensional analog of the Jones index \cite{jones_index_1983}. We denote $C(\E) := C(\Id \| \E)$ and $C_{cb}(\E) := C_{cb}(\Id \| \E)$. When $\E$ is a GNS self-adjoint conditional expectation with respect to density $\omega$, $C_{cb}(\E) \leq d \omega_{\text{min}}^{-1}$, where $d$ is the dimension of the space and $\omega_{\text{min}}$ the minimum eigenvalue of $\omega$. Hence for unital $\E$ in dimension $d$, $C_{cb}(\E) \leq d^2$. Via \cite[Theorem A]{gao_relative_2020},
\begin{equation} \label{eq:pprelent}
\begin{split}
    \log C(\E_0 \| \E) & = \sup_\rho \inf_\sigma D(\E_0(\rho) \| \E(\sigma)) = \sup_\rho D(\E_0(\rho) \| \E(\rho)) \pl,
\end{split}
\end{equation}
and one may easily see that the same holds under extensions of $\E_0$ and $\E$ to include a sufficiently large auxiliary system when $C(\E_0 \| \E)$ is replaced by $C_{cb}(\E_0 \| \E)$. Though the Pimsner-Popa index is traditionally defined for conditional expectations, Equation \eqref{eq:pimsnerpopa} may still be meaningful for other channels that expand the support of their inputs. 

We denote the Schatten norms $\|\cdot\|_p$ for $p \in [0,\infty]$. The trace distance is given by
\[ d_{tr}(\rho, \sigma) := \frac{1}{2} \|\rho - \sigma\|_1 \pl.\]
In general, for a pair of spaces $\A$ and $\B$ with respective norms $\|\cdot\|_\A$ and $\|\cdot\|_\B$, the $\A \rightarrow \B$ norm on maps from $\A$ to $\B$ is given by
\[ \|\Phi\|_{\A \rightarrow \B} := \sup_{\rho \in \A} \frac{\| \Phi(\rho)\|_{\B}}{\|\rho\|_\A} \pl. \]
Within finite dimension and when $\A$ and $\B$ have the same norm, we denote $\|\Phi\|_{\A \rightarrow \B, cb} := \sup_{\C} \|\Phi \otimes \Id^C \|_{\A \otimes \C \rightarrow \B \otimes \C}$, where $\C$ is an extension over finite-dimensional extensions with that norm. We denote $\|\Phi\|_{p \rightarrow q, cb} = \|\Phi\|_{\A \rightarrow \B, cb}$ in which $\| \cdot \|_\A = \|\cdot\|_p$ and $\|\cdot\|_\B = \|\cdot\|_q$. We call a map $\Phi$ an $\A \rightarrow \B$ contraction or describe it as contractive if $\|\Phi\|_{\A \rightarrow \B} \leq 1$. When $\Phi$ is a map from a normed Banach space $\A$ to itself, we denote $\|\Phi\| := \|\Phi\|_{\A \rightarrow \A}$.

The diamond norm on a map $\Phi$ defined as $\|\Phi\|_{\Diamond} := \|\Phi\|_{1 \rightarrow 1, cb}$ is a standard notion of distance between quantum channels \cite[Definition 9.1.3]{wilde_quantum_2017}. We often use diamond norm distance $\| \Phi - \Psi \|_{\Diamond}$ to compare quantum channels $\Phi$ and $\Psi$. We also characterize similarity of quantum densities using the Uhlmann fidelity given by $F(\rho, \omega) := \tr(\sqrt{\sqrt{\rho} \omega \sqrt{\rho}})^2$.

We recall the quantum relative entropy given by
\[ D(\rho \| \sigma) := \tr(\rho \log \rho - \rho \log \sigma) \pl. \]
for two densities $\rho, \sigma \in \D(\Hil)$ as introduced by Umegaki \cite{umegaki_conditional_1962}. The quantum relative entropy obeys the well-known data processing inequality \cite{lindblad_completely_1975}, that $D(\Phi(\rho) \| \Phi(\sigma)) \leq D(\rho \| \sigma)$ for every channel $\Phi$ and all input densities $\rho$ and $\sigma$. It is also bi-convex in its arguments, lower semicontinuous, and extensive under tensor products. An important property known as the chain rule is that if $\E$ and $\E_0$ are conditional expectations for which $\E \E_0 = \E_0 \E = \E$, then
\begin{equation} \label{eq:chainrule}
    D(\rho \| \E(\rho)) = D(\rho \| \E_0(\rho)) + D(\E_0(\rho) \| \E(\rho)) \pl.
\end{equation}
For (in)equlities involving linear combinations or ratios of entropies and logarithms, the logarithm base often does not matter, as the same (in)equality holds when multiplying both sides by an arbitrary constant. In these scenarios, we will denote the logarithm without specifying a base. In contexts where the base matters, we denote ``$\log_b$" with base $b$ and ``$\ln$" for the natural logarithm. Though the relative entropy is not a true distance measure, being asymmetric in its arguments, it is nonetheless often useful as a distance-like notion and has desirable properties as discussed in Subsection \ref{sec:decaybackground}.

Pinsker's inequality \cite[Theorem 10.8.1]{wilde_quantum_2017} upper bounds the trace distance in terms of relative entropy:
\begin{equation} \label{eq:pinsker}
    d_{tr}(\rho, \sigma) \leq \sqrt{\frac{\ln 2}{2} D(\rho \| \sigma)}
\end{equation}
when the relative entropy is in bits, hence with respect to the logarithm with base 2. There are several upper bounds on relative entropy in terms of the trace distance in special cases, in particular \cite[Proposition 3.7]{gao_relative_2020} and \cite[Lemma 7]{winter_tight_2016}. Here we recall a version of these bounds with a simplification given by \cite[Lemma 5.8]{bluhm_continuity_2023} and the corrections noted in \cite{bluhm_corrections_2024} based on \cite[Theorem 8]{lin_divergence_1991}:
\begin{lemma} \label{lem:wsb}
Let $\rho, \omega$ be densities and $\E$ a trace-symmetric conditional expectation. If $d_{tr}(\rho, \omega) \leq \epsilon$, and the relative entropy is defined with the logarithm in base 2, then
\begin{equation*}
\begin{split}
|D(\rho \| \E(\rho)) - D(\omega \| \E(\omega))| 
 & \leq \epsilon \log_2 C(\E) + (1+\epsilon) h_2\Big ( \frac{\epsilon}{1 + \epsilon} \Big )  
  \leq \epsilon \log_2 C(\E) + 2 \sqrt{\epsilon} \pl,
\end{split}
\end{equation*}
where $h(p) := - p \log p - (1-p) \log (1-p)$ is the binary entropy for $p \in [0,1]$, and $h_b$ specifies the logarithm base.
\end{lemma}
Furthermore, when $\E$ is a trace-symmetric conditional expectation,
\[D(\rho \| \E(\rho)) = S_{vn}(\E(\rho)) - S_{vn}(\rho) \]
for the von Neumann entropy $S_{vn} : \D(\Hil) \rightarrow \RR^+$. Many common information-theoretic quantities are special cases of relative entropy. For example, the mutual information $I(A:B)_\rho = D(\rho \| \rho^A \otimes \rho^B)$, where $\rho$ is a bipartite state on system $A \otimes B$. Other examples include the conditional entropy, coherent information, and various operational expressions for communication capacity \cite{wilde_quantum_2017}. Many operational quantities involve regularization over a large number of copies. For example, the capacityto transmit classical bits over many copies of a quantum channel $\Phi : A \rightarrow B$ is given by
\[ C_{\text{cl}}(\Phi) = \lim_{n \rightarrow \infty} \sup_{\rho \in (\X \otimes A)^{\otimes n}} I(\X^{\otimes n} : B^{\otimes n})_{(\Id^\X \otimes \Phi)(\rho)}  \pl, \]
where $\X$ denotes a classical register. That the capacity may depend on outputs being entangled across instances \cite{hastings_superadditivity_2009} motivates the use of the tensor-stable ``complete" orders and inequalities.

We also recall the Petz recovery map for channel $\Phi$ with state $\omega$,
\begin{equation} \label{eq:petzmap}
\Phi_{* \omega}(\cdot) := \omega^{1/2} \Phi_* \big (\Phi(\omega)^{-1/2} \cdot \Phi(\omega)^{-1/2} \big ) \omega^{1/2}
\end{equation}
as constructed in \cite{petz_sufficiency_1988}, where $\Phi_*$ is the adjoint under the inner product with respect to the trace as in Equation \eqref{eq:tracegns}. It was shown therein as well that $\Phi_{* \omega} \circ \Phi(\omega) = \omega$. The key property of the Petz map is that for any input density $\rho$,
\[ D(\Phi(\rho) \| \Phi(\omega)) = D(\rho \| \omega) \text{ if and only if } \Phi_{* \omega} \circ \Phi(\rho) = \rho \pl. \]

\subsection{Decay Under Quantum Channels} \label{sec:decaybackground}
There are many ways to quantify and model decay of quantum states, often closely related to the theory of decoherence \cite{clerk_introduction_2010, schlosshauer_quantum_2019}, and several known strategies to forestall it \cite{suter_colloquium_2016}. A recurring question is when quantum information decays exponentially, a general but not entirely universal scenario \cite{beau_nonexponential_2017}.

The Stinespring dilation models a quantum channel as unitary evolution on a system plus its initially uncorrelated environment, followed by a partial trace of the environment:
\begin{equation} \label{eq:stinespring}
    \Phi(\rho) = \tr_E (R_U(\rho \otimes \ketbra{0}^E)) \pl,
\end{equation}
where $E$ is the environment system, $\ketbra{0}^E$ is an arbitrary initial state (which without loss of generality, we may assume pure), and $U$ is a unitary on the joint system-environment. This model is however extremely general. For example, one may embed complex or even universal quantum computations in a non-Markovian environment \cite{aloisio_sampling_2023}. To effectively study decay properties, we usually must impose additional assumptions. 

In this work, we assume Markovianity of a time-parameterized family of quantum channels. Let $\Phi^{r,s}$ denote a channel for each $s \geq r \geq 0$ such that
\begin{equation} \label{eq:divis}
    \Phi^{s,t} \circ \Phi^{r,s} = \Phi^{r,t}
\end{equation}
whenever $t \geq s \geq r \geq 0$. Equation \eqref{eq:divis} defines the property of \textit{divisibility} in time. We assume that $\Phi^{t,t} = \Id$ for any $t \geq 0$, reflecting that time-evolution for zero time has no effect on the state. Moreover, we assume differentiability, recalling the potentially time-varying \textit{Lindbladian}
\[ \L[t] := \frac{d}{ds} \Phi^{t,t + s} \Big |_{s=0} \pl. \]
The Lindbladian decomposes into a Hamiltonian and a dissipative part,
\begin{equation} \label{eq:ldecomp}
    \L[t](\rho) = i [\Ha[t], \rho] + \S[t](\rho) \pl,
\end{equation}
where $\Ha$ is a Hamiltonian that would on its own generate unitary time-evolution. Canonical decompositions into Hamiltonian and dissipative parts exist \cite{gorini_completely_1976, hayden_canonical_2022}. We also consider discrete compositions of channels $\Phi_1 \circ ... \circ \Phi_k$. A QMS denoted $(\Phi_t)^{t \geq 0}$ further assumes the semigroup property: $\Phi^t \circ \Phi^s = \Phi^{t + s}$. A QMS generally admits a time-independent Lindbladian as derivative, such that $\Phi^t = \exp(- \L t)$. An often-studied problem is to derive a model of divisible time-evolution as a coarse-grained description of dynamics induced by general system-environment interactions as in Equation \eqref{eq:stinespring}. The Markovian approximation is usually derived for dissipative environments, in which information lost to the environment is quickly spread to eliminate backaction of that information on the system's subsequent environment \cite{merkli_dynamics_2022-1}.

In general, a dissipator as in Equation \eqref{eq:ldecomp} has the form
\[ \S(\rho) = \sum_i \gamma_i \Big (\frac{1}{2} (\S_i^* \S_i \rho + \rho \S_i^* \S_i) - \S_i \rho \S_i^* \Big ) \pl, \]
where each $\S_i$ is known as a jump operator. The coefficients $(\gamma_i)$ are positive and known as damping rates, controlling the strengths of different subprocesses. A particuarly common case we will study in several places herein is that of depolarizing noise, which replaces an input density by complete mixture. The Lindbladian for depolarizing noise, $\S_{dep}$, is given by
\begin{equation} \label{eq:depdisp}
\begin{split}
S_{dep}(\rho) & := \sum_{i=0}^3 \frac{1}{2} (X_i X_i^* \rho + \rho X_i X_i^*) -  X_i \rho X_i
\mightbenewline \mightbealign
 = \rho - \frac{\id}{d}
\end{split}
\end{equation}
in dimension $d$, where $X_0 = \hat{1}, X_1 = X, X_2 = Y$ and $X_3 = Z$ are the Pauli matrices normalized to have eigenvalues ${}^+_-1$. The depolarizing Lindbladian generates the time-parameterized depolarizing QMS,
\begin{equation} \label{eq:depsemi}
\Phi^t_{\mathrm{dep}(\lambda)}(\rho) := e^{- \lambda \S}(\rho) = e^{-\lambda t} \rho + (1 - e^{- \lambda t}) \frac{\id}{d} \pl.
\end{equation}
where $\lambda \geq 0$ controls the decay rate.

There are many ways to quantify the extent to which a quantum system has decayed. A strong notion from open systems research is the modified logarithmic Sobolev inequality (MLSI) introduced by Kastoryano and Temme \cite{kastoryano_quantum_2013}. A QMS $(\Phi^t)$ satisfies $\lambda$-MLSI if
\[ D(\Phi^t(\rho) \| \Phi^t \circ \E(\rho)) \leq e^{-\lambda t} D(\rho \| \E(\rho)) \]
for all inputs $\rho$. Classically stochastic versions of this inequality appear in earlier literature \cite{arnold_logarithmic_1998, bobkov_modified_2003}. MLSI differs from but was inspired by the earlier notion of a logarithmic Sobolev inequality \cite{gross_logarithmic_1975, gross_hypercontractivity_1975}, which states that for every $x$ in the domain of the Dirichlet form $x \mapsto \tr(x^\dagger \L x)$ intersected with the original algebra on which $\L$ is defined:
\begin{equation}
    \lambda \tr \big ( |x|^2 \ln |x|^2 - \E(|x|^2) \ln \E(|x|^2) \big ) \leq 2 \tr(x^\dagger \L x) \pl.
\end{equation}
It was more recently noted that the original log-Sobolev inequality fails for QMSs that lack a unique fixed point state \cite{bardet_hypercontractivity_2022}. In contrast the modified version remains valid \cite{bardet_estimating_2017} in the presence of more general fixed-point subspaces, generalizing from decay toward a fixed-point state to decay toward a fixed-point subspace. This observation motivates the strengthening known as CMLSI \cite{gao_fisher_2020}, which automatically incorporates a potentially untouched auxiliary system of unspecified dimension. A finite-dimensional QMS $(\Phi^t)$ generated by $\L$ with fixed-point subspace projection $\E$ has $\lambda$-CMLSI if
\begin{equation}
\begin{split}
   & D\big ( \big (\Phi^t \otimes \Id^B \big ) (\rho) \big \| \big (\Phi^t \otimes \Id^B \big ) \circ (\E \otimes \Id^B) (\rho) \big )
   \mightbenewline \mightbealign \leq e^{-\lambda t} D(\rho \| (\E \otimes \Id^B ) (\rho)) 
\end{split}
\end{equation}
for all extensions by a finite-dimensional auxiliary subsystem $B$. A discrete analog to CMLSI is the complete, strong data processing inequality (CSDPI), which for decay to a fixed point takes the form:
\[ D((\Phi \otimes \Id^B)(\rho) \| ((\Phi \circ \E) \otimes \Id^B)(\rho))
    \leq \lambda D(\rho \| (\E \otimes \Id^B)(\rho)) \]
for $\lambda \in [0,1)$. As some particularly useful properties of CMLSI:
\begin{itemize}
    \item For a time-varying Lindbladian and associated divisible family of channels with the same fixed point projection, divisibility also applies to estimated decay rates. Assume that for each $t \geq 0$, a time-parameterized Lindbladian $\L[t]$ generates a QMS $s \mapsto \exp(-\L[t] s)$ with time-independent fixed-point subspace projection $\E$ and $\lambda[t]$-(C)MLSI. Then
    \begin{equation} \label{eq:divisdecay}
    \begin{split}
        D(\Phi^{s,t}(\rho) \| \Phi^{s,t} \circ \E(\rho))
        \leq e^{ - \int_s^t \lambda[r] d r } D(\rho \| \E(\rho))
    \end{split}
    \end{equation}
    for all input densities $\rho$. Relatedly, (C)MLSI implies decay from arbitrarily small initial relative entropy.
    \item CMLSI and CSDPI are tensor-stable: if $(\Phi^t)$ and $(\Psi^t)$ are QMSs with $\lambda_1$-CMLSI and $\lambda_2$-CMLSI respectively, then $(\Phi^t \otimes \Psi^t)$ has $\min\{\lambda_1, \lambda_2\}$-CMLSI.
\end{itemize}
Via Pinsker's inequality as recalled in Equation \eqref{eq:pinsker}, decay under CMLSI at sufficiently long times upper-bounds the diamond-norm distance between a channel and its fixed-point subspace projection. Tensor-stability enables CMLSI to upper-bound regularized capacities \cite{bardet_group_2021}, which are defined as optimizing over potentially entangled inputs between many copies of a quantum channel. CMLSI and CSDPI can also upper bound extends of quantum advantage in certain settings \cite{stilck_franca_limitations_2021}.

CMLSI is known to hold with some constant for all finite-dimensional QMSs having GNS detailed balance with respect to a faithful state and a uniquely defined fixed-point subspace conditional expectation $\E$. This result was derived in a self-contained way in \cite{gao_complete_2025}. It was simultaneously derived in \cite{gao_geometric_2021} for QMSs that are self-adjoint with respect to the trace, which via \cite{junge_stability_2022} also extends to all finite-dimensional QMSs with detailed balance. Furthermore, it was shown in \cite{gao_complete_2025} that CSDPI holds for quantum channels with GNS detailed balance and some additional, commonly satisfied assumptions.

Follow-up works have shown applications in the theory of many-body systems \cite{bardet_entropy_2024,bardet_rapid_2023, chen_fast_2023}. The tensor-stability and divisibility of CMLSI, as well as some technical properties of the relative entropy \cite{laracuente_quasi-factorization_2022}, allowed many of these works to build up optimal decay rate bounds on whole systems from the existence of CMLSI constants for single-site constituents.

A primary limitation of known CMLSI results is that detailed balance is usually incompatible with the presence of a non-trivial Hamiltonian term as in Equation \eqref{eq:ldecomp}. For example, we may consider detailed balance with respect to complete mixture, which reduces to self-adjointness under the trace. Even if the dissipator $\S$ is self-adjoint, we still find that
\[ \L_*(\rho) = - i [\Ha, \rho] + \S(\rho) \pl, \]
reversing the sign of the Hamiltonian term compared with $\L$ as in Equation \eqref{eq:ldecomp} (physically, the unitary part of the process is time-reversed). In general, $- \Ha \neq \Ha$ and generates distinct continuous-time evolution. Motivating examples of where CMLSI estimates fail are presented in Section \ref{sec:examples}.

\section{Main Results and Proofs} \label{sec:theory}
For a channel $\Phi$ in finite dimension defined on the predual of a von Neumann algebra $\M$, the multiplicative domain (properly of $\Phi^*$, its adjoint with respect to the trace) is
\begin{equation*}
\begin{split}
\N_\Phi := \{ & x \in \M : \Phi^*(y) \Phi^*(x) = \Phi^*(y x) ,
 \mightbenewline \mightbealign
 \text{ and } \Phi^*(x) \Phi^*(y) = \Phi^*(x y) \text{ } \forall y \in \M \} \pl.
\end{split}
\end{equation*}
The multiplicative domain $\N_\Phi$ is a subalgebra, and as such, it comes with a unique projection $\E_{\Phi} : \M \rightarrow \N_\Phi$, which is a conditional expectation. Multiplicative domains have been previously applied in the theory of quantum error correction for unital quantum channels, in which they often act as decoherence-free subspaces. Leveraging this connection and \cite{kribs_quantum_2006}, results of \cite{choi_multiplicative_2009} (later generalized in \cite{johnston_generalized_2011}) imply that the multiplicative domain $\N_{\Phi}$ coincides with the fixed point subspace of $\Phi \circ \Phi^*$ \cite[Theorem 1.3]{rahaman_multiplicative_2017} when $\Phi$ is unital. Other results have specifically characterized the decoherence-free subspaces of channels and Lindbladians \cite{lidar_decoherence-free_1998, agredo_decoherence_2014, lidar_review_2014, carbone_period_2020}, on which the Lindbladian acts unitarily. Before directly addressing whether decay is exponential, we first use the analytical structure of the semigroup to infer that decay to the decoherence-free subalgebra begins immediately - a unital QMS never admits initially non-decaying states that are not contained in its decoherence-free subspace.

\begin{lemma} \label{lem:analyt}
    Let $(\Phi^t)_{t \geq 0}$ be a QMS with generator bounded in diamond norm and $\rho$ be an input density.
    Assume that $(\Phi^{t})_{t \geq 0}$ has (not necessarily unique) full rank fixed point density $\omega$. If there exists a $\tau > 0$ for which $D(\Phi^{\tau}(\rho) \| \Phi^{\tau}(\omega)) = D(\rho \| \omega)$ on density $\rho$, then for all $t \geq 0$, $\Phi^{t}_{* \omega} \Phi^{t}(\rho) = \rho$, and $D(\Phi^{t}(\rho) \| \Phi^{t}(\omega)) = D(\rho \| \omega)$.
\end{lemma}
\begin{proof}
For all $s > 0$ and every bounded generator $\L$,
\[ \Phi^s_* \circ \Phi^s (\rho)
    = e^{- s \L_*} e^{- s \L} (\rho)
    = \sum_{j,k = 0}^\infty s^{j+k} \frac{(-1)^{j + k}}{j! k!} \L^j_* \circ \L^k (\rho) \]
is manifestly analytic in $s$ within finite dimension and equal to its Taylor series around $s=0$, even after extending its domain from $s \in [0,\infty)$ to the complex plane. Because $\Phi^s_* \circ \Phi^s(\rho)$ is constant on the interval $[0, \tau)$, it must hold that
\[ \frac{d^{(k)}}{ds^{(k)}}
    \big ( \Phi^s_* \circ \Phi^s(\rho) \big ) \Big |_{s=0} = 0 \]
for all $k \in \NN$. Hence $\Phi^s_* \circ \Phi^s(\rho) = \rho$ is constant for all $s \geq 0$. Via the data processing inequality of relative entropy,
\[ D(\rho \| \omega) = D(\Phi^{t}(\rho) \| \Phi^{t}(\omega)) = D(\Phi^{\tau}(\rho) \| \Phi^{\tau}(\omega))\]
for all $t \in [0,\tau]$. Then by the properties of the Petz recovery map as in Equation \eqref{eq:petzmap}, $\Phi^{t}_{*\omega} \Phi^{t}(\rho) = \rho$ for all $t \in [0,\tau]$. Since $\Phi^{t}_{*\omega} \circ \Phi^{t}$ is constant for all $t \in (0, \tau)$, $\Phi_{*\omega}^{t} \circ \Phi^{t}(\rho) = \rho$ for all $t > \tau$.
\end{proof}
\begin{lemma} \label{lem:multdom}
Let $(\Phi^{t})_{t \geq 0}$ be a unital QMS with diamond norm-bounded generator. Let $\E_{\Phi^\tau}$ denote the multiplicative domain projection of $\Phi^{\tau}$. For any $\tau > 0$, $\Phi^{t}_* \circ \Phi^{t} \circ \E_{\Phi^\tau} = \E_{\Phi^\tau}  \circ \Phi^{t}_* \circ \Phi^{t} = \E_{\Phi^\tau}$ for all $t \geq 0$.
\end{lemma}
\begin{proof}
Since $\E_{\Phi^\tau}$ is unital, $\id / d = \E_{\Phi^\tau}(\id / d)$. Recall as noted in \cite{agredo_decoherence_2014, gao_complete_2025} that $\Phi^{\tau}$ is an isometry on its multiplicative domain, so
\begin{equation*}
\begin{split}
D(\E_{\Phi^\tau}(\rho) \| \id / d) & = D(\Phi^{\tau}(\E_{\Phi^\tau}(\rho)) \| \Phi^{\tau}(\id / d) ) \pl.
\end{split}
\end{equation*}
Since $\Phi^{\tau}$ does not decrease relative entropy between $\E_{\Phi^\tau}(\rho))$ and $\id / d$, it is inverted by its Petz recovery map \cite{petz_sufficiency_1988} with respect to $\id/d$ on this subspace, which is equal to $\Phi^{\tau}_*$, its (pre-)adjoint under the trace. Hence $\Phi^{\tau}_* \circ \Phi^{\tau}(\rho) = \rho$. Applying Lemma \ref{lem:analyt}, $\Phi^{t}_* \circ \Phi^{t} \circ \E_{\Phi^\tau}(\rho) = \E_{\Phi^\tau}(\rho)$ for any $t \geq 0$ and input state $\rho$. Since $\E_{\Phi^\tau}$ is self-adjoint with respect to the trace, $(\Phi^{t}_* \circ \Phi^{t} \circ \E_{\Phi^\tau})_* = \E_{\Phi^\tau} \circ \Phi_*^{t} \circ \Phi^{t} = \E_{\Phi^\tau}$.
\end{proof}
\begin{lemma} \label{lem:rotfix}
    Let $(\Phi^t)$ be a unital QMS, either for discrete $t \in \NN$ or continuous $t \in \RR^+$. Assume there exists a conditional expectation $\E$ such that $\Phi^t \E = \E \Phi^t = R^t \E = \E R^t$ for time-dependent unitary rotation $R^t$ and that $\Phi^t$ approaches $\Phi^t \E$ in diamond norm for sufficiently large $t$. Then for any $\tau > 0$, $\E = \E_{\Phi^\tau}$, the multiplicative domain projection of $\Phi^\tau$.
\end{lemma} 
This Lemma follows from results of \cite{kribs_quantum_2006} combined with \cite{lidar_decoherence-free_1998}. For completeness and simplicity, we present an elementary proof here.
\begin{proof}
We may easily extend rotations to negative $t$ by identifying $R^{-t} = R^t_*$, the adjoint map under the trace. Then by our original assumptions, $\| \Phi^t - R^t \E\|_{\Diamond} \leq \epsilon$ for arbitrarily small $\epsilon$ and sufficiently large $t$. Furthermore,
\[ \Phi^t_* \E = (\E \Phi^t)_* = (\Phi^t \E)_* = (R^t \E)_* = (\E R^t)_* = R^{-t} \E \pl. \]
Hence using contractiveness of the diamond norm under quantum channels, $\|\Phi^t_* \Phi^t - \E\|_{\Diamond} \leq \epsilon$. Using Lemma \ref{lem:multdom},
\[ \E \E_{\Phi^\tau} = \lim_{t \rightarrow \infty} (\Phi^t_* \Phi^t) \E_{\Phi^\tau} = 
     \E_{\Phi^\tau} \pl. \]
By self-adjointness of $\E_{\Phi^\tau}$, we also find immediately that $\E_{\Phi^\tau} \E = \E_{\Phi^\tau}$. This part of the proof shows that the subspace projected to by $\E_{\Phi^\tau}$ is contained in that projected to by $\E$.

To complete the proof, we now must show that $\E_{\Phi^\tau} \E = \E$. Let $U^t$ be defined such that $R^t(\cdot) = U^t \cdot U^t_\dagger$. Then for any operators $x, y$ and all $t \geq 0$
\begin{equation*}
\begin{split}
    & \Phi^t (\E(x)) \Phi^t(y) = U^t \E(x) U^t_\dagger \Phi^t(y) 
    \mightbenewline \mightbealign = U^t( \E(x) U^t_\dagger \Phi^t(y) U^t ) U^t_\dagger
        = U^t( \E(x) R^{-t} \Phi^t(y) ) U^t_\dagger \pl.
\end{split}
\end{equation*}
By the assumption that $\Phi^t \E = R^t \E$ and because $R^{-t}$ was defined equal to $R^t_*$, $R^{-t} \Phi^t \E = \E$. Let $R^{-t} \Phi^t$ have the Kraus representation
\[ R^{-t} \Phi^t(\omega) = \sum_i K_i \omega K_i^\dagger \] 
for all inputs $\omega$. Via \cite[Theorem 4.25]{watrous_theory_2018} and for any input $x$, $\E(x)$ being in the fixed point subspace of the unital channel $R^{-t} \Phi^t$ implies that $[\E(x), K_i] = 0$ for every $i$. Since $(R^{-t} \Phi^t)_* = \Phi^t_* R^t$ has the same Kraus decomposition up to the exchange $K_i \leftrightarrow K_i^\dagger$ and also leaves $\E(x)$ invariant, $[\E(x), K_i^\dagger] = 0$ for every $i$ as well. Commutation with the Kraus operators shows that $R^{-t} \Phi^t$ is a bimodule map for the space projected to by $\E$, hence $\E(x) R^{-t} \Phi^t(y) = R^{-t} \Phi^t(\E(x) y)$. After cancelling unitary rotations, we find that $\E$ outputs to a subspace of the multiplicative domain of $\Phi^\tau$, completing the Lemma.
\end{proof}
\begin{lemma} \label{lem:sameE}
    Let $\Phi$ be a unital quantum channel such that for any $\epsilon > 0$, there exists a sufficiently large $n \in \NN$ that for all $k \geq n$, $\|  (\Phi_* \Phi)^k - \E \|_{\Diamond} \leq \epsilon$, and $ \| (\Phi \Phi_*)^k - \tilde{\E} \|_{\Diamond} \leq \epsilon$ for projections $\E$ and $\tilde{\E}$. Further assume that $\| \Phi - \Id\|_{\Diamond} \leq \delta$. Then $\| \E - \tilde{\E}\|_{\Diamond} \leq 2 \epsilon + 2 \delta$.
\end{lemma}
\begin{proof}
The channel family $(\Phi_* \Phi)^m$ is a (discrete-time) semigroup in $m \in \NN$. Via \cite[Theorem 6.7]{junge_noncommutative_2006}, since $\Phi^{t}_* \Phi^{t}$ is trace-symmetric, $\lim_{m \rightarrow \infty} (\Phi_* \Phi)^m = \E_{(t)}$ for some fixed point projection $\E_{(t)}$ that a priori might depend on $t$. Within finite dimension, this convergence can be taken in the diamond norm. The same holds for $(\Phi \Phi_*)^m$.

We observe that
\begin{equation} \label{eq:phiexchange}
    \Phi (\Phi_* \Phi)^m = (\Phi \Phi_*)^m \Phi \pl.
\end{equation}
Since quantum channels are diamond norm contractions, upper bounds on the diamond norm are stable when concatenated with some other channel applications. By the assumption that $\| \Phi - \Id\|_{\Diamond} \leq \delta$ and by the monotonicity of the diamond norm under pre- or post-application of quantum channels,
\begin{equation*}
\begin{split}
 & \| \Phi (\Phi_* \Phi)^m  - (\Phi_* \Phi)^m \|_{\Diamond} \leq \delta \text{, and }
    \mightbenewline \mightbealign
    \| (\Phi \Phi_*)^m \Phi  - (\Phi \Phi_*)^m \|_{\Diamond} \leq \delta  \pl,
\end{split}
\end{equation*}
so
\[ \| (\Phi_* \Phi)^m - (\Phi \Phi_*)^m \|_{\Diamond} \leq 2 \delta \pl. \]
Combining via the triangle inequality with the assumed convergences completes the Lemma.
\end{proof}

\begin{lemma} \label{lem:converge1}
    Let $(\Phi^{t})$ be a finite-dimensional, unital QMS. Then there exists a trace-symmetric conditional expectation $\E$ such that for any $\epsilon > 0$ and $t > 0$, there is a sufficiently large $n \in \NN$ that for every $m \geq n$,
    \begin{equation*}
    \begin{split}
    & \|  (\Phi^{t}_* \Phi^{t})^m - \E \|_{\Diamond} \leq \epsilon \pl, \pl \| (\Phi^{t} \Phi^{t}_*)^m - \E \|_{\Diamond} \leq \epsilon ,
    \mightbenewline \mightbealign
    \text{ and } \| [\Phi^{t}, (\Phi^{t}_* \Phi^{t})^m] \|_{\Diamond} \leq \epsilon \pl.
    \end{split}
    \end{equation*}
\end{lemma}
\begin{proof}
The first two inequalities simply apply Lemma \ref{lem:sameE} for channels corresponding to semigroups at fixed time.

We easily see that $\E_{(t)}$ is self-adjoint with respect to the trace, so it is a conditional expectation. Using the monotonicity of diamond norm under channel application to both arguments and that $\E_{(t)}$ is a fixed point projection, we obtain the existence of such $n$ for any desired $\epsilon$ in the first inequality. It still remains to show that $\E_{(t)}$ does not actually depend on $t$. By the data processing inequality for relative entropy,
\begin{equation*}
\begin{split}
& D(\Phi^{t}_* \Phi^{t} \E_{(t)}(\rho) \| \id / d) \leq D( \Phi^{t} \E_{(t)}(\rho) \| \id / d)
\mightbenewline \mightbealign \leq D(\E_{(t)}(\rho) \| \id/d)
 = D(\Phi^{t}_* \Phi^{t} \E_{(t)}(\rho) \| \id / d)  \pl,
\end{split}
\end{equation*}
so all of these expressions are equal. Using Lemma \ref{lem:analyt}, it holds for every $s \geq 0$ that $\Phi_*^s \Phi^s \E_{(t)}(\rho) = \E_{(t)}(\rho)$. Since $\E_{(s)} = \lim_{m \rightarrow \infty} (\Phi^s_* \Phi^s)^m$ for every $s$, $\E_{(s)} \E_{(t)} = \E_{(t)}$. Since the projections are self-adjoint, $\E_{(t)} \E_{(s)} = \E_{(t)}$ too. Switching the roles of $t$ and $s$, $\E_{(t)} = \E_{(s)}$. This equality holds for every $t, s > 0$, so all of the fixed point projections are the same map. We henceforth drop the subscript and call this projection $\E$.

To obtain the 2nd inequality, we set $t$ sufficiently small to apply Lemma \ref{lem:sameE} for sufficiently small $\delta$ and use the first inequality to obtain arbitrarily small $\epsilon$ for that Lemma's assumptions.


Recalling Equation \ref{eq:phiexchange} with $\| (\Phi^{t} \Phi^{t}_*)^m - (\Phi^{t}_* \Phi^{t})^m \|_{\Diamond} \leq 2 \epsilon$ recovers the final commutator bound for all $t \geq 0$. To go from $2 \epsilon$ to $\epsilon$, we can simply choose the original $\epsilon$ smaller.
\end{proof}
Combining previous Lemmas, the multiplicative domain projection universally acts as a fixed point projection up to rotation for unital channel families with the appropriate series expansion:
\begin{lemma} \label{lem:converge2}
Let $(\Phi^{t})_{t \geq 0}$ be a finite-dimensional, unital QMS. Then there is a unique, trace-symmetric conditional expectation $\E$ projecting to the multiplicative domain of $\Phi^{t}$ for every $t > 0$. Moreover, there exists a family of unitary rotations $(R^{t})_{t \geq 0}$ for which $\Phi^{t} \E = \E \Phi^{t} = R^{t} \E = \E R^{t}$.
\end{lemma}
\begin{proof}
 First let $\E = \lim_{m \rightarrow \infty} (\Phi^{\tau}_* \Phi^{\tau})^m$ for some fixed $\tau > 0$. Lemma \ref{lem:converge1} confirms that $\E$ is a unital conditional expectation that does not depend on the choice of $\tau$. Moreoever, Lemma \ref{lem:converge1} implies that $\Phi^t \E = \E \Phi^t$ for every $t \geq 0$. Hence by Lemma \ref{lem:rotfix}, $\E$ is the multiplicative domain projection for channel $\Phi^t$ independently of $t$ for all $t > 0$. Using that $\Phi^{t} \E = \Phi^{t} \circ \lim_{m \rightarrow \infty} (\Phi^{t}_* \Phi^{t})^m$, $\Phi^{t}_* \Phi^{t} \E = \E$.
 It was shown in \cite[Proposition 2.1]{dhahri_decoherence-free_2010} and \cite[Proposition 2]{fagnola_role_2019} that $\Phi^t \circ \E(\rho) = e^{-i \Ha t} \E(\rho) e^{i \Ha t}$, assuming that the Lindbladian is written with a Hamiltonian $\Ha$ and traceless, orthonormal jump operators as in Equation \eqref{eq:lindform}. Hence $R^t = e^{- i \Ha t}$.
 We identify $R^{-t} := R^{t}_*$. Since $\E$ is self-adjoint under the trace, $\E R^t = (R^{-t} \E)_*$. We then apply the same arguments for $(\Phi^{t}_*)$, using the same $\E$ via Lemma \ref{lem:converge1} and obtaining that $\Phi^{t}_* \E = \tilde{R}^{t} \E$ for some $\tilde{R}^{t}$. Furthermore,
\begin{equation*}
\begin{split}
& \Phi^{t}_* \Phi^{t} \E = \Phi^{t}_* \E \Phi^{t} = \tilde{R}^{t} \E \Phi^{t}
    \mightbenewline \mightbealign = \tilde{R}^{t} \Phi^{t} \E = \tilde{R}^{t} R^{t} \E = \E \pl.
\end{split}
\end{equation*}
Hence $\tilde{R}^{t} \E = R^{t}_* \E$. This observation yields the expected commutation of $R^{t}$ and $\E$:
\[ \E R^{t} = (R^{-t} \E)_* = (\Phi^{t}_* \E)_* = \E \Phi^{t} = \Phi^{t} \E = R^{t} \E \pl. \]
We thereby conclude that
\begin{equation*}
    \Phi^{t} \E = \E \Phi^{t} = R^{t} \E = \E R^{t}
\end{equation*}
for each value of $t$. Recalling Lemma \ref{lem:rotfix} with the discrete semigroup given by $m \rightarrow (\Phi^{t}_* \Phi^{t})^m$, $\E$ is the multiplicative domain projection of $\Phi^{t}$.
\end{proof}

\subsection{Adding a Hamiltonian to a Dissipator that has (C)MLSI}
We delineate the regimes in which CMLSI appears and note a somewhat counter-intuitive relationship between the decay rate induced by $\S$ and that of $\L$ for QMSs in the form of Equation \eqref{eq:ldecomp}.

We first address the case of a Lindbladian for which the dissipative generator $\S$ would on its own induce decay to a fixed point subspace with projection $\E_0$. We show that general CMLSI still holds when the Hamiltonian commutes with $\E_0$ so that $\E = \E_0$:
\begin{lemma} \label{lem:commute}
Let $\L$ be a Lindbladian with fixed point conditional expectation $\E$. Let $\tilde{\L}$ be a Lindbladian generating time-evolution that commutes with $\E$, $\Psi_1, ..., \Psi_m$ be quantum channels commuting with $\E$, $\Phi^{s}$ be the channel generated by $\L + \tilde{\L}$ from time $0$ to $s$, and
\[ \Theta := \Psi_1 \circ \Phi^{t_2 - t_1} \circ \Psi_2 \circ \Phi^{t_3 - t_2} \circ ... \Psi_{m-1} \circ \Phi^{t_{m} - t_{m-1}} \circ \Psi_m \pl. \]
If $\L$ has $\lambda$-(C)MLSI, then $\L + \tilde{\L}$ has $\lambda$-CMLSI. Moreover,
\begin{equation*}
\begin{split}
    & D(\Theta (\rho) \| \E \circ \Theta (\rho)) \leq e^{- t \lambda} D(\rho \| \E(\rho))
\end{split}
\end{equation*}
for any $t_1, ..., t_m > 0$ such that $t_1 + ... + t_m = t$.
\end{lemma}
\begin{proof}
First, assume that $\tilde{\L} = 0$. For each $j \in 1...m$, let
\[ \Theta_j := \Psi_j \Phi^{t_{j+1} - t_j} ... \Phi^{t_m - t_{m-1}} \Psi_m \pl. \]
Via the data processing inequality and commutation of $\Psi_j$ and $\Phi^{t_{j+1} - t_j}$ with $\E$,
\begin{equation*}
\begin{split}
    &D(\Theta_j (\rho) \| \E \circ \Theta_j(\rho))
	 \leq D(\Phi^{t_{j+1} - t_j} \Theta_{j+1} (\rho) \| \E(\Theta_{j+1} (\rho))) \pl.
\end{split}
\end{equation*}
Then using assumed (C)MLSI,
\begin{equation*}
\begin{split}
&D(\Phi^{t_{j+1} - t_j} \Theta_{j+1} (\rho) \| \E(\Theta_{j+1} (\rho)))
\mightbenewline \mightbealign \leq e^{-(t_{j+1} - t_j) \lambda} D(\Theta_{j+1} (\rho) \| \E(\Theta_{j+1} (\rho))) \pl.
\end{split}
\end{equation*}
Iterating the inequality from $j=1$ to $j=m$ completes this case.

We apply the Suzuki-Trotter expansion. We subdivide the time interval $[t_1, t_m]$ into smaller intervals of length $(t_m - t_1) / k$, where $k$ is chosen to be an integer multiple of $m$. Let $s_j$ label the start of each such interval for each $j \in 1...k-1$, and $s_k = s_m$. Let $\tilde{\Psi}_j$ be equal to $\Psi_{j / (k/m)}$ when $j m / k$ is an integer $\leq k$, and the identity map otherwise. Finally, let
\[ \tilde{\Theta}_j := \tilde{\Psi}_j \Phi^{s_{j+1} - s_j} ... \tilde{\Psi}_{k-1} \Phi^{s_{k} - s_{k-1}} \tilde{\Psi_k} \pl. \]
Via the data processing inequality of relative entropy and assumed commutation of $\E$ with $\tilde{\Psi}_j$
\begin{equation} \label{eq:smallintdec}
\begin{split}
& D(\tilde{\Theta}_j (\rho) \| \E (\tilde{\Theta}_j(\rho))) 
\mightbenewline \mightbealign \pl \leq D( \Phi^{s_{j+1} - s_j} \tilde{\Theta}_{j+1} (\rho) \| \E (\Phi^{s_{j+1} - s_j} \tilde{\Theta}_{j+1}(\rho)))
\end{split}
\end{equation}
Using the Kato-Suzuki-Trotter expansion,
\[ \| \Phi^{s_{j+1} - s_j} - \Phi^{s_{j+1} - s_j}_\L \Phi^{s_{j+1} - s_j}_{\tilde{\L}} \|_{\Diamond} \leq O(1/k^2) \pl, \]
where $\Phi^{s_{j+1} - s_j}_\L$ and $\Phi^{s_{j+1} - s_j}_{\tilde{\L}}$ denote respective time-evolutions generated by $\L$ and $\tilde{\L}$. In particular,
\[ \E(\Phi^{s_{j+1} - s_j}) = \Phi^{s_{j+1} - s_j}_{\tilde{\L}} \E = \E \Phi^{s_{j+1} - s_j}_{\tilde{\L}} \pl, \]
which we obtain by further subdividing the time interval $[s_j, s_{j+1}]$ into infinitesimal intervals, achieving equality in the limit. Therefore,
\begin{equation*}
\begin{split}
&  D( \Phi^{s_{j+1} - s_j} \tilde{\Theta}_{j+1} (\rho) \| \E (\Phi^{s_{j+1} - s_j} \tilde{\Theta}_{j+1}(\rho)))
\mightbenewline \mightbealign
 = D( \Phi^{s_{j+1} - s_j} \tilde{\Theta}_{j+1} (\rho) \| \Phi^{s_{j+1} - s_j}_{\tilde{\L}} \E ( \tilde{\Theta}_{j+1}(\rho))) \pl.
\end{split}
\end{equation*}
Applying the Kato-Suzuki-Trotter formula with Lemma \ref{lem:wsb}, then data processing,
\begin{equation*}
\begin{split}
 & D( \Phi^{s_{j+1} - s_j} \tilde{\Theta}_{j+1} (\rho) \| \Phi^{s_{j+1} - s_j}_{\tilde{\L}} \E ( \tilde{\Theta}_{j+1}(\rho)))
\mightbenewline \mightbealign
 \leq D( \Phi^{s_{j+1} - s_j}_\L \tilde{\Theta}_{j+1} (\rho) \|  \E ( \tilde{\Theta}_{j+1}(\rho))) + O \Big ( \frac{\log k}{k^2} \Big ) \pl,
\end{split}
\end{equation*}
Now using assumed (C)MLSI with the previous two Equations and Equation \eqref{eq:smallintdec},
\begin{equation*}
\begin{split}
& D(\tilde{\Theta}_j (\rho) \| \E (\tilde{\Theta}_j(\rho))) 
\mightbenewline \mightbealign
\leq e^{- (s_{j+1} - s_j) \lambda} D(\tilde{\Theta}_{j+1} (\rho) \| \E (\tilde{\Theta}_{j+1}(\rho))) + O \Big ( \frac{\log k}{k^2} \Big ) \pl.
\end{split}
\end{equation*}
Iterating the inequality while taking the limit as $k \rightarrow \infty$ completes the proof.

\end{proof}
To handle continuous-time dynamics, we invoke Trotter decompositions for small times:
\begin{rem} \label{rem:bimod}
When $\E$ is a conditional expectation, and $\L = i [\Ha, \cdot]$ for some Hamiltonian $\Ha$, the bimodule property of conditional expectations implies that for any input $\rho$,
\begin{equation}
\E(i[\Ha, \E(\rho)]) = i \E(\Ha \E(\rho)) - i \E(\E(\rho) \Ha) = i [\E(\Ha), \E(\rho)] \pl.
\end{equation}
\end{rem}
\begin{lemma} \label{lem:hamlnorm}
For any $p$ such that $\|\cdot\|_{p \rightarrow q}$ is a norm and $\|\cdot\|_p$ obeys H\"older's inequality,
\[ \|[\Ha, \cdot]^m \|_{p \rightarrow q, (cb)} \leq 2^m \| \Ha \|_{\infty}^m \sup_\rho \|\rho\|_q / \|\rho\|_p \]
\end{lemma}
\begin{proof}
Recall that $[\Ha, \cdot]^m(\rho)$ generates $2^m$ terms on any density $\rho$, each of which contains $m$ powers of $\Ha$ and one of $\rho$. Using H\"older's inequality and its inductive generalization,
\[ \| \Ha^k \rho \Ha^{m-k} \|_q \leq \| \Ha^k \|_{\infty} \| \rho \Ha^{m-k} \|_q
	\leq \| \Ha^k \|_{\infty} \| \| \rho \|_q \| \Ha^{m-k} \|_\infty
		\leq \| \rho \|_q \| \Ha \|_{\infty}^m \]
for any integer $k$ such that $0 \leq k \leq m$. Hence
\[ \|[\Ha, \cdot]^m\|_{p \rightarrow p} =
	\sup_\rho \|[\Ha, \cdot]^m(\rho)\|_q / \|\rho\| \leq
	2^m \| \Ha \|_{\infty}^m \sup_\rho \|\rho\|_q / \|\rho\|_p \pl. \]
To see that this holds for the cb-norm, note that $\|\Ha \otimes \id\|_\infty = \|\Ha\|_\infty$. Therefore, one may apply the same result with arbitrary extensions and find the same factors.
\end{proof}

\begin{prop}[Technical Version of Prop \ref{prop:main1intro}] \label{prop:decayhaml}
Let $\L(\rho) := i [\Ha, \rho] + \S(\rho)$ be a Lindbladian generating unital, finite-dimensional QMS $(\Phi^t)$. Then there exists a unique conditional expectation $\E$ such that $\Phi^t \E = \E \Phi^t = R^t \E = \E R^t$, where $R^t = \exp(\cdot \rightarrow - i [\Ha, \cdot])$. Moreover, assume unital $\Phi_0^t(\rho) := \exp(-t \S)(\rho)$ with $\lambda_0$-(C)MLSI and fixed point conditional expectation $\E_0$.
\begin{itemize}
    \item If evolution via $\Ha$ commutes with $\E_0$, then $\E = \E_0$, and $\Phi^t$ has $\lambda_0$-(C)MLSI.
    \item If $\Ha$ does not commute with $\E_0$, for every input $\rho$ such that $D(\E_0(\rho) \| \E(\rho)) > 0$,
    \[ D(\E_0(\rho) \| \E(\rho)) - D(\Phi^t \E_0((\rho)) \| \E(\Phi^t(\E_0(\rho))))
        \leq c_t t^2 \log C(\E) + (1 + c_t t^2) h \Big (\frac{c_t t^2}{1+c_t t^2} \Big ) = O(t^2)  \]
    for all $t \geq 0$, where
    \[ c_t := \frac{e}{2} \big ( \|\L\|_{\Diamond}^2 e^{t \|\L\|_{\Diamond}} + 4 \|\Ha\|_\infty^2 e^{2 t \|\Ha\|_{\infty}} \big ) \pl, \]
    precluding (C)MLSI or similar decay at small times.
\end{itemize}
\end{prop}
\begin{proof}
When conjugation by $\Ha$ commutes with $\E_0$, the transference of CMLSI from $\S$ to $\L = [\Ha, \cdot] + \S$ follows \cite[Proposition 1.6]{laracuente_quasi-factorization_2022}, setting $\L_0$ therein equal to $- i [\Ha, \cdot]$. Lemma \ref{lem:commute} is a generalization thereof. This proves the first point.

We now consider the second point. Whether or not conjugation by $\Ha$ commutes with $\E_0$, Lemma \ref{lem:converge2} implies that there exists a decoherence-free subspace projection $\E$ and persistent rotation $t \mapsto R^t$ such that $\Phi^t \E  = \E \Phi^t = R^t \E = \E R^t$. Observe hat
\[ \Phi_0^t(\rho) = \rho - t \S(\rho) 
    + \sum_{n=2}^\infty \frac{t^n (-\S)^n}{n!}(\rho) \pl. \]
Explicitly noted as Remark \ref{rem:expapprox} and Lemma \ref{lem:exphaml} in the appendix,
\[ \| \Phi_0^t - (\rho - t \S) \|_{\Diamond} \leq \frac{t^{2} \|\S\|_{\Diamond}^{2} \exp(t \| \S \|_{\Diamond})}{2} \pl.\]
Were $\S \circ \E_0 \neq 0$ or $\E_0 \circ \S \neq 0$, the above would contradict that $\E_0 \Phi_0^t = \Phi_0^t \E_0 = \E_0$ for extremely small $t$. Therefore, $\E_0 \circ \S = \S \circ \E_0$, and
\[ \E_0 (\Phi^t(\rho)) - \Phi^t(\E_0(\rho)) = i t ([\Ha, \E_0(\rho)] - \E_0([\Ha, \rho])) + 
    \Big [ \E_0, \sum_{n=2}^\infty \frac{t^n (-\L)^n}{n!} \Big ](\rho) \pl, \]
which is $i \times t$ times the commutator of $\cdot \rightarrow [\Ha, \cdot]$ with $\E_0$ plus subleading corrections. Since $\Phi^t \E  = \E \Phi^t$ for arbitrarily small $t$, if $\E = \E_0$, then $\cdot \rightarrow [\Ha, \cdot]$ must commute with $\E_0$. Hence by contrapositive, if $\cdot \rightarrow [\Ha, \cdot]$ does not commute with $\E_0$, then $\E \neq \E_0$.

Since $\E = \E_0 \E$, and $\S \circ \E_0 = 0$,
\[ \exp(- t \L)(\E(\rho)) = \E(\rho) - i t [\Ha, \E(\rho)] + \sum_{n=2}^\infty \frac{t^n }{n!} (-\L)^n \circ \E(\rho) \pl.\]
Again via series expansions,
\[ \| \exp(- t \L) \circ \E - (\E - (\cdot \mapsto i t [\Ha, \E(\cdot))]) \|_{\Diamond}
    \leq \frac{t^{2} \|\L\|_{\Diamond}^{2} \exp(t \| \L \|_{\Diamond})}{2} \pl. \]
Since it holds infinitesimally to linear order, $\exp(- t \L) \circ \E = R^{- i t \Ha}$, where $R^{- i t \Ha}$ is the rotation generated by conjugation with $\exp(-i t \Ha)$. For any input state $\rho$, again using Taylor expansions,
\begin{equation*}
    \Phi^t \circ \E_0(\rho) = \E_0(\rho) - i [\Ha, \E_0(\rho)] - \S \circ \E_0(\rho) + \sum_{n=2}^\infty \frac{t^n}{n} (-\L)^n \circ \E_0(\rho) \pl.
\end{equation*}
Again recalling Remark \ref{rem:expapprox} and Lemma \ref{lem:exphaml} in the appendix while noting that $\S \circ \E_0 = 0$,
\begin{equation*}
    \| \Phi^t \circ \E_0(\rho) - (\E_0(\rho) - t i [\Ha, \E_0(\rho)]) \|_1
        \leq \frac{t^2 \|\L\|_{\Diamond}^2 \exp(t \|\L\|_{\Diamond})}{2} \pl.
\end{equation*}
Also,
\begin{equation*}
    \| R^{- t i [\Ha, \cdot]} \circ \E_0(\rho) - (\E_0(\rho) - t i [\Ha, \E_0(\rho)]) \|_1
        \leq \frac{4 t^2 \|\Ha\|_\infty^2 \exp(2 t \|\Ha\|_{\infty})}{2} \pl.
\end{equation*}
Therefore,
\begin{equation} \label{eq:t-bound-prop}
\begin{split}
    \| \Phi^t \circ \E_0(\rho) - R^{- i t \Ha} \circ \E_0(\rho) \|_1
        \leq \frac{e}{2} t^2 \big ( \|\L\|_{\Diamond}^2 e^{t \|\L\|_{\Diamond}} + 4 \|\Ha\|_\infty^2 e^{2 t \|\Ha\|_{\infty}} \big ) \pl.
\end{split}
\end{equation}
Using Lemma \ref{lem:wsb},
\begin{equation} \label{eq:entropy-diff-prop}
\begin{split}
    & |D(R^{- i t \Ha} \E_0(\rho) \| \E(R^{- i t \Ha} \E_0(\rho))) - D(\Phi^t \E_0((\rho)) \| \E(\Phi^t(\E_0(\rho))))|
    \\ & \leq c_t t^2 \log C(\E) + (1 + c_t t^2) h \big (\frac{c_t t^2}{1+c_t t^2} \big ) ,
\end{split}
\end{equation}
where $\epsilon = \frac{1}{2} \| \Phi^t \E_0 (\rho) - R^{- i t \Ha} \E_0 (\rho) \|$. Since $\E$ is the decoherence-free subspace projection of $\L$, $D(R^{- i t \Ha} \E_0(\rho) \| \E(R_{- i t \Ha} \E_0(\rho))) = D(\E_0(\rho) \| \E(\rho))$. 
In contrast, CMLSI would predict that
\[ 
D(\E_0(\rho) \| \E(\rho)) - D(\Phi^t \E_0((\rho)) \| \E(\Phi^t(\E_0(\rho))))
        \geq (1 - e^{-\lambda t}) D(\E_0(\rho) \| \E(\rho)) \pl.
\]
To simplify the CMLSI bound, recall that $e^{-\lambda t} \geq 1 - \lambda t$, so
\[ 
    D(\E_0(\rho) \| \E(\rho)) - D(\Phi^t \E_0((\rho)) \| \E(\Phi^t(\E_0(\rho))))
        \geq \lambda t D(\E_0(\rho) \| \E(\rho)) \pl.
\]
For sufficiently small $t$, the linear reduction predicted by CMLSI is incompatible with quadratic order decay. Via the definition of the Pimsner-Popa index and Equation \eqref{eq:pimsnerpopa}, there exists a $\rho$ achieving $D(\E_0(\rho) \| \E(\rho)) = \log C(\E_0 \| \E) > 0$ when $\E_0 \neq \E$.
\end{proof}
Though Lemma \ref{lem:commute} shows that the Hamiltonian term in Equation \eqref{eq:ldecomp} does not directly block decay via the dissipator, a number of previous results suggest that the dissipator may alter the effective Hamiltonian. These results include the generalized Zeno or watchdog effect \cite{misra_zenos_1977, barankai_generalized_2018, mobus_quantum_2019, burgarth_quantum_2020, becker_quantum_2021, mobus_optimal_2023}, adiabatic theorems \cite{kato_adiabatic_1950, burgarth_generalized_2019}, and strong damping limits \cite{popkov_effective_2018, zanardi_coherent_2014, burgarth_generalized_2019, burgarth_one_2022} in which a fast dissipator modifies the effective dynamics of a simultaneous, slower process.  We derive a variation on the strong damping limit from cp-order assumptions:
\begin{cor} (Strong Damping) \label{cor:ctsfinal}
Let $\S$ be a bounded Lindbladian in submultiplicative norm $\|\cdot\|$ with contractive generated semigroup and fixed point projection $\E_0$ such that $\exp(- \S \tau ) \geq_{cp} (1 - \epsilon) \E_0$ for some $\tau \in \RR^+$. Let $\L$ be a bounded Lindbladian such that $e^{-\L t}$ is also contractive for all $t \in \RR^+$. Then for every $t \in \RR^+$,
\begin{equation*}
\begin{split}
 & \big \| e^{- (\S + \L) t} \E_0 - e^{- \E_0 \L \E_0 t} \E_0 \big \| 
    \leq \frac{7 e^2 \tau \|\L \|^2 t }{1-\epsilon} \pl .
\end{split}
\end{equation*}
\end{cor}
\begin{proof}
Via the assumption that $\exp(- \S \tau ) \geq_{cp} (1 - \epsilon) \E_0$, $\exp( - \S t / k)^g \geq_{cp} (1 - \epsilon) \E_0$ as long as $g t / k \geq \tau$ for any $k \in \NN$. Via Theorem \ref{thm:zeno} applied to $(\exp(-\S / k) \exp(- \L / k))^k$,
for any $k \in \NN$, $q \in 1...k$, and $k_1, ..., k_q$ for which $k_1 + ... + k_q = k$,
\begin{equation*}
\begin{split}
 & \big \| \big (e^{- \S t / k} \circ e^{- \L t / k} \big )^k \E_0 - e^{- \E_0 \L \E_0 t} \E_0 \big \|
    \leq \\ & \sum_{m=1}^q \bigg ( 
        2 k_m \epow{\| \L \| t / k}{2} +  \frac{\lceil k \tau / t \rceil}{k_m} \Big ( 1 + \frac{1}{k_m} \Big ) \frac{(\|\L \| t k_m / k)^3 + 4 (\|\L \| t k_m / k )^2 + 2 \|\L \| t k_m / k }{1-\epsilon} e^{\|\L \| t k_m / k} \bigg ) .
\end{split}
\end{equation*}
Set $q = \lceil \|\L\| t \rceil$ and assume that $k$ is a multiple of $q$. Then set $k_1 = ... = k_q = k/q$. Note that $t \|\L\| / \lceil t \|\L\|\rceil = t \|\L\| / q  \leq 1$. Hence
\begin{equation*}
\begin{split}
 & \Big \| \prod_{m=1}^k \big (e^{- \S t / k} \circ e^{- \L t / k} \big ) \E_0 - \prod_{m=1}^k e^{- \E_0 \L \E_0 t / k} \E \Big \| 
    \leq \frac{7 e \lceil \|\L\| t \rceil \lceil k \tau / t \rceil}{k} \Big ( 1 + \frac{q}{k} \Big ) \frac{\|\L \| t }{1-\epsilon} \pl .
\end{split}
\end{equation*}
The Suzuki-Trotter formula \cite{suzuki_generalized_1976} yields that $\exp(-t (\S + \L)) = \lim_{k \rightarrow \infty} (e^{- \S t / k} \circ e^{- \L t / k} )^k$. In this limit, $\lceil k \tau / t \rceil \rightarrow k \tau / r$, and $q / k \rightarrow 0$. This limit yields that
\[ \big \| e^{- (\S + \L) t} \E_0 - e^{- \E_0 \L \E_0 t} \E_0 \big \| 
    \leq \min \frac{7 e \tau \|\L \| \lceil \|\L\|t  \rceil}{1-\epsilon} \pl.
\]
Alternatively, setting $q = 1$ yields that
\[
\big \| e^{- (\S + \L) t} \E_0 - e^{- \E_0 \L \E_0 t} \E_0 \big \| 
    \leq  \frac{7 e \tau \|\L \| ^2 t }{1-\epsilon} e^{\| \L\| t}  \pl.
\]
If $\|\L\| t \leq 1$, then $\exp(\| \L\| t) \leq e$. If $\|\L\| t \geq 1$, then $\lceil \| \L \| t \rceil \leq \| \L \| t + 1 \leq e \| \L \| t$. Therefore, an extra factor of $e$ allows us to combine the two bounds, avoiding both the ceiling function and the exponential growth.
\end{proof}
Unlike some previous formulations \cite{zanardi_coherent_2014, popkov_effective_2018, burgarth_generalized_2019}, Corollary \ref{cor:ctsfinal} depends on decay properties of the damping Lindbladian rather than on an explicit, multiplying constant. This property will be essential for obtaining bounds from CMLSI and other norm or entropy decay. In the literature on strong damping, strong coupling, and generalized Zeno effects, it is broadly expected that the asumptotic order should be $O(1/\gamma)$ for continuous damping with strength paramter $\gamma$ \cite{burgarth_generalized_2019}, or $O(1/k)$ in discrete interruptions when $k$ is the interruption number within a fixed time interval \cite{mobus_optimal_2023}. Corollary \ref{cor:ctsfinal} and Theorem \ref{thm:zeno} present relatively simple formulations with the optimal $1/\gamma$ asymptotic order expected in these bounds, and a polynomial dependence on $\|\L\|$, which translates to a polynomial time-dependence when including an explicit time parameter.

The QMS generated by $\E_0 \L \E_0$ is conventionally referred to as Zeno dynamics. A related, counter-intuitive observation is that a chain of projections may approach the action of a unitary on a particular subspace. Let $\E_t := R_{\exp(-i \Ha t)} \circ \E_0 \circ R_{\exp(i \Ha t)}$ for Hamiltonian $\Ha$ and any $t \in \RR$. Note that for any $k \in \NN$, $\E_{t} \E_{t-1/k} ... \E_{1/k} = R_{\exp(-i \Ha t)} \circ (\E_0 R_{\exp(i \Ha t/k)})^k \pl$. As a direct consequence,
\begin{equation} \label{eq:condexpunitary}
\lim_{k \rightarrow \infty} \E_{t} \E_{t-1/k} ... \E_{1/k} \E_0 = R_{\exp(i (\E_0(\Ha) - \Ha) t)} \circ \E_0
\end{equation}
Though each $\E_t, ..., \E_{1/k}$ is a projection that we might interpret as rotation-free, in the continuum limit, the chain of composed projections approaches unitary rotation following $\E_0$. In the limit, such a chain of conditional expectations does not induce decay of a state toward an intersection of fixed point subspaces but rotates the subspace projected to by $\E_0$.
\begin{rem} \label{rem:clsinorm} \normalfont
If a $d$-dimensional semigroup $\Phi^t$ generated by $\S$ induces decay toward decoherence-free subspace projection $\E$ as
\begin{equation} \label{eq:decayassump}
    D(\Phi^t(\rho) \| R^t \E(\rho)) \leq e^{-\lambda t} D(\rho \| \E(\rho)) \pl,
\end{equation}
for $t$ in some range, then via Pinsker's inequality as in Equation \eqref{eq:pinsker}, for any $\E_0$ such that $\E \E_0 = \E_0 \E = \E$,
\[ \| \Phi^t \circ \E_0 - R^t \E\|_{\Diamond} \leq e^{- \lambda t / 2} \sqrt{2 \ln C_{cb}(\E_0 \| \E)} \pl. \]
where $C_{cb}(\E_0 \| \E)$ is the Pimsner-Popa index as in Equation \eqref{eq:pimsnerpopa}. Recall that the diamond norm upper bounds the infinity norm difference of Choi matrices. In dimension $d$, there exists a constant $c \leq d^2$ for which if any channel $\Psi$ has that $\E \Psi = \Psi \E = \E$, and $\| \Psi - \E\|_{\Diamond} \leq \epsilon$, then $\Psi \geq_{cp} (1 - c \epsilon) \E$ \cite[Proposition 11.16]{laracuente_quasi-factorization_2022}. Therefore, if $\Phi^t$ induces decay as in Equation \eqref{eq:decayassump}, and $R^t = \Id$ for all $t \geq 0$, then
\[ \Phi^t \geq_{cp} \big (1 - e^{- \lambda t / 2} c \sqrt{2 \ln C_{cb}(\E_0 \| \E)} \big ) \E \pl. \]
In particular, to achieve $\Phi^t \geq_{cp} \E / 2$, one may set
\begin{equation} \label{eq:defaultt}
    t = \frac{2}{\lambda} \ln (2 c \sqrt{2 \ln C_{cb}(\E_0 \| \E)})
\end{equation}
assuming the exponential decay inequality holds at this value.
\end{rem}
\begin{prop} \label{prop:toptimizezeno}
    Let $\L(\cdot) = \L' + \S(\cdot)$ generate QMS $(\Phi^t)$, where $\S$ a Lindbladian generating semigroup $(\Phi_0^t)$ with $\lambda_0$-CMLSI and fixed point conditional expectation $\E_0$. Let $\| \cdot\|$ be a submultiplicative superoperator norm. Let $c \leq d^2$ be the smallest constant for which if any channel $\Psi$ has that $\E \Psi = \Psi \E = \E$, and $\| \Psi - \E\| \leq \epsilon$, then $\Psi \geq_{cp} (1 - c \epsilon) \E$. Then
    \[ \frac{1}{2} \| e^{- t \E_0 \circ \L' \circ \E_0} \E_0 - \Phi^t \E_0 \| 
        \leq \frac{t}{\lambda_0} \times 28 e^2 \|\L'\|^2 \ln \Big (2 c \sqrt{2 \ln C_{cb}(\E_0 \| \E)} \Big)
        \pl. \]
\end{prop}
\begin{proof}
         Let
    \[ \delta_t := \frac{1}{2} \| e^{\L' t} \E_0 - \Phi^t \E_0 \| \pl. \]
    Recalling Corollary \ref{cor:ctsfinal},
    \[ \delta_t \leq \frac{7 e^2 \tau \|\L'\|^2 t}{1-\epsilon} \]
    for any $\epsilon, \tau > 0$ for which $\exp(- \S \tau) \geq_{cp} (1-\epsilon) \E_0$. Setting $\tau = (2/\lambda_0) \ln (2 c \sqrt{2 \ln C_{cb}(\E_0 \| \E)})$ and $\epsilon = 1/2$ as in Equation \eqref{eq:defaultt} yields that 
    \begin{equation*} \label{eq:deltabound}
        \delta_t \leq t \times 28 e^2 \|\L\|^2 \ln \big (2 c \sqrt{2 \ln C_{cb}(\E_0 \| \E)} \big) / \lambda_0 \pl.
    \end{equation*}
    Observe that $\delta_t$ is linear in $t / \lambda_0$.
\end{proof}
\begin{prop} \label{prop:normconverse}
    Let $\L(\cdot) = i [\Ha, \cdot] + \S(\cdot)$ generate QMS $(\Phi^t)$ with decoherence-free subspace projection $\E$, where $\S$ a Lindbladian generating semigroup $(\Phi_0^t)$ with $\lambda_0$-CMLSI and fixed point conditional expectation $\E_0$. Assume that $\E \neq \E_0$, that $\E \E_0 = \E_0 \E = \E$.
    Let $c \leq d^2$ be the smallest constant for which if any channel $\Psi$ has that $\E \Psi = \Psi \E = \E$, and $\| \Psi - \E\|_{\Diamond} \leq \epsilon$, then $\Psi \geq_{cp} (1 - c \epsilon) \E$. Then for every input density $\rho$,
    \[ \frac{1}{2} \| (\Phi^t \circ \E_0 - \Phi^t \circ \E)(\rho) \|_1 \geq c_0(\rho) -  \frac{t}{\lambda_0} \times 224 e^2 \| \Ha \|_{\infty}^2 \ln  \big ( 2 c \sqrt{2 \ln C_{cb}(\E_0)} \big ) \pl, \]
    where
    \[ c_0(\rho) : = \frac{1}{(\log_2 C(\E))^2} \Big ( \sqrt{1 + \log_2 D( \E_0(\rho) \| \E(\rho)) \log_2 C(\E)} - 1 \Big )^2 \pl.\]
\end{prop}
\begin{proof}
    That $\{\E_0(\omega) | D(\E_0(\omega) \| \E(\omega)) = \log C(\E_0 \| \E) \}$ is non-empty follows from the definition of the Pimsner-Popa index as in Equation \eqref{eq:pimsnerpopa} as the maximum value of $D(\E_0(\cdot) \| \E(\cdot))$.
    
    Via the triangle inequality,
    \begin{equation*}
    \frac{1}{2} \| (\Phi^t \E_0 - \Phi^t \E)(\rho) \|_1 \geq \frac{1}{2} \big \| (\Phi^t \E -  e^{- i t \E_0 \circ [\Ha, \cdot] \circ \E_0} \E_0)(\rho) \big \|_1
    		- \frac{1}{2} \big \| (\Phi^t \E_0 - e^{- i t \E_0 \circ [\Ha, \cdot] \circ \E_0 } \E_0 )(\rho) \big \|_1 \pl.
    \end{equation*}
    Again by the triangle inequality and invariance of trace distance under rotation,
    \begin{equation*}
    \begin{split}
        & \frac{1}{2} \big \| (\Phi^t \E -  e^{- i t \E_0 \circ [\Ha, \cdot] \circ \E_0} \E_0)(\rho) \big \|_1
     \\ & \geq   \frac{1}{2} \big \| ( e^{- i t \E_0 \circ [\Ha, \cdot] \circ \E_0} \E -  e^{- i t \E_0 \circ [\Ha, \cdot] \circ \E_0} \E_0)   (\rho) \big \|_1
        -  \frac{1}{2} \big \| (\Phi^t \E -  e^{- i t \E_0 \circ [\Ha, \cdot] \circ \E_0} \E)(\rho) \big \|_1 \pl.
     \\ & =   \frac{1}{2} \big \| ( \E - \E_0)   (\rho) \big \|_1
        -  \frac{1}{2} \big \| (\Phi^t \E -  e^{- i t \E_0 \circ [\Ha, \cdot] \circ \E_0} \E)(\rho) \big \|_1 \pl.
    \end{split}
    \end{equation*}
    Let $c_0 :=\frac{1}{2} \big \| (\E - \E_0) (\rho) \big \|_1$.
    By Lemma \ref{lem:wsb},
    \begin{equation}
        |D(\E(\rho) \| \E('')) - D(\E_0(\rho) \| \E(''))|
            \leq c_0 \log_2 C(\E) + 2 \sqrt{c_0} \pl.
    \end{equation}
    where each $''$ as the second argument to relative entropy repeats the first argument inside the parenthesis. Since $\Phi_0^t$ commutes with $\E$, $D(\Phi_0^t \E(\rho) \| \E('')) = 0$. Therefore,
    \begin{equation}
        D( \E_0(\rho) \| \E(''))
            \leq c_0 \log_2 C(\E) + 2 \sqrt{c_0} \pl.
    \end{equation}
    Via the quadratic formula,
    \begin{equation}
    \begin{split}
        \sqrt{c_0} \geq & \frac{- 2 + \sqrt{4 + 4 \log_2 D( \E_0(\rho) \| \E(\rho)) \log_2 C(\E)} }{2 \log_2 C(\E)}
        \\ & = \frac{1}{\log_2 C(\E)} \Big ( \sqrt{1 + \log_2 D( \E_0(\rho) \| \E(\rho)) \log_2 C(\E)} - 1 \Big ) \pl.
    \end{split}
    \end{equation}
    Recall that
    \begin{equation*}
    \begin{split}
        & \frac{1}{2} \big \| (\Phi^t \E -  e^{- i t \E_0 \circ [\Ha, \cdot] \circ \E_0} \E_0)(\rho) \big \|_1
            \geq c_0
            \\ & - \frac{1}{2} \big \| (\Phi^t \E_0 - e^{- i t \E_0 \circ [\Ha, \cdot] \circ \E_0 } \E_0 )(\rho) \big \|_1
                - \frac{1}{2} \big \| (\Phi^t \E -  e^{- i t \E_0 \circ [\Ha, \cdot] \circ \E_0} \E)(\rho) \big \|_1 \pl.
    \end{split}
    \end{equation*}
    Each of the subtracted terms is bounded by Proposition \ref{prop:toptimizezeno}. Finally, we convert the factor of $28$ multiplying each $\|\L'\|^2$ therein to four times that multiplying $\|\Ha\|_\infty^2$ via Lemma \ref{lem:hamlnorm}. 
\end{proof}
\begin{theorem}[Self-Restricting Noise, Technical Version of Theorem \ref{thm:zvscmlsiintro}] \label{thm:zvscmlsi}
     Let $\L(\cdot) = i [\Ha, \cdot] + \S(\cdot)$ generate QMS $(\Phi^t)$ in dimension $d$, where $\S$ a Lindbladian generating semigroup $(\Phi_0^t)$ with $\lambda_0$-CMLSI and fixed point conditional expectation $\E_0$. Let $\E$ be the decoherence-free subspace projection of $(\Phi^t)$. Assume that $\E \neq \E_0$, and $\E \E_0 = \E_0 \E = \E$. Let $c \leq d^2$ be the smallest constant for which if any channel $\Psi$ has that $\E \Psi = \Psi \E = \E$, and $\| \Psi - \E\|_{\Diamond} \leq \epsilon$, then $\Psi \geq_{cp} (1 - c \epsilon) \E$. If relative entropy is defined with the natural logarithm, for any $a \in (0,1)$ and input state $\rho$,
     \[ D(\Phi^t (\E_0(\rho)) \| \Phi^t \circ \E(\rho)) \geq e^{- \lambda t} D(\E_0(\rho) \| \E(\rho)) \]
     at $t = a \lambda_0 c_0 / \delta \| \Ha \|_\infty^2$, where 
     \begin{equation*}
     \begin{split}
         \lambda \leq \frac{\delta \| \Ha \|_\infty^2}{\lambda_0 c_0 a} \times \ln \Big ( \frac{D(\E_0(\rho) \| \E(\rho))}{2 (1-a)^2 c_0^2}  \Big ) \pl,
     \end{split}
     \end{equation*}
     \begin{equation*}
         c_0 = \frac{1}{(\log_2 C(\E))^2} \Big ( \sqrt{1 + \log_2 D(\E_0(\rho) \| \E(\rho)) \log_2 C(\E)} - 1 \Big )^2 \pl, \text{ and }
         \delta = 224 e^2 \ln  \big ( 2 c \sqrt{2 \ln C_{cb}(\E_0)} \big ) \pl.
     \end{equation*}
     If $\rho \in \{\E_0(\omega) | D(\E_0(\omega) \| \E(\omega)) = \log C(\E_0 \| \E) \}$, then one may replace $D(\E_0(\rho) \| \E(\rho))$ by $C(\E_0 \| \E)$ for state-independent constants.
\end{theorem}
In Theorem \ref{thm:zvscmlsiintro}, we set $c_1 = \log C(\E_0 \| \E)$ and other constants as noted.
\begin{proof}
    Let $D_0 := D(\E_0(\rho) \| \E(\rho))$. Assuming that
     \[ D(\Phi^t (\E_0(\rho)) \| \Phi^t \circ \E(\rho)) = e^{- \lambda t} D(\E_0(\rho) \| \E(\rho)) \pl, \]
     we must have that
    \[ \lambda = \frac{1}{t} \ln \frac{ D_0}{D(\Phi^t \circ \E_0(\rho) \| \Phi^t \circ \E(\rho))} \pl. \]
    Via Pinsker's inequality as in Equation \eqref{eq:pinsker},
    \[ D(\Phi^t \circ \E_0(\rho) \| \Phi^t \circ \E(\rho)) \geq \frac{1}{2} \| \Phi^t \circ \E_0(\rho) - \Phi^t \circ \E (\rho) \|_1^2 \pl. \]
    Let $\L'(\cdot) := i [\Ha, \cdot]$. Assuming that $D(\E_0(\rho) \| \E(\rho))$ is maximized by choice of $\rho$ and invoking Proposition \ref{prop:normconverse},
    \[ D(\Phi^t \circ \E_0(\rho) \| \Phi^t \circ \E(\rho)) \geq 2 \Big ( c_0
    		-  \tilde{\delta} \frac{t}{\lambda_0} \Big )^2 \pl, \]
    where $\tilde{\delta} = \| \Ha \|_{\infty}^2 \delta$. Therefore,
    \begin{equation*}
        \lambda \leq - \frac{1}{t} \ln \Big ( \frac{2}{D_0} \Big ( c_0
    		-  \tilde{\delta} \frac{t}{\lambda_0} \Big )^2  \Big ) 
    \end{equation*}
    whenever $t$ is large enough that the argument of the natural logarithm is positive. For a convenient bound, set
    \begin{equation*}
        t = c_0 \times \frac{a \lambda_0}{\tilde{\delta}} \pl,
    \end{equation*}
    obtaining that
    \begin{equation*}
        \lambda \leq \frac{\tilde{\delta}}{\lambda_0 a} \times \frac{1}{c_0} \times \ln \Big ( \frac{D_0}{2 (1-a)^2} \Big ( \frac{1}{c_0} \Big )^2  \Big ) \pl.
    \end{equation*}
    The primary bound of the Theorem follows.
\end{proof}
When $\E \neq \E_0$, Theorem \ref{thm:zvscmlsi} upper bounds the decay rate of $\L$, $\lambda$, \textit{inversely} to that of $\S, \lambda_0$. This phenomenon is what we call \textit{self-restricting noise}. The effect is clearest for unital noise, as because $\Ha$ generates unitary time-evolution, it cannot reverse the effects of mixing. Were $\S$ to act alone, it would induce decay at rate $\lambda_0$ but possibly in a smaller portion of the system, as $\E_0$ leaves a larger subspace decoherence-free than does $\E$. When $\S$ acts simultaneously with $\Ha$, the Hamiltonian spreads the noise around, inducing more decay. When $\S$ has a large (C)MLSI constant, however, the strong damping effect of Corollary \ref{cor:ctsfinal} effectively reduces the magnitude of the noise-spreading components of $\Ha$, shifting back toward the regime in which $\S$ acts alone, but not enough to prevent $\E$ from still being the ultimate fixed-point projection. The Hamiltonian does not slow down the dissipator-induced noise - the noise impedes its own means of spread. Nonetheless, the system has CMLSI neither toward the subspace projected to by $\E_0$, since the wider spread of mixing could actually increase the relative entropy with respect to $\E_0$, nor to $\E$, since decay takes time to spread from its initial subspace to the wider system.

Corollary \ref{cor:ctsfinal} differs from results of \cite{burgarth_quantum_2020} by bounding convergence in terms of cp-order, which leads to a bound in terms of CMLSI. This distinction may seem subtle, but it is essential for Theorem \ref{thm:zvscmlsi}, which compares decay constants of a Lindbladian and its rotation-free constituent. Furthermore, the asymptotic dependence on $t$ and $\L$ is explicitly shown and of polynomial order. Remark \ref{rem:alternate} exposits an alternate approach to the self-restricting noise bound with a less optimal asymptotic depenendence. As seen in Example \ref{ex:graphs}, there are cases in which stronger decay does not arise from multiplicative comparability. A related phenomenon is that fast dissipation may protect systems from Hamiltonian interactions with other environment systems \cite{erez_correcting_2004}, acting as a sort of error suppression.
\begin{rem} \normalfont \label{rem:dd}
Theorem \ref{thm:zvscmlsi} is \textit{not} a consequence of dynamical decoupling, strong coupling, or feed-forward error correction. Furthermore, analogous effects are unlikely to arise from emergent, passive error correction in the primary settings considered.

One may consider the unitary in Equation \eqref{eq:stinespring} to be generated by a combination of internal dynamics, $\Ha_{\text{sys}}$, environmental dynamics, $\Ha_{\text{env}}$ and system-environment coupling, $\Ha_{\text{sys-env}}$. It was previously observed \cite{burgarth_generalized_2019, burgarth_one_2022} that within finite dimension, and if $\Ha_{\text{env}}$ is not too fast, then
\[ \lim_{\gamma \rightarrow \infty} e^{i (\gamma \Ha_{\text{sys}} + \Ha_{\text{env}} + \Ha_{\text{sys-env}}) t}
    = e^{i \gamma \Ha_{\text{sys}} t } e^{i (\Ha_{\text{env}} + \E(\Ha_{\text{sys-env}}))t} . \]
In this case, $\E$ is a projection to the subalgebra of operators that commute with $\Ha_{\text{sys}}$. If $\Ha_{\text{sys}}$ is replaced by a series of discrete interruptions, one obtains a form of dynamical decoupling \cite{viola_dynamical_1998, viola_dynamical_1999}. For a more comprehensive review of dynamical decoupling and related techniques, see \cite{lidar_review_2014, suter_colloquium_2016}.

In the Markov regime, however, adding a Hamiltonian generally does not induce dynamical decoupling or strong coupling against the dissipative term \cite{gough_can_2017}. Primarily, the dissipative assumption requires that $\Ha_{\text{env}}$ acts much more quickly than $\Ha_{\text{sys-env}}$, effectively de-localizing correlations between the environment and the system's past extremely quickly. This assumption yields the divisibility of entropy decay. To illustrate, we recall the depolarizing semigroup $(\Phi^t_{\mathrm{dep}(1)})$ given in Equation \eqref{eq:depsemi}.
At any $t \geq 0$, this channel commutes with application of any local unitary. Hence if one applies a series of dynamical decoupling unitaries $R_0, ..., R_{k}$, one finds that
\[ R_0 \circ \prod_{j=1}^k (R_j \circ \Phi^{t/k}) = R_0 \circ ... \circ R_k \circ \Phi_{\mathrm{dep}(1)}^{t} \pl. \]
Similarly, for the Lindbladian $\S_{dep}$ generating depolarizing noise,
\[ e^{- (i [\Ha_{\text{sys}}, \cdot] + \S_{dep}) t} = e^{- i [\Ha{\text{sys}}, \cdot] t} e^{- \S_{dep} t} \pl. \]
No unitary may undo the entropy-generating effects of $\Phi_{\mathrm{dep}(1)}^t$. Hence unitary interruption does not forestall the decay predicted by (C)MLSI, as generalized and formalized in Lemma \ref{lem:commute}. To apply dynamical decoupling to this and related forms of noise, one must use fast pulses that bypass the dissipative approximation, exposing the underlying system-environment unitary coupling. The pulse rate needed depends not just on $\Ha_{\text{sys-env}}$, but on $\Ha_{\text{env}}$.

In contrast, Theorem \ref{thm:zvscmlsi} concerns the regime in which $\Ha_{\text{sys-env}}$ and $\Ha_{\text{env}}$ generate faster dynamics than $\Ha_{\text{sys}}$, formalized by the effective dissipator $\S$ having $\lambda_0$-(C)MLSI such that $\lambda_0$ is large compared to the maximum eigenvalue of $\Ha_{\text{sys}}$. This is the opposite regime from that of the strong coupling limit or that in which dynamical decoupling is expected to perform well. The regime of Theorem \ref{thm:zvscmlsi} is related to strong damping, in which a dissipative processes suppresses other processes. However, it appears more surprising that the limit of fast dissipation (rather than of strong coupling or dynamical decoupling) may nonetheless suppress noise. Some precedent for self-restriction is implicit in earlier works on strongly dissipative regimes \cite{zanardi_coherent_2014, popkov_effective_2018}, although there the strength of dissipation is tuned by an explicit constant multiplying the Lindbladian term rather than a (C)MLSI constant as herein.

Since the examples we will consider do not include feed-forward or explicitly describe measurements, there is no active error correction. There is some connection to error reduction via Zeno-based dissipative error reduction \cite{beige_quantum_2000}. In the scenario of this paper, however, noise restriction via Zeno-like effects is self-mediated, rather than constructed through additional interactions. It is also possible for dissipative interactions to induce passive error correction \cite{liu_dissipative_2024}. As we primarily focus on unital Lindbladians, the emergence of passive error correction is however unexpected in this regime.
\end{rem}
\begin{rem} \normalfont
An immediate consequence of Proposition \ref{prop:decayhaml} is that (C)MLSI fails for asymptotically small time. As studied in the following Subsection \ref{sec:decayyes}, one can often obtain (C)MLSI-like decay at finite timescales, suggesting that (C)MLSI might still be a reasonable effective model at physically relevant scales. Theorem \ref{thm:zvscmlsi}, however, upper bounds decay rates if either  $\|\Ha\|_\infty$ is small or the decay induced by $\S$ alone large. Recalling Remark \ref{rem:dd}, here we see that noise becomes constrained when the dissipative interactions overpower the system's internal Hamiltonian, a regime \textit{opposite} to that of usual dynamical decoupling or strong coupling.
\end{rem}

\subsection{Repeatable CSDPI for Unital Markov Processes at Finite Timescales} \label{sec:decayyes}
The breakdown of (C)MLSI in Proposition \ref{prop:decayhaml} for unital QMSs when $\E \neq \E_0$ arises from the fact that at small timescales, the dissipator is still effectively decaying the system toward the subspace of $\E_0$, but at long timescales toward that of $\E$. The underlying component that fails is divisibility: the exponential decay rate in relative entropy toward $\E$ at time $t > 0$ is not the integral of decay rates from times 0 to $t$, violating Equation \eqref{eq:divisdecay}. As noted in Remark \ref{rem:dd}, the action of a Hamiltonian cannot reverse relative entropy decay toward a unital fixed point subspace projection. Hence a natural question is whether an indivisible generalization of CMLSI holds more generally for unital QMSs. In this Section, we answer this question in the affirmative.

Previously, \cite[Theorem 2.6]{gao_complete_2025} showed that for any unital quantum channel $\Phi$, there is a CSDPI constant $\lambda_\Phi \in [0,1)$ such that
\[ D(\Phi(\rho) \| \Phi \circ \E_{\Phi}(\rho)) \leq \lambda_\Phi D(\rho \| \E_{\Phi}(\rho)) \]
and gave a procedure to calculate $\lambda_\Phi$ in terms of the behavior of $(\Phi \Phi_*)^k$ for sufficiently large $k \in \NN$. An immediate question is whether $\Phi^k$ immediately has (C)SDPI constant $\lambda_\Phi^k$ by iteration. If so, such a result would immediately lead to a discrete-time (C)MLSI-like bound. As shown in Counterexample \ref{cexam:csdpi}, however, such iteration is not always possible. In general, $[\Phi, \E_{\Phi}] \neq 0$, so the same (C)SPDI result no longer applies after the first step. More strikingly, we also see that sometimes there is no initial decay to the large-$k$ fixed point of $\Phi^k$, since the eventual fixed point subspace projection may not coincide with the multiplicative domain projection. Other results of \cite{gao_complete_2025} avoid this barrier by assuming GNS detailed balance.

The intuition for our adjustment is to show that the semigroup property suffices to ensure that $\Phi^{t} \E = \E \Phi^{t} = R^{t} \E = \E R^{t}$ for some unitary conjugation family $R^{t}$, and for a projection $\E$ that does not depend on the choice of $t \geq 0$. A natural intuition would be that as $\Phi^{t}$ is analytic for $t \geq 0$, certain qualitative properties must be preserved for all $t \geq 0$. If this analyticity passes to the induced decay, then it should be impossible for a QMS that induces late-time decay not to induce decay in any interval $(0,\tau)$ for arbitrarily small $\tau > 0$.
\begin{lemma} \label{lem:semigroupcsdpi}
Let $(\Phi^{t})_{t \geq 0}$ be a unital, finite-dimensional QMS. Then there is a decoherence-free subspace projection $\E$, and for each $\tau \geq 0$ there is a minimal $k_{\tau} \in \NN$ such that for all $k \geq k_{\tau}$,
\begin{equation} \label{eq:kcond}
(1 + 1/10) \E \geq_{cp} (\Phi^{\tau}_* \Phi^{\tau})^{k} \geq_{cp} (1 - 1/10) \E \pl .
\end{equation}
Furthermore there exist unitary rotations $(R^t)$ for which $\Phi^{t} \E = \E \Phi^{t} = R^{t} \E = \E R^{t}$ for all $t \geq 0$, and
\begin{equation*}
\begin{split}
& D((\Phi^{\tau} \otimes \Id^B)(\rho) \| ((\E \circ \Phi^{\tau}) \otimes \Id^B) (\rho))
    \mightbenewline \mightbealign
    \leq (1-1/2 k_{\tau}) D(\rho \| (\E \otimes \Id^B)(\rho))
\end{split}
\end{equation*}
for any finite-dimensional auxiliary system $B$ and input density $\rho$.
\end{lemma}
\begin{proof}
The first part of the Lemma is the assertion that for some $k \in \NN$, $(\Phi^{t}_* \Phi^{t})^k$ becomes cp-order comparable to a projection $\E$ that is uniform in $t$ and projects to the multiplicative domain. By Lemma \ref{lem:converge1}, $(\Phi^{t}_* \Phi^{t})^k$ is arbitrarily close to $\E$ in diamond norm for sufficiently large $k$. By \cite[Proposition II.16]{laracuente_quasi-factorization_2022}, diamond norm convergence indeed implies that $(\Phi^{t}_* \Phi^{t})^k \geq_{cp} (1-\zeta) \E$ for arbitrarily small $\zeta$ with sufficiently large $k$. Let us thereby write $(\Phi^{t}_* \Phi^{t})^k = (1-\zeta) \E + \zeta \Psi$ for some channel $\Psi$. Since $(\Phi^{t}_* \Phi^{t})^k \E = \E (\Phi^{t}_* \Phi^{t})^k = \E$, it must also hold that $\Psi \E = \E \Psi = \E$. Furthermore, since $\E$ is a conditional expectation, $\E \geq_{cp} c \Psi$ for uniform $c$ given by the Pimsner-Popa index as in Equation \eqref{eq:pimsnerpopa}. Hence $(1 - \zeta + \zeta c) \E \geq_{cp} (\Phi^{t}_* \Phi^{t})^k $. Since $\zeta$ can be made arbitrarily small, the cp-order bound converges.

We recall \cite[Theorem 2.6]{gao_complete_2025}, which shows that
\[ D(\Phi^{t}(\rho) \| \Phi^{t} \circ \E(\rho)) \leq (1 - 2 k_{t}) D(\rho \| \E(\rho)) \pl. \]
Lemma \ref{lem:converge2} yields the corresponding family of rotations $(R^t)$ for which $\Phi^{t} \circ \E = \E \circ R^t$.
\end{proof}

A subtle but essential difference between Lemma \ref{lem:semigroupcsdpi} and the bound implied by \cite[Theorem 2.6]{gao_complete_2025} is that in the left-hand side of the inequality, the second argument to relative entropy is $\E(R^t(\rho)) = \E(\Phi^t(\rho))$, rather than $\Phi^t \circ \E(\rho)$. It is because of this order of operations that unlike \cite[Theorem 2.6]{gao_complete_2025}, one may iterate Lemma \ref{lem:semigroupcsdpi} to obtain the following, motivating Theorem:
\begin{theorem}[Technical Version of Theorem \ref{thm:introdecay}] \label{thm:maindecay}
Let $(\Phi^t)_{t \geq 0}$ be a finite-dimensional, unital QMS. For any $\tau > 0$, there exists a minimal $k_\tau \in \NN$ such that
\begin{equation*}
(1 + 1/10) \E \geq_{cp} (\Phi^{\tau}_* \Phi^{\tau})^{k} \geq_{cp} (1 - 1/10) \E
\end{equation*}
for all $k \geq k_\tau$. Furthermore, there exists a family of unitary rotations $(R^t)_{t \geq 0}$ for which $\Phi^{t} \E = \E \Phi^{t} = R^{t} \E = \E R^{t}$, and
\begin{equation*}
\begin{split}
& D((\Phi^t \otimes \Id^B)(\rho) \| ((\E \circ \Phi^{t}) \otimes \Id^B)(\rho)) \mightbenewline \mightbealign
 \leq (1 - 1 / 2 k_{\tau})^{\lfloor t / \tau \rfloor} D(\rho \| (\E \otimes \Id^B)(\rho))
\end{split}
\end{equation*}
for all $t > 0$ and every input density $\rho$, where $\lfloor \cdot \rfloor$ denotes the floor function, and $B$ is any finite-dimensional auxiliary system.
\end{theorem}
\begin{proof}
We apply Lemma \ref{lem:semigroupcsdpi} after noting that $\Phi^t \circ \E \circ R^w = \Phi^{t - \tau} \circ \E \circ R^{w + \tau}$ for any $\tau \leq t$ and $w \geq 0$. Hence
\[ D(\Phi^w(\rho) \| \E_w(\rho)) \geq (1 - 1 / 2 k_{\tau}) D(\Phi^{w + \tau}(\rho) \| \E( R^{w + \tau}(\rho))) \pl. \]
Starting at $w = 0$, we iterate, obtaining that
\[ D(\Phi^{\lfloor t/\tau \rfloor}(\rho) \| \E (R^{\lfloor t/\tau \rfloor}(\rho))) \leq D(\rho \| \E(\rho)) \pl. \]
The final inequality of the Theorem then follows from the data processing inequality and that $\E R^{t - \tau \lfloor t/\tau \rfloor} = \Phi^{t - \tau \lfloor t/\tau \rfloor} \E$.
\end{proof}

We emphasize that Counterexample \ref{cexam:csdpi} illustrates how one cannot simply iterate \cite[Theorem 2.6]{gao_complete_2025} to obtain a result analogous to Theorem \ref{thm:maindecay}, as such would imply too general a bound. Instead, we use the smoothness implied by the QMS structure to relate eventual decay in cp-order to early-time decay in relative entropy.

Theorem \ref{thm:maindecay}'s generality in light of Theorem \ref{thm:zvscmlsi} constrains the size of $\lambda_\tau$ for small $\tau$. Given the breakdown of CMLSI in general, it must sometimes hold that $\lambda_\tau \rightarrow 0$ as $\tau \rightarrow 0$. Indeed, the intuition for this constraint is apparent from Counterexample \ref{cexam:counterchain}.

 \section{Examples and Applications} \label{sec:examples}
In this Section, logarithms are taken with base 2 unless specified as ``$\ln$'' for convenience in describing qubit-based examples. These examples take inspiration from earlier results on effective Zeno dynamics in systems with strong, local dissipation \cite{zanardi_coherent_2014, popkov_effective_2018, burgarth_generalized_2019}.
\begin{cexam}[Failure of (C)MLSI] \normalfont \label{cexam:counterchain}
In this counterexample, consider an $n$-qubit Hamiltonian of the form
$\Ha = \sum_{j=1}^n \Ha_{j,j+1} + \Ha_j \pl.$
This form represents nearest-neighbor interactions on a one-dimensional chain with open boundary conditions. Notable examples include Heisenberg and Ising models in one spatial dimension. Add a dissipative generator $\S = \S_1 \otimes \id^{\otimes n-1}$, which acts only on the leftmost qubit. Physically, one may think of such a system as well-isolated from its noisy environment except for the left end of the chain. Until the $(n-1)$th term in a Taylor expansion of the QMS around $t=0$, no term contains more than $n-2$ qubit swaps. Via continuity of relative entropy to a subalgebra restriction (see \cite[Lemma 7]{winter_tight_2016} and \cite[Proposition 3.7]{gao_relative_2020}),
\[ D(\Phi^t(\rho) \| \Phi^t(\E(\rho))) \geq (1 - O(t^{n-1} \log t)) D(\rho \| \E(\rho)) \pl. \]
Since $\Phi^t$ cannot have MLSI with any $\lambda > 0$, it does not have MLSI.

\begin{figure}[h!] \scriptsize \centering
	\begin{subfigure}[b]{0.3\textwidth}
		\vspace{5.5mm}
		\includegraphics[width=0.98\textwidth]{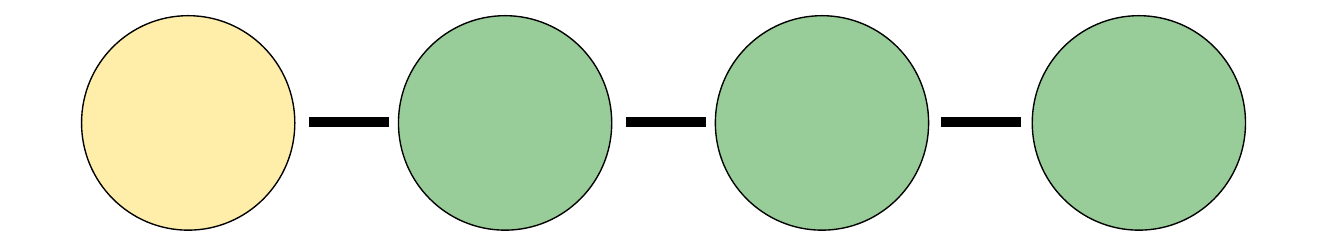}
		\caption{}
	\end{subfigure}
	\hspace{5mm}
	\begin{subfigure}[b]{0.25\textwidth}
		\includegraphics[width=0.98\textwidth]{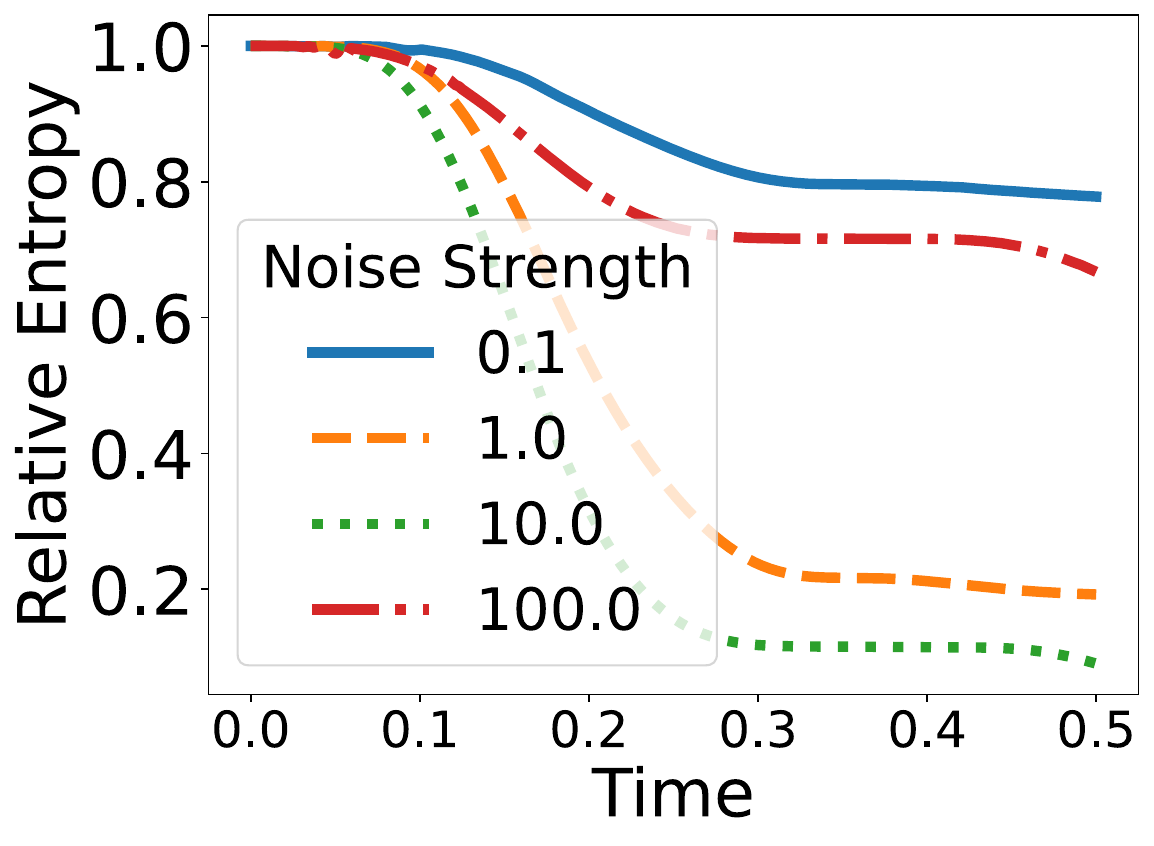}
		\caption{}
	\end{subfigure}
    \hspace{10mm}
	\begin{subfigure}[b]{0.25\textwidth}
	    \includegraphics[width=0.98\textwidth]{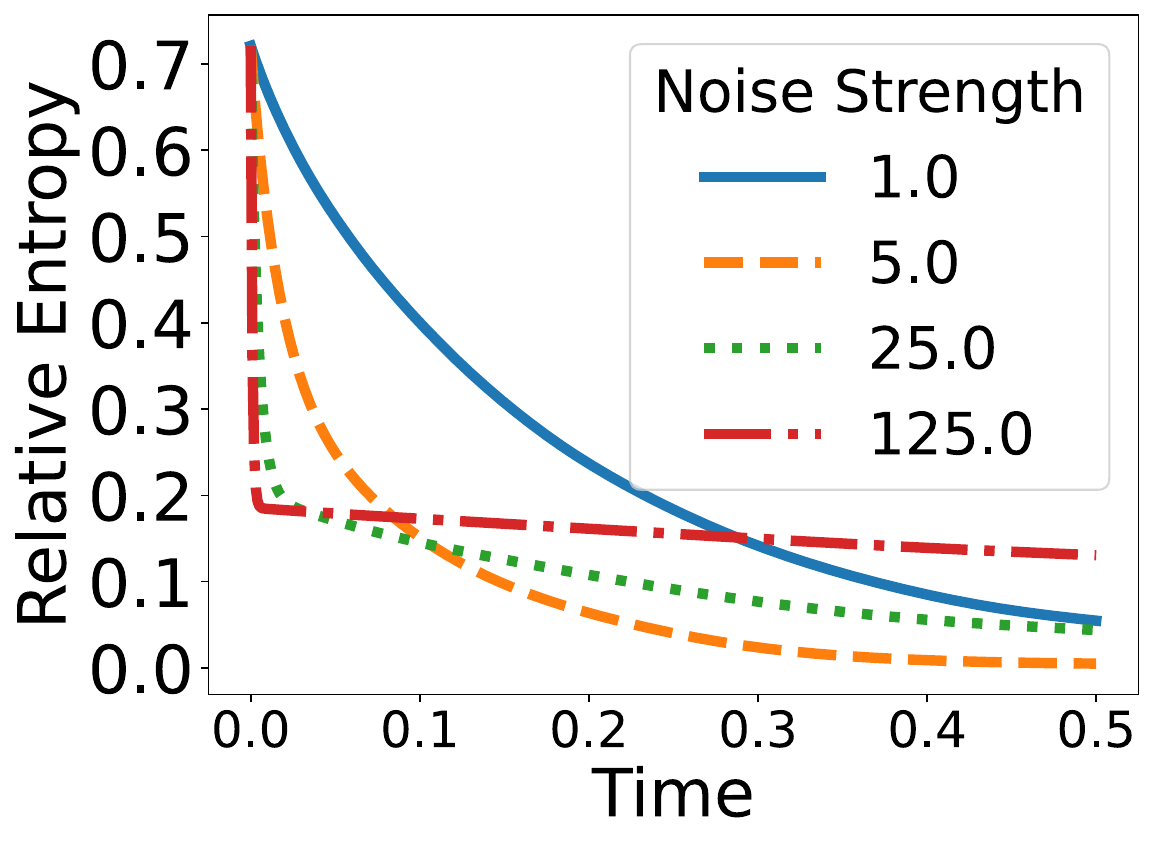}
	    \caption{}
	\end{subfigure}
	\caption{Relative entropy of a 4-qubit spin chain to (fully decayed) complete mixture: (1) Spin chain illustration. The noised qubit is on top and shaded yellow. The 3 qubits below are shaded green with nearest-neighbor interactions. (2) Relative entropy vs. time with input $(\id/2)^{\otimes 3} \otimes \ketbra{0}$, where the legend notes $\gamma$ in Equation \eqref{eq:chainl}. (3) Relative entropy vs. time averaged over 50 randomly selected input densities. \label{fig:spinchain}}
\end{figure}
Using Qiskit Dynamics's Lindbladian solver, we plot relative entropy decay for $XX + YY$ nearest neighbor interactions (see \cite{noauthor_solving_2021} for a simple example of a similar system) on 4 qubits. The simulated Lindbladian has the form
\begin{equation} \label{eq:chainl}
\L(\rho) = i \Big [ 2 \pi \sum_{j=1}^3 (X_j X_{j+1} + Y_j Y_{j+1} ), \rho \Big ] + \gamma S_{dep}(\rho) \otimes \id^{(2-4)} \pl,
\end{equation}
where $\id^{(2-4)}$ denotes the identity operator on the 2nd-4th qubits, $S_{dep}$ is the depolarizing generator as in Equation \eqref{eq:depdisp}, and $X,Y,Z$ denote unnormalized Pauli matrices. By $\rho^{(2-4)}$ we denote the marginal of density $\rho$ on the 2nd-4th qubits. The fixed point is complete mixture, and the initial initial state is one of $\ket{0000}$, $(\id/2)^{\otimes 3} \otimes \ketbra{0}$, or a density randomly generated using seed values 100-149 using the \textit{qiskit.quantum\_info.random\_density\_matrix} method with the Hilbert-Schmidt metric. Results are plotted in Figure \ref{fig:spinchain}. The parameter $\gamma$ multiplies the noise terms, controlling the strength of noise relative to time and interaction terms. Note that one may technically define the depolarizing channel with depolarizing ``probability" greater than 1 - we do not study this regime. Strong noise in our case refers to extremely fast decay to complete mixture, not to extending the parameter range beyond that. Von Neumann entropies are calculated using Qiskit's `quantum\_info.entropy' subroutine and subtracted from the maximum possible entropy to obtain the relative entropy with respect to the completely mixed fixed point state.

Figure \ref{fig:spinchain} shows simulations of the model defined by Equation \eqref{eq:chainl} with initial state $(\id/2)^{\otimes 3} \otimes \ketbra{0}$. Initial decay in Subigure \ref{fig:spinchain}.(2) is of subleading order, showing a violation of CMLSI. More striking is the inversion in the relationship between noise strength and entropy decay when going from 10.0 to 100.0. While decay rate expectedly correlates with noise strength for weak noise, the relationship soon inverts. One sees in Figure \ref{fig:spinchain} a rebound effect in which stronger noise begins to increase rather than decay the relative entropy at fixed time. Extremely strong noise suppresses the interaction between the noised qubit and others, slowing its own spread.

As noted in Remark \ref{rem:dd}, self-restricting noise is not the result of dynamical decoupling. Furthermore, there is no truly decoherence-free subspace in the long-time limit of this example - the late-time fixed point of the QMS is complete mixture. There is also no active error correction, and since the noise is unital, also no way for passive error correction in the usual sense to emerge.
\end{cexam}

\begin{exam}[Decay Rate vs. Noise Strength] \normalfont
Consider the following scenario: two qubits $A$ and $B$ undergo coherent time-evolution under an interaction Hamiltonian $\Ha = Z \otimes X / 2$, while $A$ undergoes depolarizing noise. Here the system's evolution is described by the Lindbladian
\begin{equation} \label{eq:exprl}
\L(\rho) := i [Z \otimes X / 2, \rho] + \lambda_0 (\S_{dep} \otimes \Id^B)(\rho) \pl.
\end{equation}
The fixed point state of the stochastic part, $\lambda_0 (\S_{dep} \otimes \Id^B)(\rho) = - \lambda_0(\id / 2 \otimes \rho^B - \rho)$, is $\E_0(\rho) := \id / 2 \otimes \rho^B$. That of $\L$ is $\E(\rho) = \id / 4$, complete mixture. Here $\E_0(\Ha) = 0$, generating the identity.  Let $\Phi_{ZX(t)}$ denote the unitary generated by $\Ha$ in time $t$.

Numerical results are simulated for time-evolution under $\L$ using Qiskit and Qiskit Dynamics.
\begin{figure}[h!] \scriptsize \centering
	\includegraphics[width=0.64\textwidth]{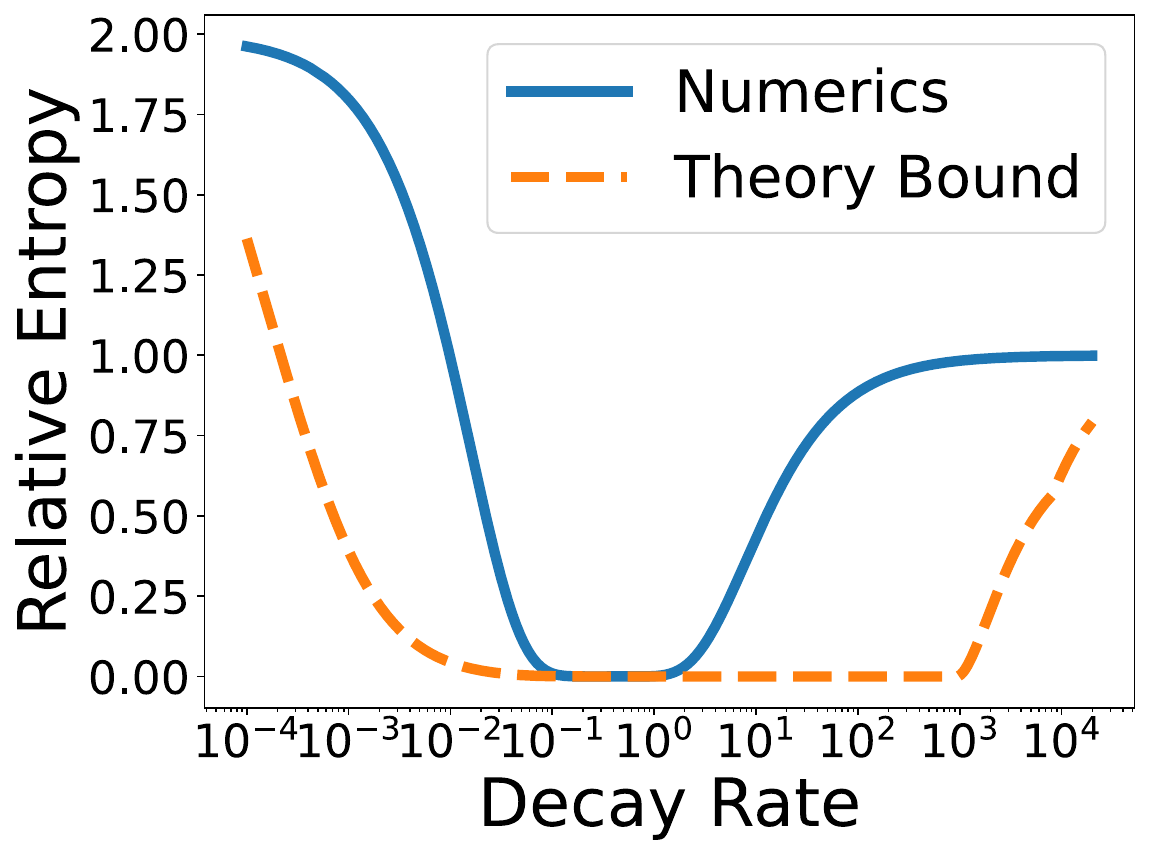}
	\caption{Relative entropy of qubit $B$'s state to the fixed point, which for this input is completely mixed, computed numerically and compared to theoretical lower bounds. The decay rate is $\lambda_0$ as in Equation \eqref{eq:exprl}. The plot is scaled logarithmically on the horizontal axis and linearly on the vertical. \label{fig:simulation}}
\end{figure}
The time is fixed to $t=4 \pi$, so a fully coherent rotation results in perfect fidelity with the original state. An initial state of $\ketbra{00}$ eventually decays to complete mixture. The dissipative term strength $\lambda_0$ varies from 0.0001 to 20000.0. Results appear in Figure \ref{fig:simulation}.

The theoretical bound is from Corollary \ref{cor:ctsfinal}. Note that $\|\cdot \mapsto [Z \otimes X / 2, \cdot] \| = 1$, and $\exp( - \lambda_0 \tau \S_{dep} \otimes \Id^B)(\rho) = \exp(- \lambda_0 \tau) \rho + (1 - \exp(- \lambda_0 \tau)) \E_0(\rho)$. Hence $\tau = (\ln 2) / \lambda_0$ to achieve $\epsilon = 1/2$. For $\L$ as in Equation \eqref{eq:exprl},
\[ \big \| e^{- \L}(\rho) - \E_0 (\rho) \big \|_{1} \leq \frac{14 e^2 \lceil 4 \pi \rceil \ln 2}{\lambda_0} \pl, \]
Via the triangle inequality,
\[ \big \| \Phi^t(\rho) - \E(\rho) \big \|_1 \geq \big \| \E_0(\rho) - \E(\rho) \big \|_1  - \big \| e^{- \L}(\rho) - \E_0 (\rho) \big \|_1 \pl, \]
and $\| \E_0(\rho) - \E(\rho) \big \|_1 \leq 1$ via direct calculation. Since $\E(\rho)$ is completely mixed, $\Phi^t(\rho)$ always commutes with it. Therefore a refinement of Pinkser's inequality for total variation distance given by \cite[Equation (16)]{canonne_short_2023} originally from \cite{vajda_note_1970} applies:
\[ D(\rho \| \E(\rho)) \geq \frac{1}{\ln 2} \Big ( \ln \Big ( \frac{1 + d_{tv}}{1 - d_{tv}} \Big ) - \frac{2 d_{tv}}{1 + d_{tv}} \Big ) \pl, \]
where $d_{tv} = d_{tv}(\rho, \E(\rho))$ is the trace distance for commuting densities or classical variables. As an alternative, Lemma \ref{lem:wsb} yields an upper bound on $| D(\Phi^t(\rho) \| \E(\Phi^t(\rho))) - D(\E_0(\rho) \| \E(\rho)) |$ in terms of $\| e^{- \L}(\rho) - \E_0 (\rho) \|_1$, and $D(\E_0(\rho) \| \E(\rho)) = 1$. Combining the two bounds yields a theoretical lower bound on relative entropy for large $\lambda_0$. For small $\lambda_0$,  a direct bound follows \cite[Lemma 4.3]{laracuente_information_2023}, which states that since $\rho \leq c \E(\rho)$,
\[ D((1 - \zeta) \rho + \zeta \E(\rho) \| \E(\rho)) \geq \sup_\tau \frac{(1 - \zeta)^2}{\tau + \zeta}
    \Big ( 1 - \frac{\tau (1 - \ln \tau)}{\kappa(c)} \Big ) D(\rho \| \E(\rho)) \pl, \]
where $\kappa(c) = (c \ln c - c - 1) / (c-1)^2$. Combining bounds yields the theory curve in Figure \ref{fig:simulation}.

For very small values of $\lambda_0$, the state expectedly does not decay noticeably. With increasing $\lambda$, a region of strong decay appears, where a time of $4 \pi$ is long enough to mostly dephase qubit $B$ under almost continuous interactions with a regularly noised qubit $A$. With large $\lambda_0$, however, the relationship inverts. As the noise channel approaches completely depolarizing, decays slow with increasing $\lambda_0$. Analyzing distinct regimes as in Figure \ref{fig:simulation}:
\begin{itemize}
	\item For small $\lambda_0$, $\Phi^t(\rho) \approx \exp(- i t [\Ha, \cdot])(\rho)$. Relative entropy decays slowly.
	\item The regime of strongest overall decay appears for intermediate values of $\lambda_0$.
	\item When $\lambda_0$ is large and $t$ not too small, $\Phi^t \approx \tilde{R}_t \E_0$ for rotation $\tilde{R}_t$ generated by $\E_0(\Ha)$. This modified Hamiltonian generates Zeno dynamics with the same invariant subspace as $\E_0$ and $\S$.
    Self-restriction as in Theorem \ref{thm:zvscmlsi} protects information in the fixed point subspace of evolution generated by $\S$, which might be larger than the fixed point subspace of evolution generated by $\L$.
\end{itemize}
The theory bounds in this example are \textit{expectedly} loose, as they rely on Theorems designed to apply to all input states, while the specific input state chosen is intended to illustrate effects in a convenient range of values.
\end{exam}

Corollary \ref{cor:ctsfinal} does not immediately follow from \cite{burgarth_generalized_2019}. A subtle but essential difference is that the main Theorems of \cite{burgarth_generalized_2019} control the relative strength of processes through an explicit multiplier, ``$\gamma$." In contrast, Theorem \ref{thm:zvscmlsi} uses the CMLSI constant of the dissipative part of the process. The following example illustrates a scenario in which growth of the CMLSI constant emerges not from an explicit multiplier but from the internal structure of the process:
\begin{exam} \normalfont \label{ex:graphs}
Let $G$ be a finite, undirected graph on $n$ vertices, defined as as a set of pairs $\{i,j\} : i,j \in 1...n$. Let
\[ \Phi_{l,j}(\rho) = \ket{l}\braket{j | \rho | j} \bra{l} + \ket{j} \braket{l | \rho | l} \bra{j}
		+ \Big (\sum_{s \neq l,j} \ketbra{s} \otimes \id^B \Big ) \rho \Big (\sum_{r \neq l,j} \ketbra{r} \otimes \id^B \Big ) \]
represent a single edge on Hilbert space of dimension $n$, with the possibility of extension by an arbitrary auxiliary system with an interaction Hamiltonian $\Ha$. As noted in \cite{laracuente_quasi-factorization_2022}, the complete graph Lindbladian given by
\[ \S_n(\rho) = \frac{1}{2} \sum_{i, j \in 1...n} \Big ( \rho - \Phi_{i,j}(\rho) \Big ) \]
has $O(n)$-CMLSI. Hence via Corollary \ref{cor:ctsfinal}, $\S_n$ has an increasing CMLSI constant with $n$, even though increasing $n$ does not multiply the Lindbladian by a constant larger than 1. Accelerated convergence to the Zeno limit arises because the structure of all-to-all interactions. Hence this family of Lindbladians exhibits growing decay rate that is not captured by an overall factor multiplying the Lindbladian.
\end{exam}

\begin{cexam}[Non-iterability of (C)SPDI] \normalfont \label{cexam:csdpi}
In Section \ref{sec:decayyes}, we considered whether the CSPDI inequality implied by \cite[Theorem 2.6]{gao_complete_2025} generally iterates to imply a discrete-time (C)MLSI-like inequality. Here we show a counterexample. Consider a channel on the $n$-qubit spin chain:
\[ \Phi(\rho) := U_{cyc} ( (1-\epsilon) \frac{\id}{2} \otimes \rho^{(>2)} + \epsilon \rho ) U_{cyc}^\dagger \pl, \]
where $\rho^{(>2)}$ denotes the state on all but the first qubit, and $U_{cyc}$ applies the leftward cyclic qubit permutation $1 \leftarrow 2, 2 \leftarrow 3, ... , n-1 \leftarrow n, n \leftarrow 1$. One may easily see that $\Phi^n$ is a completely depolarizing channel. For any $m \in $ such that $1 < m < n$, the action of $\Phi^m$ on the state $(\id/2)^{\otimes m} \otimes \ketbra{0}^{\otimes n-m}$ is perfectly invertible and equivalent to the leftward cyclic permutation. Hence there is no initial decay to the eventual fixed point subspace.

Examining $\Phi_* \Phi$, the problem is clear: this composition fully depolarizes the first qubit, applies the cyclic permutation, reverses the cyclic permutation, and then applies the same (now trivial) depolarization. Subsequent applications merely continue depolarizing the first qubit. In \cite{gao_complete_2025}, imposing the detailed balance condition excludes channels such as this $\Phi$ by forbidding most forms that include unitary time-evolution. In Theorem \ref{thm:maindecay}, the analytic structure of the QMS precludes the decoherence-free subspace from changing at finite time.
\end{cexam}

\begin{exam}[Cross-Resonance Hamiltonian with Markovian Dephasing] \normalfont \label{exam:qi}
In principle, inequalities such as CMLSI and CSDPI lower bound decoherence rates, as well as quantum capacities and other key notions for quantum information. The imposition of detailed balance is however a substantial restriction, as it precludes studying systems governed by unitary drift or during active quantum processing. Noise during quantum processing or other dynamics is often larger than that during idle times, and the majority of such noise in superconducting qubits often arises from two-qubit interactions \cite[p1026]{murali_noise-adaptive_2019}. In this example, we show how passive dephasing noise during the executation of an entangling gate on superconducting qubits depolarizes the second qubit. We denote the 2-qubit system $A \otimes B$. One entangling gate is via the cross-resonance (CR) interaction between two qubits.  Let
\begin{equation} \label{eq:crcoeffs}
\begin{split}
    & \Ha_0 = \frac{Z}{2} \otimes \big ( \omega_{ZI} \id + \omega_{ZX} X + \omega_{ZY} Y + \omega_{ZZ} Z \big ) \text{ ,}
    \\ & \Ha_1 = \frac{\id}{2} \otimes \big ( \omega_{IX} X + \omega_{IY} Y + \omega_{IZ} Z \big ) \text{ , and }
    \\ & \S(\rho) = \frac{1}{T_2} \big (2 \rho - (Z \otimes \id ) \rho (Z \otimes \id) 
        - (\id \otimes Z) \rho (\id \otimes Z) \big ) \pl.
\end{split}
\end{equation}
The CR Hamiltonian is given by $\Ha_{CR} = \Ha_0 + \Ha_1$
as described in \cite{alexander_qiskit_2020}, where the coefficients $\omega_{??}$ are hardware-specific. Although the gate's precise implementation involves a time-dependent turning on and off of the Hamiltonian, \cite{alexander_qiskit_2020} notes a reasonable approximation is that the unitary gate $U_{CR} \approx \exp(-i t_{CR} \Ha_{CR})$, as the true pulse contains short rise and fall intervals surrounding a longer interval in which the effective Hamiltonian is approximately constant.

Passive noise in qubits is often modeled as a combination of non-unital relaxation with characteristic time denoted $T_1$ and dephasing noise with characteristic time $T_2$, modeled by
\[ \Phi_{deph}^t(\rho) = e^{- t / T_2} \rho + \big ( 1 - e^{- t / T_2} \big ) \E_Z(\rho) \pl, \]
where $\E_Z$ is the pinching map to the computational or Pauli $Z$ basis \cite{krantz_quantum_2019}. When $T_1 >> T_2$, we may consider the two qubits in the CR gate to undergo individual dephasing noise simultaneous with the coupling gate. This process is modeled by the Lindbladian $\L(\rho) := i [\Ha_{CR}, \rho] + \S(\rho)$ generating the QMS $\Phi^t_{CR}$.
Though the decoherence rate will depend on specifically on the coefficients in Equation \eqref{eq:crcoeffs}, most combinations will yield that time-evolution under $H_{CR}$ does not commute with $\Id \otimes \E_Z$, though it does commute with $\E_Z \otimes \Id$. Via Lemma \ref{lem:converge2}, the Hamiltnoian must therefore reduce the fixed point subspace. In particular, we see that as long as $\omega_{IX}$ or $\omega_{IY}$ is non-zero, the Lindbladian given by $\cdot \mapsto i[\Ha_1, \cdot] + \S$ has fixed point projection $\E$ given by $\E(\rho) := \E_Z(\rho^A) \otimes \id / 2$. Because $\Ha_0$ applies only $Z$-basis coupling to qubit $A$, it does not change the fixed point algebra on that qubit. Since the fixed point subspace must be smaller than that projected to be $\E_Z \otimes \E_Z$ and remain unital, it must be that projected to by $\E$.

Theorem \ref{thm:zvscmlsi} shows that this noisy cross-resonance Lindbladian does not have CMLSI. Furthermore, if the noise were overpoweringly strong, an input $\rho$ would quickly decay to $(\E_Z \otimes \E_Z)(\rho)$ and slowly to $\E_Z(\rho^A) \otimes \id/2$. Theorem \ref{thm:maindecay} nonetheless shows that for any fixed $\tau > 0$,
\[ D(\Phi_{CR}^\tau(\rho) \| \E_{Z}(\rho^A) \otimes \id / 2) \leq \lambda_\tau D( \rho \| \E_{Z}(\rho^A) \otimes \id / 2) \pl. \]
for some $\lambda_\tau < 1$ that is independent of $\rho$. A long, noisy cross-resonance interaction or repeated cross-resonance gate application with dephasing noise does result in exponential decay towards a fixed point subspace in which qubit $B$ is fully depolarized.
\end{exam}
A long-term motivation to study noise is understanding decoherence affecting quantum computers, memories, communication links, and related systems. Though fault-tolerant error correction presents an eventual solution to noise in quantum computation, there are still major challenges in the near term \cite{preskill_quantum_2018, aharonov_polynomial-time_2023}. In the long term, noise will still contribute to the number of physical qubits required to implement fault-tolerant error correction, so understanding and reducing noise will help reduce the cost and energy demands of fault-tolerant quantum computers. Furthermore, noise in long-distance communication links and long-term memory will not simply disappear with error correction but impact error-corrected capacities \cite{wilde_quantum_2017}. As in Example \ref{exam:qi}, extending CMLSI-like bounds to situations that combine unitary dynamics with noise is a step toward modeling systems that often appear in quantum information, including quantum processing. There are several relevant issues to address in future work, including time-dependent Hamiltonians and how they combine with non-unital noise, in order to complete the picture of CMLSI for pulsed gates with Markovian noise. In full generality, one likely must consider non-Markovian dynamics to fully understand the noise arising during quantum gates.



\section{Conclusions and Further Open Problems} \label{sec:conclusions}
As Theorem \ref{thm:zvscmlsi}, conditions were determined for when adding a Hamiltonian term to a trace-symmetric Lindbladian generator breaks or preserves (C)MLSI. A point we re-emphasize is that unlike in dynamical decoupling, adding a Hamiltonian to unital, Markovian noise does not slow the growth of entropy. Rather, CMLSI typically fails because the decoherence-free subspace of the new Lindbladian is smaller than that of the original, so the original dissipator no longer induces the full range of eventual decay.

A long-running problem in semigroup theory is to understand the most general class of canonical dissipator \cite{hayden_canonical_2022}. In this paper, we focus on how adding a Hamiltonian affects decay rate. It remains open whether there exist canonically dissipative (having no extractable, non-trivial Hamiltonian part), finite-dimensional generators that do not decay states to a fixed point subspace or do so slower than exponentially.

Left open by Theorem \ref{thm:maindecay} is whether an analogous decay estimate holds for Lindbladians of the form in Equation \eqref{eq:ldecomp}, but where $\S$ has (C)MLSI with a \textit{non-unital} $\E_0$. In this scenario, the fixed point subspace projected to be $\E$ is not obviously contained in that by $\E_0$, or the full time-evolution may exhibit limit cycles \cite{walter_quantum_2014, ben_arosh_quantum_2021}. The Hamiltonian might more directly counteract effects of the dissipator. Mathematically, the primary barrier to extending Theorem \ref{thm:maindecay} beyond unital channels is to generalize properties of the multiplicative domain to the non-unital case. For example, an appropriate generalization of \cite[Theorem 1.3]{rahaman_multiplicative_2017} would strongly hint at the proper definition of a fixed point up to rotation in the non-unital case. As noted in \cite{choi_multiplicative_2009}, however, the connection between the multiplicative domain, unitarily correctable subspace, and fixed point subspace of $\Phi \circ \Phi^*$ for a channel $\Phi$ is weaker without unitality. Deriving and applying an analog of \cite[Theorem 1.3]{rahaman_multiplicative_2017} without this correspondence might be the topic for another work. We therefore leave this generalization to future investigation.

Relatedly, another natural question is whether all QMSs decay states to a fixed point subspace up to persistent rotations. In this work, we exploited a connection to the multiplicative domain for which more in the unital channel setting. Physically, the restriction to unitality in Theorem \ref{thm:maindecay} ensures that coherent dynamics ultimately commute with projection to the overall fixed point subspace, so the fixed point subspace of the dissipative component is always at least as large as that of the combined process. Such an assumption is potentially violated when combining a non-unital dissipator with Hamiltonian time-evolution. Furthermore, initial relative entropy is unbounded for fixed point states with arbitrarily small components. These complications naturally extend to those encountered with non-Markovian processes. In full generality, time-varying unitary dynamics might preclude an analogous notion of a fixed point or decoherence-free subspace, requiring a different formulation of decay.

 \section{Acknowledgements}
We thank John Smolin and Marius Junge for helpful feedback. NL was previously supported by the IBM Postdoc program at the Chicago Quantum Exchange. 

\section{Statement on Competing Interests}
The author declares no competing interests.

\section{Data and Code Availability}
All code and data that support the plots within this paper and other findings of this study are available from the corresponding author upon reasonable request.

\bibliography{Drift}

\begin{thebibliography}{82}
\providecommand{\natexlab}[1]{#1}
\providecommand{\url}[1]{\texttt{#1}}
\expandafter\ifx\csname urlstyle\endcsname\relax
  \providecommand{\doi}[1]{doi: #1}\else
  \providecommand{\doi}{doi: \begingroup \urlstyle{rm}\Url}\fi

\bibitem[Clerk et~al.(2010)Clerk, Devoret, Girvin, Marquardt, and
  Schoelkopf]{clerk_introduction_2010}
A.~A. Clerk, M.~H. Devoret, S.~M. Girvin, Florian Marquardt, and R.~J.
  Schoelkopf.
\newblock Introduction to quantum noise, measurement, and amplification.
\newblock \emph{Reviews of Modern Physics}, 82\penalty0 (2):\penalty0
  1155--1208, April 2010.
\newblock \doi{10.1103/RevModPhys.82.1155}.
\newblock URL \url{https://link.aps.org/doi/10.1103/RevModPhys.82.1155}.
\newblock Publisher: American Physical Society.

\bibitem[Preskill(2018)]{preskill_quantum_2018}
John Preskill.
\newblock Quantum {Computing} in the {NISQ} era and beyond.
\newblock \emph{Quantum}, 2:\penalty0 79, August 2018.
\newblock ISSN 2521-327X.
\newblock \doi{10.22331/q-2018-08-06-79}.
\newblock URL \url{https://quantum-journal.org/papers/q-2018-08-06-79/}.

\bibitem[Wilde(2017)]{wilde_quantum_2017}
Mark~M. Wilde.
\newblock \emph{Quantum {Information} {Theory}}.
\newblock Cambridge University Press, Cambridge, 2 edition, 2017.
\newblock ISBN 978-1-107-17616-4.
\newblock \doi{10.1017/9781316809976}.
\newblock URL
  \url{https://www.cambridge.org/core/books/quantum-information-theory/247A740E156416531AA8CB97DFDAE438}.

\bibitem[Rebolledo(2005)]{rebolledo_decoherence_2005}
R~Rebolledo.
\newblock Decoherence of quantum {Markov} semigroups.
\newblock \emph{Annales de l'Institut Henri Poincare (B) Probability and
  Statistics}, 41\penalty0 (3):\penalty0 349--373, May 2005.
\newblock ISSN 02460203.
\newblock \doi{10.1016/j.anihpb.2004.12.003}.
\newblock URL
  \url{https://linkinghub.elsevier.com/retrieve/pii/S0246020305000257}.

\bibitem[Bardet et~al.(2023)Bardet, Capel, Gao, Lucia, P{\'e}rez-Garc{\'i}a,
  and Rouz{\'e}]{bardet_rapid_2023}
Ivan Bardet, {\'A}ngela Capel, Li~Gao, Angelo Lucia, David
  P{\'e}rez-Garc{\'i}a, and Cambyse Rouz{\'e}.
\newblock Rapid {Thermalization} of {Spin} {Chain} {Commuting} {Hamiltonians}.
\newblock \emph{Physical Review Letters}, 130\penalty0 (6):\penalty0 060401,
  February 2023.
\newblock \doi{10.1103/PhysRevLett.130.060401}.
\newblock URL \url{https://link.aps.org/doi/10.1103/PhysRevLett.130.060401}.
\newblock Publisher: American Physical Society.

\bibitem[Bluhm et~al.(2022)Bluhm, Capel, and
  P{\'e}rez-Hern{\'a}ndez]{bluhm_exponential_2022}
Andreas Bluhm, {\'A}ngela Capel, and Antonio P{\'e}rez-Hern{\'a}ndez.
\newblock Exponential decay of mutual information for {Gibbs} states of local
  {Hamiltonians}.
\newblock \emph{Quantum}, 6:\penalty0 650, February 2022.
\newblock \doi{10.22331/q-2022-02-10-650}.
\newblock URL \url{https://quantum-journal.org/papers/q-2022-02-10-650/}.
\newblock Publisher: Verein zur F{\"o}rderung des Open Access Publizierens in
  den Quantenwissenschaften.

\bibitem[Lindblad(1976)]{lindblad_generators_1976}
G.~Lindblad.
\newblock On the generators of quantum dynamical semigroups.
\newblock \emph{Communications in Mathematical Physics}, 48\penalty0
  (2):\penalty0 119--130, June 1976.
\newblock ISSN 0010-3616, 1432-0916.
\newblock \doi{10.1007/BF01608499}.
\newblock URL \url{http://link.springer.com/10.1007/BF01608499}.

\bibitem[Gorini et~al.(1976)Gorini, Kossakowski, and
  Sudarshan]{gorini_completely_1976}
Vittorio Gorini, Andrzej Kossakowski, and E.~C.~G. Sudarshan.
\newblock Completely positive dynamical semigroups of {N}-level systems.
\newblock \emph{Journal of Mathematical Physics}, 17\penalty0 (5):\penalty0
  821--825, 1976.
\newblock ISSN 0022-2488.
\newblock \doi{10.1063/1.522979}.
\newblock URL \url{https://doi.org/10.1063/1.522979}.

\bibitem[Umegaki(1962)]{umegaki_conditional_1962}
Hisaharu Umegaki.
\newblock Conditional expectation in an operator algebra. {IV}. {Entropy} and
  information.
\newblock \emph{Kodai Mathematical Seminar Reports}, 14\penalty0 (2):\penalty0
  59--85, 1962.
\newblock ISSN 0023-2599.
\newblock \doi{10.2996/kmj/1138844604}.
\newblock URL \url{https://projecteuclid.org/euclid.kmj/1138844604}.

\bibitem[Gao and Rouz{\'e}(2022)]{gao_complete_2022}
Li~Gao and Cambyse Rouz{\'e}.
\newblock Complete {Entropic} {Inequalities} for {Quantum} {Markov} {Chains}.
\newblock \emph{Archive for Rational Mechanics and Analysis}, 245\penalty0
  (1):\penalty0 183--238, July 2022.
\newblock ISSN 1432-0673.
\newblock \doi{10.1007/s00205-022-01785-1}.
\newblock URL \url{https://doi.org/10.1007/s00205-022-01785-1}.

\bibitem[Gao et~al.(2021)Gao, Junge, and Li]{gao_geometric_2021}
Li~Gao, Marius Junge, and Haojian Li.
\newblock Geometric {Approach} {Towards} {Complete} {Logarithmic} {Sobolev}
  {Inequalities}.
\newblock February 2021.
\newblock \doi{10.48550/arXiv.2102.04434}.
\newblock URL \url{http://arxiv.org/abs/2102.04434}.
\newblock arXiv:2102.04434 [quant-ph].

\bibitem[Junge et~al.(2022)Junge, Laracuente, and
  Rouz{\'e}]{junge_stability_2022}
Marius Junge, Nicholas Laracuente, and Cambyse Rouz{\'e}.
\newblock Stability of {Logarithmic} {Sobolev} {Inequalities} {Under} a
  {Noncommutative} {Change} of {Measure}.
\newblock \emph{Journal of Statistical Physics}, 190\penalty0 (2):\penalty0 30,
  December 2022.
\newblock ISSN 1572-9613.
\newblock \doi{10.1007/s10955-022-03026-x}.
\newblock URL \url{https://doi.org/10.1007/s10955-022-03026-x}.

\bibitem[Gao et~al.(2020{\natexlab{a}})Gao, Junge, and
  LaRacuente]{gao_fisher_2020}
Li~Gao, Marius Junge, and Nicholas LaRacuente.
\newblock Fisher {Information} and {Logarithmic} {Sobolev} {Inequality} for
  {Matrix}-{Valued} {Functions}.
\newblock \emph{Annales Henri Poincar{\'e}}, 21\penalty0 (11):\penalty0
  3409--3478, November 2020{\natexlab{a}}.
\newblock ISSN 1424-0661.
\newblock \doi{10.1007/s00023-020-00947-9}.
\newblock URL \url{https://doi.org/10.1007/s00023-020-00947-9}.

\bibitem[Kastoryano and Temme(2013)]{kastoryano_quantum_2013}
Michael~J. Kastoryano and Kristan Temme.
\newblock Quantum logarithmic {Sobolev} inequalities and rapid mixing.
\newblock \emph{Journal of Mathematical Physics}, 54\penalty0 (5):\penalty0
  052202, May 2013.
\newblock ISSN 0022-2488.
\newblock \doi{10.1063/1.4804995}.
\newblock URL \url{https://doi.org/10.1063/1.4804995}.

\bibitem[Dhahri et~al.(2010)Dhahri, Fagnola, and
  Rebolledo]{dhahri_decoherence-free_2010}
Ameur Dhahri, Franco Fagnola, and Rolando Rebolledo.
\newblock {The} {Decoherence}-{Free} {Subalgebra} {of} {a} {Quantum} {Markov}
  {Semigroup} {with} {Unbounded} {Generator}.
\newblock \emph{Infinite Dimensional Analysis, Quantum Probability and Related
  Topics}, 13\penalty0 (03):\penalty0 413--433, September 2010.
\newblock ISSN 0219-0257, 1793-6306.
\newblock \doi{10.1142/S0219025710004176}.
\newblock URL
  \url{https://www.worldscientific.com/doi/abs/10.1142/S0219025710004176}.

\bibitem[Fagnola et~al.(2019)Fagnola, Sasso, and Umanità]{fagnola_role_2019}
Franco Fagnola, Emanuela Sasso, and Veronica Umanità.
\newblock The role of the atomic decoherence-free subalgebra in the study of
  quantum {Markov} semigroups.
\newblock \emph{Journal of Mathematical Physics}, 60\penalty0 (7):\penalty0
  072703, July 2019.
\newblock ISSN 0022-2488, 1089-7658.
\newblock \doi{10.1063/1.5030954}.
\newblock URL
  \url{https://pubs.aip.org/jmp/article/60/7/072703/319785/The-role-of-the-atomic-decoherence-free-subalgebra}.

\bibitem[Popkov et~al.(2018)Popkov, Essink, Presilla, and
  Schütz]{popkov_effective_2018}
Vladislav Popkov, Simon Essink, Carlo Presilla, and Gunter Schütz.
\newblock Effective quantum {Zeno} dynamics in dissipative quantum systems.
\newblock \emph{Physical Review A}, 98\penalty0 (5):\penalty0 052110, November
  2018.
\newblock ISSN 2469-9926, 2469-9934.
\newblock \doi{10.1103/PhysRevA.98.052110}.
\newblock URL \url{https://link.aps.org/doi/10.1103/PhysRevA.98.052110}.

\bibitem[Gao et~al.(2020{\natexlab{b}})Gao, Junge, and
  LaRacuente]{gao_relative_2020}
Li~Gao, Marius Junge, and Nicholas LaRacuente.
\newblock Relative entropy for von {Neumann} subalgebras.
\newblock \emph{International Journal of Mathematics}, 31\penalty0
  (06):\penalty0 2050046, June 2020{\natexlab{b}}.
\newblock ISSN 0129-167X.
\newblock \doi{10.1142/S0129167X20500469}.
\newblock URL
  \url{https://www.worldscientific.com/doi/abs/10.1142/S0129167X20500469}.
\newblock Publisher: World Scientific Publishing Co.

\bibitem[Pimsner and Popa(1986)]{pimsner_entropy_1986}
Mihai Pimsner and Sorin Popa.
\newblock Entropy and index for subfactors.
\newblock \emph{Annales scientifiques de l'{\'E}cole Normale Sup{\'e}rieure},
  19\penalty0 (1):\penalty0 57--106, 1986.
\newblock ISSN 1873-2151.
\newblock \doi{10.24033/asens.1504}.
\newblock URL \url{http://www.numdam.org/item/?id=ASENS_1986_4_19_1_57_0}.

\bibitem[Jones(1983)]{jones_index_1983}
V.~F.~R. Jones.
\newblock Index for subfactors.
\newblock \emph{Inventiones mathematicae}, 72\penalty0 (1):\penalty0 1--25,
  February 1983.
\newblock ISSN 1432-1297.
\newblock \doi{10.1007/BF01389127}.
\newblock URL \url{https://doi.org/10.1007/BF01389127}.

\bibitem[Lindblad(1975)]{lindblad_completely_1975}
G{\"o}ran Lindblad.
\newblock Completely positive maps and entropy inequalities.
\newblock \emph{Communications in Mathematical Physics}, 40\penalty0
  (2):\penalty0 147--151, June 1975.
\newblock ISSN 1432-0916.
\newblock \doi{10.1007/BF01609396}.
\newblock URL \url{https://doi.org/10.1007/BF01609396}.

\bibitem[Winter(2016)]{winter_tight_2016}
Andreas Winter.
\newblock Tight {Uniform} {Continuity} {Bounds} for {Quantum} {Entropies}:
  {Conditional} {Entropy}, {Relative} {Entropy} {Distance} and {Energy}
  {Constraints}.
\newblock \emph{Communications in Mathematical Physics}, 347\penalty0
  (1):\penalty0 291--313, October 2016.
\newblock ISSN 1432-0916.
\newblock \doi{10.1007/s00220-016-2609-8}.
\newblock URL \url{https://doi.org/10.1007/s00220-016-2609-8}.

\bibitem[Bluhm et~al.(2023)Bluhm, Capel, Gondolf, and
  Pérez-Hernández]{bluhm_continuity_2023}
Andreas Bluhm, Ángela Capel, Paul Gondolf, and Antonio Pérez-Hernández.
\newblock Continuity of {Quantum} {Entropic} {Quantities} via {Almost}
  {Convexity}.
\newblock \emph{IEEE Transactions on Information Theory}, 69\penalty0
  (9):\penalty0 5869--5901, September 2023.
\newblock ISSN 0018-9448, 1557-9654.
\newblock \doi{10.1109/TIT.2023.3277892}.
\newblock URL \url{https://ieeexplore.ieee.org/document/10129917/}.

\bibitem[Bluhm et~al.(2024)Bluhm, Capel, Gondolf, and
  Pérez-Hernández]{bluhm_corrections_2024}
Andreas Bluhm, Ángela Capel, Paul Gondolf, and Antonio Pérez-Hernández.
\newblock Corrections to “{Continuity} of {Quantum} {Entropic} {Quantities}
  via {Almost} {Convexity}”.
\newblock \emph{IEEE Transactions on Information Theory}, 70\penalty0
  (7):\penalty0 5410--5411, July 2024.
\newblock ISSN 0018-9448, 1557-9654.
\newblock \doi{10.1109/TIT.2024.3390839}.
\newblock URL \url{https://ieeexplore.ieee.org/document/10504886/}.

\bibitem[Lin(1991)]{lin_divergence_1991}
J.~Lin.
\newblock Divergence measures based on the {Shannon} entropy.
\newblock \emph{IEEE Transactions on Information Theory}, 37\penalty0
  (1):\penalty0 145--151, January 1991.
\newblock ISSN 00189448.
\newblock \doi{10.1109/18.61115}.
\newblock URL \url{http://ieeexplore.ieee.org/document/61115/}.

\bibitem[Hastings(2009)]{hastings_superadditivity_2009}
M.~B. Hastings.
\newblock Superadditivity of communication capacity using entangled inputs.
\newblock \emph{Nature Physics}, 5\penalty0 (4):\penalty0 255--257, April 2009.
\newblock ISSN 1745-2481.
\newblock \doi{10.1038/nphys1224}.
\newblock URL \url{https://www.nature.com/articles/nphys1224}.
\newblock Number: 4 Publisher: Nature Publishing Group.

\bibitem[Petz(1988)]{petz_sufficiency_1988}
D{\'e}nes Petz.
\newblock Sufficiency of {Channels} {Over} {Von} {Neumann} {Algebras}.
\newblock \emph{The Quarterly Journal of Mathematics}, 39\penalty0
  (1):\penalty0 97--108, March 1988.
\newblock ISSN 0033-5606.
\newblock \doi{10.1093/qmath/39.1.97}.
\newblock URL \url{https://doi.org/10.1093/qmath/39.1.97}.

\bibitem[Schlosshauer(2019)]{schlosshauer_quantum_2019}
Maximilian Schlosshauer.
\newblock Quantum decoherence.
\newblock \emph{Physics Reports}, 831:\penalty0 1--57, October 2019.
\newblock ISSN 0370-1573.
\newblock \doi{10.1016/j.physrep.2019.10.001}.
\newblock URL
  \url{https://www.sciencedirect.com/science/article/pii/S0370157319303084}.

\bibitem[Suter and {\'A}lvarez(2016)]{suter_colloquium_2016}
Dieter Suter and Gonzalo~A. {\'A}lvarez.
\newblock Colloquium: {Protecting} quantum information against environmental
  noise.
\newblock \emph{Reviews of Modern Physics}, 88\penalty0 (4):\penalty0 041001,
  October 2016.
\newblock \doi{10.1103/RevModPhys.88.041001}.
\newblock URL \url{https://link.aps.org/doi/10.1103/RevModPhys.88.041001}.
\newblock Publisher: American Physical Society.

\bibitem[Beau et~al.(2017)Beau, Kiukas, Egusquiza, and del
  Campo]{beau_nonexponential_2017}
M.~Beau, J.~Kiukas, I.~L. Egusquiza, and A.~del Campo.
\newblock Nonexponential {Quantum} {Decay} under {Environmental} {Decoherence}.
\newblock \emph{Physical Review Letters}, 119\penalty0 (13):\penalty0 130401,
  September 2017.
\newblock \doi{10.1103/PhysRevLett.119.130401}.
\newblock URL \url{https://link.aps.org/doi/10.1103/PhysRevLett.119.130401}.
\newblock Publisher: American Physical Society.

\bibitem[Aloisio et~al.(2023)Aloisio, White, Hill, and
  Modi]{aloisio_sampling_2023}
I.A. Aloisio, G.A.L. White, C.D. Hill, and K.~Modi.
\newblock Sampling {Complexity} of {Open} {Quantum} {Systems}.
\newblock \emph{PRX Quantum}, 4\penalty0 (2):\penalty0 020310, April 2023.
\newblock \doi{10.1103/PRXQuantum.4.020310}.
\newblock URL \url{https://link.aps.org/doi/10.1103/PRXQuantum.4.020310}.
\newblock Publisher: American Physical Society.

\bibitem[Hayden and Sorce(2022)]{hayden_canonical_2022}
Patrick Hayden and Jonathan Sorce.
\newblock A canonical {Hamiltonian} for open quantum systems.
\newblock \emph{Journal of Physics A: Mathematical and Theoretical},
  55\penalty0 (22):\penalty0 225302, May 2022.
\newblock ISSN 1751-8121.
\newblock \doi{10.1088/1751-8121/ac65c2}.
\newblock URL \url{https://dx.doi.org/10.1088/1751-8121/ac65c2}.
\newblock Publisher: IOP Publishing.

\bibitem[Merkli(2022)]{merkli_dynamics_2022-1}
Marco Merkli.
\newblock Dynamics of {Open} {Quantum} {Systems} {II}, {Markovian}
  {Approximation}.
\newblock \emph{Quantum}, 6:\penalty0 616, January 2022.
\newblock \doi{10.22331/q-2022-01-03-616}.
\newblock URL \url{https://quantum-journal.org/papers/q-2022-01-03-616/}.
\newblock Publisher: Verein zur F{\"o}rderung des Open Access Publizierens in
  den Quantenwissenschaften.

\bibitem[Arnold et~al.(1998)Arnold, Markowich, Toscani, and
  Unterreiter]{arnold_logarithmic_1998}
Anton Arnold, Peter Markowich, Giuseppe Toscani, and Andreas Unterreiter.
\newblock On logarithmic {Sobolev} inequalities, {Csiszar}-{Kullback}
  inequalities, and the rate of convergence to equilibrium for
  {Fokker}-{Planck} type equations.
\newblock \emph{Preprint-Reihe Mathematik}, 592:\penalty0 77, 1998.

\bibitem[Bobkov and Tetali(2003)]{bobkov_modified_2003}
Sergey Bobkov and Prasad Tetali.
\newblock Modified log-sobolev inequalities, mixing and hypercontractivity.
\newblock In \emph{Proceedings of the thirty-fifth annual {ACM} symposium on
  {Theory} of computing}, {STOC} '03, pages 287--296, New York, NY, USA, June
  2003. Association for Computing Machinery.
\newblock ISBN 978-1-58113-674-6.
\newblock \doi{10.1145/780542.780586}.
\newblock URL \url{https://dl.acm.org/doi/10.1145/780542.780586}.

\bibitem[Gross(1975{\natexlab{a}})]{gross_logarithmic_1975}
Leonard Gross.
\newblock Logarithmic {Sobolev} {Inequalities}.
\newblock \emph{American Journal of Mathematics}, 97\penalty0 (4):\penalty0
  1061--1083, 1975{\natexlab{a}}.
\newblock ISSN 0002-9327.
\newblock \doi{10.2307/2373688}.
\newblock URL \url{https://www.jstor.org/stable/2373688}.
\newblock Publisher: Johns Hopkins University Press.

\bibitem[Gross(1975{\natexlab{b}})]{gross_hypercontractivity_1975}
Leonard Gross.
\newblock Hypercontractivity and logarithmic {Sobolev} inequalities for the
  {Clifford}-{Dirichlet} form.
\newblock \emph{Duke Mathematical Journal}, 42\penalty0 (3):\penalty0 383--396,
  September 1975{\natexlab{b}}.
\newblock ISSN 0012-7094, 1547-7398.
\newblock \doi{10.1215/S0012-7094-75-04237-4}.
\newblock URL
  \url{https://projecteuclid.org/journals/duke-mathematical-journal/volume-42/issue-3/Hypercontractivity-and-logarithmic-Sobolev-inequalities-for-the-Clifford-Dirichlet-form/10.1215/S0012-7094-75-04237-4.full}.

\bibitem[Bardet and Rouz{\'e}(2022)]{bardet_hypercontractivity_2022}
Ivan Bardet and Cambyse Rouz{\'e}.
\newblock Hypercontractivity and {Logarithmic} {Sobolev} {Inequality} for
  {Non}-primitive {Quantum} {Markov} {Semigroups} and {Estimation} of
  {Decoherence} {Rates}.
\newblock \emph{Annales Henri Poincar{\'e}}, 23\penalty0 (11):\penalty0
  3839--3903, November 2022.
\newblock ISSN 1424-0661.
\newblock \doi{10.1007/s00023-022-01196-8}.
\newblock URL \url{https://doi.org/10.1007/s00023-022-01196-8}.

\bibitem[Bardet(2017)]{bardet_estimating_2017}
Ivan Bardet.
\newblock Estimating the decoherence time using non-commutative {Functional}
  {Inequalities}.
\newblock October 2017.
\newblock \doi{10.48550/arXiv.1710.01039}.
\newblock URL \url{http://arxiv.org/abs/1710.01039}.
\newblock arXiv:1710.01039 [math-ph, physics:quant-ph].

\bibitem[Bardet et~al.(2021)Bardet, Junge, Laracuente, Rouz{\'e}, and Fran{\c
  c}a]{bardet_group_2021}
Ivan Bardet, Marius Junge, Nicholas Laracuente, Cambyse Rouz{\'e}, and
  Daniel~Stilck Fran{\c c}a.
\newblock Group {Transference} {Techniques} for the {Estimation} of the
  {Decoherence} {Times} and {Capacities} of {Quantum} {Markov} {Semigroups}.
\newblock \emph{IEEE Transactions on Information Theory}, 67\penalty0
  (5):\penalty0 2878--2909, May 2021.
\newblock ISSN 1557-9654.
\newblock \doi{10.1109/TIT.2021.3065452}.
\newblock Conference Name: IEEE Transactions on Information Theory.

\bibitem[Stilck~Fran{\c c}a and
  Garc{\'i}a-Patr{\'o}n(2021)]{stilck_franca_limitations_2021}
Daniel Stilck~Fran{\c c}a and Raul Garc{\'i}a-Patr{\'o}n.
\newblock Limitations of optimization algorithms on noisy quantum devices.
\newblock \emph{Nature Physics}, 17\penalty0 (11):\penalty0 1221--1227,
  November 2021.
\newblock ISSN 1745-2481.
\newblock \doi{10.1038/s41567-021-01356-3}.
\newblock URL \url{https://www.nature.com/articles/s41567-021-01356-3}.
\newblock Number: 11 Publisher: Nature Publishing Group.

\bibitem[Gao et~al.(2025)Gao, Junge, LaRacuente, and Li]{gao_complete_2025}
Li~Gao, Marius Junge, Nicholas LaRacuente, and Haojian Li.
\newblock Complete positivity order and relative entropy decay.
\newblock \emph{Forum of Mathematics, Sigma}, 13:\penalty0 e31, January 2025.
\newblock ISSN 2050-5094.
\newblock \doi{10.1017/fms.2024.117}.
\newblock URL
  \url{https://www.cambridge.org/core/journals/forum-of-mathematics-sigma/article/complete-positivity-order-and-relative-entropy-decay/7270B6BBA7891B612C00FAB066F1A380}.

\bibitem[Bardet et~al.(2024)Bardet, Capel, Gao, Lucia, Pérez-García, and
  Rouzé]{bardet_entropy_2024}
Ivan Bardet, Ángela Capel, Li~Gao, Angelo Lucia, David Pérez-García, and
  Cambyse Rouzé.
\newblock Entropy {Decay} for {Davies} {Semigroups} of a {One} {Dimensional}
  {Quantum} {Lattice}.
\newblock \emph{Communications in Mathematical Physics}, 405\penalty0
  (2):\penalty0 42, February 2024.
\newblock ISSN 1432-0916.
\newblock \doi{10.1007/s00220-023-04869-5}.
\newblock URL \url{https://doi.org/10.1007/s00220-023-04869-5}.

\bibitem[Chen and Brand{\~a}o(2023)]{chen_fast_2023}
Chi-Fang Chen and Fernando G. S.~L. Brand{\~a}o.
\newblock Fast {Thermalization} from the {Eigenstate} {Thermalization}
  {Hypothesis}.
\newblock March 2023.
\newblock \doi{10.48550/arXiv.2112.07646}.
\newblock URL \url{http://arxiv.org/abs/2112.07646}.
\newblock arXiv:2112.07646 [cond-mat, physics:math-ph, physics:quant-ph].

\bibitem[LaRacuente(2022)]{laracuente_quasi-factorization_2022}
Nicholas LaRacuente.
\newblock Quasi-factorization and multiplicative comparison of
  subalgebra-relative entropy.
\newblock \emph{Journal of Mathematical Physics}, 63\penalty0 (12):\penalty0
  122203, December 2022.
\newblock ISSN 0022-2488.
\newblock \doi{10.1063/5.0053698}.
\newblock URL \url{https://doi.org/10.1063/5.0053698}.

\bibitem[Kribs and Spekkens(2006)]{kribs_quantum_2006}
David~W. Kribs and Robert~W. Spekkens.
\newblock Quantum error-correcting subsystems are unitarily recoverable
  subsystems.
\newblock \emph{Physical Review A}, 74\penalty0 (4):\penalty0 042329, October
  2006.
\newblock ISSN 1050-2947, 1094-1622.
\newblock \doi{10.1103/PhysRevA.74.042329}.
\newblock URL \url{https://link.aps.org/doi/10.1103/PhysRevA.74.042329}.

\bibitem[Choi et~al.(2009)Choi, Johnston, and Kribs]{choi_multiplicative_2009}
Man-Duen Choi, Nathaniel Johnston, and David~W Kribs.
\newblock The multiplicative domain in quantum error correction.
\newblock \emph{Journal of Physics A: Mathematical and Theoretical},
  42\penalty0 (24):\penalty0 245303, June 2009.
\newblock ISSN 1751-8113, 1751-8121.
\newblock \doi{10.1088/1751-8113/42/24/245303}.
\newblock URL
  \url{https://iopscience.iop.org/article/10.1088/1751-8113/42/24/245303}.

\bibitem[Johnston and Kribs(2011)]{johnston_generalized_2011}
Nathaniel Johnston and David Kribs.
\newblock Generalized multiplicative domains and quantum error correction.
\newblock \emph{Proceedings of the American Mathematical Society}, 139\penalty0
  (2):\penalty0 627--639, February 2011.
\newblock ISSN 0002-9939, 1088-6826.
\newblock \doi{10.1090/S0002-9939-2010-10556-7}.
\newblock URL
  \url{https://www.ams.org/proc/2011-139-02/S0002-9939-2010-10556-7/}.

\bibitem[Rahaman(2017)]{rahaman_multiplicative_2017}
Mizanur Rahaman.
\newblock Multiplicative properties of quantum channels.
\newblock \emph{Journal of Physics A: Mathematical and Theoretical},
  50\penalty0 (34):\penalty0 345302, August 2017.
\newblock ISSN 1751-8113, 1751-8121.
\newblock \doi{10.1088/1751-8121/aa7b57}.
\newblock URL
  \url{https://iopscience.iop.org/article/10.1088/1751-8121/aa7b57}.

\bibitem[Lidar et~al.(1998)Lidar, Chuang, and
  Whaley]{lidar_decoherence-free_1998}
D.~A. Lidar, I.~L. Chuang, and K.~B. Whaley.
\newblock Decoherence-{Free} {Subspaces} for {Quantum} {Computation}.
\newblock \emph{Physical Review Letters}, 81\penalty0 (12):\penalty0
  2594--2597, September 1998.
\newblock ISSN 0031-9007, 1079-7114.
\newblock \doi{10.1103/PhysRevLett.81.2594}.
\newblock URL \url{https://link.aps.org/doi/10.1103/PhysRevLett.81.2594}.

\bibitem[Agredo et~al.(2014)Agredo, Fagnola, and
  Rebolledo]{agredo_decoherence_2014}
Julián Agredo, Franco Fagnola, and Rolando Rebolledo.
\newblock Decoherence free subspaces of a quantum {Markov} semigroup.
\newblock \emph{Journal of Mathematical Physics}, 55\penalty0 (11):\penalty0
  112201, November 2014.
\newblock ISSN 0022-2488, 1089-7658.
\newblock \doi{10.1063/1.4901009}.
\newblock URL
  \url{https://pubs.aip.org/jmp/article/55/11/112201/400172/Decoherence-free-subspaces-of-a-quantum-Markov}.

\bibitem[Lidar(2014)]{lidar_review_2014}
Daniel~A. Lidar.
\newblock Review of {Decoherence}-{Free} {Subspaces}, {Noiseless} {Subsystems},
  and {Dynamical} {Decoupling}.
\newblock In \emph{Quantum {Information} and {Computation} for {Chemistry}},
  pages 295--354. John Wiley \& Sons, Ltd, 2014.
\newblock ISBN 978-1-118-74263-1.
\newblock \doi{10.1002/9781118742631.ch11}.
\newblock URL
  \url{https://onlinelibrary.wiley.com/doi/abs/10.1002/9781118742631.ch11}.

\bibitem[Carbone and Jenčová(2020)]{carbone_period_2020}
Raffaella Carbone and Anna Jenčová.
\newblock On {Period}, {Cycles} and {Fixed} {Points} of a {Quantum} {Channel}.
\newblock \emph{Annales Henri Poincaré}, 21\penalty0 (1):\penalty0 155--188,
  January 2020.
\newblock ISSN 1424-0661.
\newblock \doi{10.1007/s00023-019-00861-9}.
\newblock URL \url{https://doi.org/10.1007/s00023-019-00861-9}.

\bibitem[Watrous(2018)]{watrous_theory_2018}
John Watrous.
\newblock \emph{The {Theory} of {Quantum} {Information}}.
\newblock Cambridge University Press, Cambridge, 2018.
\newblock ISBN 978-1-107-18056-7.
\newblock \doi{10.1017/9781316848142, online version at
  https://cs.uwaterloo.ca/\%7Ewatrous/TQI/, accessed Jan 2023}.
\newblock URL
  \url{https://www.cambridge.org/core/books/theory-of-quantum-information/AE4AA5638F808D2CFEB070C55431D897}.

\bibitem[Junge and Xu(2006)]{junge_noncommutative_2006}
Marius Junge and Quanhua Xu.
\newblock Noncommutative maximal ergodic theorems.
\newblock \emph{Journal of the American Mathematical Society}, 20\penalty0
  (2):\penalty0 385--439, 2006.
\newblock ISSN 0894-0347, 1088-6834.
\newblock \doi{10.1090/S0894-0347-06-00533-9}.
\newblock URL \url{https://www.ams.org/jams/2007-20-02/S0894-0347-06-00533-9/}.

\bibitem[Misra and Sudarshan(1977)]{misra_zenos_1977}
B.~Misra and E.~C.~G. Sudarshan.
\newblock The {Zeno}{\textquoteright}s paradox in quantum theory.
\newblock \emph{Journal of Mathematical Physics}, 18\penalty0 (4):\penalty0
  756--763, 1977.
\newblock ISSN 0022-2488.
\newblock \doi{10.1063/1.523304}.
\newblock URL \url{https://doi.org/10.1063/1.523304}.

\bibitem[Barankai and Zimbor{\'a}s(2018)]{barankai_generalized_2018}
Norbert Barankai and Zolt{\'a}n Zimbor{\'a}s.
\newblock Generalized quantum {Zeno} dynamics and ergodic means.
\newblock November 2018.
\newblock \doi{10.48550/arXiv.1811.02509}.
\newblock URL \url{http://arxiv.org/abs/1811.02509}.
\newblock arXiv:1811.02509 [math-ph, physics:quant-ph].

\bibitem[M{\"o}bus and Wolf(2019)]{mobus_quantum_2019}
Tim M{\"o}bus and Michael~M. Wolf.
\newblock Quantum {Zeno} effect generalized.
\newblock \emph{Journal of Mathematical Physics}, 60\penalty0 (5):\penalty0
  052201, May 2019.
\newblock ISSN 0022-2488.
\newblock \doi{10.1063/1.5090912}.
\newblock URL \url{https://doi.org/10.1063/1.5090912}.

\bibitem[Burgarth et~al.(2020)Burgarth, Facchi, Nakazato, Pascazio, and
  Yuasa]{burgarth_quantum_2020}
Daniel Burgarth, Paolo Facchi, Hiromichi Nakazato, Saverio Pascazio, and Kazuya
  Yuasa.
\newblock Quantum {Zeno} {Dynamics} from {General} {Quantum} {Operations}.
\newblock \emph{Quantum}, 4:\penalty0 289, July 2020.
\newblock \doi{10.22331/q-2020-07-06-289}.
\newblock URL \url{https://quantum-journal.org/papers/q-2020-07-06-289/}.
\newblock Publisher: Verein zur F{\"o}rderung des Open Access Publizierens in
  den Quantenwissenschaften.

\bibitem[Becker et~al.(2021)Becker, Datta, and Salzmann]{becker_quantum_2021}
Simon Becker, Nilanjana Datta, and Robert Salzmann.
\newblock Quantum {Zeno} {Effect} in {Open} {Quantum} {Systems}.
\newblock \emph{Annales Henri Poincar{\'e}}, 22\penalty0 (11):\penalty0
  3795--3840, November 2021.
\newblock ISSN 1424-0661.
\newblock \doi{10.1007/s00023-021-01075-8}.
\newblock URL \url{https://doi.org/10.1007/s00023-021-01075-8}.

\bibitem[M{\"o}bus and Rouz{\'e}(2023)]{mobus_optimal_2023}
Tim M{\"o}bus and Cambyse Rouz{\'e}.
\newblock Optimal {Convergence} {Rate} in the {Quantum} {Zeno} {Effect} for
  {Open} {Quantum} {Systems} in {Infinite} {Dimensions}.
\newblock \emph{Annales Henri Poincar{\'e}}, 24\penalty0 (5):\penalty0
  1617--1659, May 2023.
\newblock ISSN 1424-0661.
\newblock \doi{10.1007/s00023-022-01241-6}.
\newblock URL \url{https://doi.org/10.1007/s00023-022-01241-6}.

\bibitem[Kato(1950)]{kato_adiabatic_1950}
Tosio Kato.
\newblock On the {Adiabatic} {Theorem} of {Quantum} {Mechanics}.
\newblock \emph{Journal of the Physical Society of Japan}, 5\penalty0
  (6):\penalty0 435--439, November 1950.
\newblock ISSN 0031-9015.
\newblock \doi{10.1143/JPSJ.5.435}.
\newblock URL \url{https://journals.jps.jp/doi/10.1143/JPSJ.5.435}.
\newblock Publisher: The Physical Society of Japan.

\bibitem[Burgarth et~al.(2019)Burgarth, Facchi, Nakazato, Pascazio, and
  Yuasa]{burgarth_generalized_2019}
Daniel Burgarth, Paolo Facchi, Hiromichi Nakazato, Saverio Pascazio, and Kazuya
  Yuasa.
\newblock Generalized {Adiabatic} {Theorem} and {Strong}-{Coupling} {Limits}.
\newblock \emph{Quantum}, 3:\penalty0 152, June 2019.
\newblock \doi{10.22331/q-2019-06-12-152}.
\newblock URL \url{https://quantum-journal.org/papers/q-2019-06-12-152/}.
\newblock Publisher: Verein zur F{\"o}rderung des Open Access Publizierens in
  den Quantenwissenschaften.

\bibitem[Zanardi and Campos~Venuti(2014)]{zanardi_coherent_2014}
Paolo Zanardi and Lorenzo Campos~Venuti.
\newblock Coherent {Quantum} {Dynamics} in {Steady}-{State} {Manifolds} of
  {Strongly} {Dissipative} {Systems}.
\newblock \emph{Physical Review Letters}, 113\penalty0 (24):\penalty0 240406,
  December 2014.
\newblock ISSN 0031-9007, 1079-7114.
\newblock \doi{10.1103/PhysRevLett.113.240406}.
\newblock URL \url{https://link.aps.org/doi/10.1103/PhysRevLett.113.240406}.

\bibitem[Burgarth et~al.(2022)Burgarth, Facchi, Gramegna, and
  Yuasa]{burgarth_one_2022}
Daniel Burgarth, Paolo Facchi, Giovanni Gramegna, and Kazuya Yuasa.
\newblock One bound to rule them all: from {Adiabatic} to {Zeno}.
\newblock \emph{Quantum}, 6:\penalty0 737, June 2022.
\newblock \doi{10.22331/q-2022-06-14-737}.
\newblock URL \url{https://quantum-journal.org/papers/q-2022-06-14-737/}.
\newblock Publisher: Verein zur F{\"o}rderung des Open Access Publizierens in
  den Quantenwissenschaften.

\bibitem[Suzuki(1976)]{suzuki_generalized_1976}
Masuo Suzuki.
\newblock Generalized {Trotter}'s formula and systematic approximants of
  exponential operators and inner derivations with applications to many-body
  problems.
\newblock \emph{Communications in Mathematical Physics}, 51\penalty0
  (2):\penalty0 183--190, June 1976.
\newblock ISSN 1432-0916.
\newblock \doi{10.1007/BF01609348}.
\newblock URL \url{https://doi.org/10.1007/BF01609348}.

\bibitem[Erez et~al.(2004)Erez, Aharonov, Reznik, and
  Vaidman]{erez_correcting_2004}
Noam Erez, Yakir Aharonov, Benni Reznik, and Lev Vaidman.
\newblock Correcting quantum errors with the {Zeno} effect.
\newblock \emph{Physical Review A}, 69\penalty0 (6):\penalty0 062315, June
  2004.
\newblock \doi{10.1103/PhysRevA.69.062315}.
\newblock URL \url{https://link.aps.org/doi/10.1103/PhysRevA.69.062315}.
\newblock Publisher: American Physical Society.

\bibitem[Viola and Lloyd(1998)]{viola_dynamical_1998}
Lorenza Viola and Seth Lloyd.
\newblock Dynamical suppression of decoherence in two-state quantum systems.
\newblock \emph{Physical Review A}, 58\penalty0 (4):\penalty0 2733--2744,
  October 1998.
\newblock \doi{10.1103/PhysRevA.58.2733}.
\newblock URL \url{https://link.aps.org/doi/10.1103/PhysRevA.58.2733}.
\newblock Publisher: American Physical Society.

\bibitem[Viola et~al.(1999)Viola, Knill, and Lloyd]{viola_dynamical_1999}
Lorenza Viola, Emanuel Knill, and Seth Lloyd.
\newblock Dynamical {Decoupling} of {Open} {Quantum} {Systems}.
\newblock \emph{Physical Review Letters}, 82\penalty0 (12):\penalty0
  2417--2421, March 1999.
\newblock \doi{10.1103/PhysRevLett.82.2417}.
\newblock URL \url{https://link.aps.org/doi/10.1103/PhysRevLett.82.2417}.
\newblock Publisher: American Physical Society.

\bibitem[Gough and Nurdin(2017)]{gough_can_2017}
John~E. Gough and Hendra~I. Nurdin.
\newblock Can quantum {Markov} evolutions ever be dynamically decoupled?
\newblock In \emph{2017 {IEEE} 56th {Annual} {Conference} on {Decision} and
  {Control} ({CDC})}, pages 6155--6160, Melbourne, VIC, December 2017. IEEE.
\newblock ISBN 9781509028733.
\newblock \doi{10.1109/CDC.2017.8264587}.
\newblock URL \url{http://ieeexplore.ieee.org/document/8264587/}.

\bibitem[Beige et~al.(2000)Beige, Braun, Tregenna, and
  Knight]{beige_quantum_2000}
Almut Beige, Daniel Braun, Ben Tregenna, and Peter~L. Knight.
\newblock Quantum {Computing} {Using} {Dissipation} to {Remain} in a
  {Decoherence}-{Free} {Subspace}.
\newblock \emph{Physical Review Letters}, 85\penalty0 (8):\penalty0 1762--1765,
  August 2000.
\newblock \doi{10.1103/PhysRevLett.85.1762}.
\newblock URL \url{https://link.aps.org/doi/10.1103/PhysRevLett.85.1762}.

\bibitem[Liu and Lieu(2024)]{liu_dissipative_2024}
Yu-Jie Liu and Simon Lieu.
\newblock Dissipative phase transitions and passive error correction.
\newblock \emph{Physical Review A}, 109\penalty0 (2):\penalty0 022422, February
  2024.
\newblock ISSN 2469-9926, 2469-9934.
\newblock \doi{10.1103/PhysRevA.109.022422}.
\newblock URL \url{https://link.aps.org/doi/10.1103/PhysRevA.109.022422}.

\bibitem[noa(2021)]{noauthor_solving_2021}
Solving the {Lindblad} dynamics of a qubit chain {\textemdash} {Qiskit}
  {Dynamics} 0.4.2 documentation, 2021.
\newblock URL
  \url{https://qiskit.org/ecosystem/dynamics/tutorials/Lindblad_dynamics_simulation.html}.

\bibitem[Canonne(2023)]{canonne_short_2023}
Clément~L. Canonne.
\newblock A short note on an inequality between {KL} and {TV}.
\newblock August 2023.
\newblock \doi{10.48550/arXiv.2202.07198}.
\newblock URL \url{http://arxiv.org/abs/2202.07198}.
\newblock arXiv:2202.07198 [math, stat] version: 2.

\bibitem[Vajda(1970)]{vajda_note_1970}
I.~Vajda.
\newblock Note on discrimination information and variation ({Corresp}.).
\newblock \emph{IEEE Transactions on Information Theory}, 16\penalty0
  (6):\penalty0 771--773, November 1970.
\newblock ISSN 1557-9654.
\newblock \doi{10.1109/TIT.1970.1054557}.
\newblock URL \url{https://ieeexplore.ieee.org/document/1054557}.

\bibitem[Laracuente and Smith(2023)]{laracuente_information_2023}
Nicholas Laracuente and Graeme Smith.
\newblock Information {Fragility} or {Robustness} {Under} {Quantum} {Channels}.
\newblock December 2023.
\newblock \doi{10.48550/arXiv.2312.17450}.
\newblock URL \url{http://arxiv.org/abs/2312.17450}.
\newblock arXiv:2312.17450 [quant-ph].

\bibitem[Murali et~al.(2019)Murali, Baker, Javadi-Abhari, Chong, and
  Martonosi]{murali_noise-adaptive_2019}
Prakash Murali, Jonathan~M. Baker, Ali Javadi-Abhari, Frederic~T. Chong, and
  Margaret Martonosi.
\newblock Noise-{Adaptive} {Compiler} {Mappings} for {Noisy}
  {Intermediate}-{Scale} {Quantum} {Computers}.
\newblock In \emph{Proceedings of the {Twenty}-{Fourth} {International}
  {Conference} on {Architectural} {Support} for {Programming} {Languages} and
  {Operating} {Systems}}, {ASPLOS} '19, pages 1015--1029, New York, NY, USA,
  April 2019. Association for Computing Machinery.
\newblock ISBN 978-1-4503-6240-5.
\newblock \doi{10.1145/3297858.3304075}.
\newblock URL \url{https://dl.acm.org/doi/10.1145/3297858.3304075}.

\bibitem[Alexander et~al.(2020)Alexander, Kanazawa, Egger, Capelluto, Wood,
  Javadi-Abhari, and McKay]{alexander_qiskit_2020}
Thomas Alexander, Naoki Kanazawa, Daniel~J. Egger, Lauren Capelluto,
  Christopher~J. Wood, Ali Javadi-Abhari, and David~C. McKay.
\newblock Qiskit pulse: programming quantum computers through the cloud with
  pulses.
\newblock \emph{Quantum Science and Technology}, 5\penalty0 (4):\penalty0
  044006, August 2020.
\newblock ISSN 2058-9565.
\newblock \doi{10.1088/2058-9565/aba404}.
\newblock URL \url{https://dx.doi.org/10.1088/2058-9565/aba404}.
\newblock Publisher: IOP Publishing.

\bibitem[Krantz et~al.(2019)Krantz, Kjaergaard, Yan, Orlando, Gustavsson, and
  Oliver]{krantz_quantum_2019}
P.~Krantz, M.~Kjaergaard, F.~Yan, T.~P. Orlando, S.~Gustavsson, and W.~D.
  Oliver.
\newblock A quantum engineer's guide to superconducting qubits.
\newblock \emph{Applied Physics Reviews}, 6\penalty0 (2):\penalty0 021318, June
  2019.
\newblock ISSN 1931-9401.
\newblock \doi{10.1063/1.5089550}.
\newblock URL
  \url{https://pubs.aip.org/apr/article/6/2/021318/570326/A-quantum-engineer-s-guide-to-superconducting}.

\bibitem[Aharonov et~al.(2023)Aharonov, Gao, Landau, Liu, and
  Vazirani]{aharonov_polynomial-time_2023}
Dorit Aharonov, Xun Gao, Zeph Landau, Yunchao Liu, and Umesh Vazirani.
\newblock A {Polynomial}-{Time} {Classical} {Algorithm} for {Noisy} {Random}
  {Circuit} {Sampling}.
\newblock In \emph{Proceedings of the 55th {Annual} {ACM} {Symposium} on
  {Theory} of {Computing}}, {STOC} 2023, pages 945--957, New York, NY, USA,
  June 2023. Association for Computing Machinery.
\newblock ISBN 978-1-4503-9913-5.
\newblock \doi{10.1145/3564246.3585234}.
\newblock URL \url{https://dl.acm.org/doi/10.1145/3564246.3585234}.

\bibitem[Walter et~al.(2014)Walter, Nunnenkamp, and
  Bruder]{walter_quantum_2014}
Stefan Walter, Andreas Nunnenkamp, and Christoph Bruder.
\newblock Quantum {Synchronization} of a {Driven} {Self}-{Sustained}
  {Oscillator}.
\newblock \emph{Physical Review Letters}, 112\penalty0 (9):\penalty0 094102,
  March 2014.
\newblock \doi{10.1103/PhysRevLett.112.094102}.
\newblock URL \url{https://link.aps.org/doi/10.1103/PhysRevLett.112.094102}.
\newblock Publisher: American Physical Society.

\bibitem[Ben~Arosh et~al.(2021)Ben~Arosh, Cross, and
  Lifshitz]{ben_arosh_quantum_2021}
Lior Ben~Arosh, M.~C. Cross, and Ron Lifshitz.
\newblock Quantum limit cycles and the {Rayleigh} and van der {Pol}
  oscillators.
\newblock \emph{Physical Review Research}, 3\penalty0 (1):\penalty0 013130,
  February 2021.
\newblock \doi{10.1103/PhysRevResearch.3.013130}.
\newblock URL \url{https://link.aps.org/doi/10.1103/PhysRevResearch.3.013130}.
\newblock Publisher: American Physical Society.

\end{thebibliography}

\newpage

\appendix
\newpage

\section{Strong Damping, Zeno Dynamics, and Entropy Decay} \label{app:Zeno}
Much of this Section is devoted to a technical reanalysis of the generalized Zeno, strong damping, and related effects \cite{burgarth_generalized_2019, mobus_quantum_2019, burgarth_quantum_2020}. Our estimates are based on cp-order inequalities and seek comparability with CMLSI and $\lambda$-decay. The bounds derived herein are nonetheless in terms of norms. These bounds are in principle very general, requiring only sub-multiplicativity of $\| \cdot \|_{\A \rightarrow \B}$ in addition to its being a norm. A restriction, however, is that many of the results must assume contractivity of most or all maps involved. The diamond norm is especially convenient in this sense, as channels automatically satisfy this assumption. Though many of the techniques do not require finite dimension, we assume that the diamond norm is well-defined and bounded for the superoperators considered herein.

To simplify notation, let
\begin{equation}
\epow{a}{m} := \frac{a^m \exp(a)}{m!}
\end{equation}
for any scalar $a > 0$, $k \in \NN$.
\begin{rem} \label{rem:expapprox}
For any $a > 0$ such that $\exp(a)$ equals its Taylor series,
\[ \exp(a) - \sum_{n=0}^k \frac{a^n}{n!} = \sum_{n={k+1}}^\infty \frac{a^n}{n!} = a^{k+1} \sum_{n=0}^\infty \frac{a^n}{(n+k+1)!}
	\leq \frac{a^{k+1} \exp(a)}{(k+1)!} = \epow{a}{k+1} \pl. \]
\end{rem}
\begin{lemma} \label{lem:exphaml}
Let $\L$, and $\E$ be respectively a Lindbladian and a map on the same von Neumann algebra. Let $\A$ be a normed input subspace that is preserved by $\L$, and $\B$ be the normed output space of $\E$. Then for any $t \in \RR$
\[ \Big \| \E \circ  \sum_{m=k}^\infty \frac{(it)^m}{m!} \L^m \Big \|_{\A \rightarrow \B, (cb)}
	\leq \epow{\|\L\|_{\A, (cb)} t}{k} \|\E\|_{\A \rightarrow \B, (cb)} \pl. \]
\end{lemma}
\begin{proof}
First, we name a given input $\rho$ and use the triangle inequality to separate terms.
\begin{equation} \label{eq:termsplit1}
\begin{split}
\Big \| \E \circ\sum_{m=k}^\infty \frac{(it)^m}{m!} \L^m \Big \|_{\A \rightarrow \B, (cb)}
	\leq \sum_{m=k}^\infty \frac{t^m}{m!} \| \E \circ \L^m \|_{\A \rightarrow \B, (cb)} \pl.
\end{split}
\end{equation}
We then consider each term.
\[ \| \E \circ \L^m(\rho) \|_{\A \rightarrow \B, (cb)} \leq \| \E \|_{\A \rightarrow \B, cb} \| \L^m(\rho) \|_{\A, (cb)} \pl. \]
The proof then follows from Remark \ref{rem:expapprox}.
\end{proof}
\begin{lemma} \label{lem:tunsum}
Let $(f_m)_{m=1}^k, (g_m)_{m=1}^k$ be families of maps for $k \in \NN$ such that $f_m \circ f_{m-1}$ and $g_{m} \circ g_{m-1}$ are valid compositions. Then
\[ \prod_{m=1}^k f_m - \prod_{m=1}^k g_m = \sum_{l=1}^k \Big (\prod_{n=l+1}^{k} g_n \Big )(f_{l} - g_{l})\Big ( \prod_{m=1}^{l-1} f_m  \Big ) \pl, \]
Moreover, if the maps are contractions in norm $\| \cdot \|$, then
\[ \bigg \| \prod_{m=k}^1 g_m - \prod_{m=k}^1 f_m \bigg \|
	\leq \sum_{l=1}^{k} \| f_l - g_l \| \pl. \]
\end{lemma}
\begin{proof}
Let $\omega_l = \big ( \prod_{m=1}^l f_m \big ) (\rho)$ for input $\rho$ and each $l \in 1...k$. For each value of $l$,
\[ \Big ( \prod_{n=l+1}^{k} g_n \Big ) (\omega_l)  - \Big (\prod_{n=l}^k g_n \Big ) (\omega_{l-1})
	= \Big ( \prod_{n=l+1}^{k} g_n \Big) (f_l - g_l)(\omega_{l-1}) \pl. \]
The Lemma's first Equation then follows from induction.

We apply the Lemma's first Equation to the left hand side of its second, obtaining that
\[ \bigg \| \prod_{m=k}^1 g_m - \prod_{m=k}^1 f_m \bigg \|
    = \sup_\rho \frac{1}{\|\rho\|} \bigg \| \sum_{l=1}^k \Big (\prod_{n=k}^{l+1} g_n \Big )(f_{l} - g_{l}) \Big ( \prod_{m=l-1}^{1} f_m \Big ) (\rho) \bigg \| \pl. \]
Via the triangle inequality, we may separate the terms in the sum. We then split the product via submultiplicativity. The overall supremum over $\rho$ then underestimates the per-term and per-factor suprema, completing the Lemma.
\end{proof}

\begin{lemma} \label{lem:firstord}
Let $(\Phi_m)_{m=1}^k$ be a family of contractive maps on space $\A$ with submultiplicative norm $\| \cdot \|$. Let $(\L_m)_{m=1}^k$ be a family of bounded Lindbladians generating contractive semigroups. Let $t_1, ..., t_k \in \RR^+$. Then
\[ \bigg \| \prod_{m=1}^k (\Phi_m \circ e^{-\L_m t_m})- \prod_{m=1}^k (\Phi_m \circ (1 -\L_m t_m)) \bigg \|
	\leq \sum_{m=1}^k \epow{\|\L_m \| t_m}{2} \pl. \]
\end{lemma}
\begin{proof}
The Lemma follows from noting that $(1 -\L t)$ is the 1st order Taylor expansion of $e^{-\L t}$ for any $t \in \RR^+$, so
\[ \|e^{-\L_m t_m} - (1 -\L t_m)\| \leq \epow{\|\L_m\| t_m}{2}\]
for each $m \in 1...k$. Lemma \ref{lem:tunsum} completes this Lemma.
\end{proof}
While it is often intuitive to think of a Lindbladian as having units of inverse time and appearing alongside a time parameter in the expression $\exp(- \L t)$, $t$ is redundant in many of the mathematical expressions we will use. When $t$ is an overall parameter (not changing by interval as in Lemma \ref{lem:firstord}), we may instead write $\exp(- \L)$, implicitly absorbing $t$ as a multiplying factor in $\L$. Doing so simplifies notation, and one may trivially re-extract the parameter $t$ by substituting $\L \rightarrow t \L$ in resulting expressions.

The next Lemma uses additional notation. For $m < k \in \NN$, let
\[ WS(m,k) \subset \{ [r_0+1, ..., r_1], ..., [r_{m-1}+1, ..., r_m] : 0 = r_0 < \dots < r_m = k \}\]
denote the set of partitions of $k$ into $m$ contiguous, ordered, non-overlapping intervals. For given $W \in WS(m,k)$, let $W(j)$ denotes a contiguous sequence of indices for $j \in 1...m$. Let $|W(j)|$ denote the number of indices in $W(j)$, which we will refer to as its length. By $W(j)[l]$ we denote the $l$th index in $W(j)$ for $l \in 1...|W(j)|$. As an example, we might take $W = ([1,2,3], [4,5]) \in WS(2,5)$, in which case $W(1) = [1,2,3], W(2) = [4,5]$, and $W(2)[1] = 4$. In this example, we would have $r_1 = 3, r_2 = 5$.

\begin{lemma} \label{lem:zenolem}
Let $(\Phi_m)_{m=1}^k$ be a family of contractions for any $k \in \NN$. Let $(\L_m)_{m=1}^k$ be bounded Lindbladians such that $e^{-\L_m}$ is also contractive for each $m$. Let $\|\L\| := \max_m \{\| \L_m \|\}$. Assume for given $g \in \NN$, projection $\E$, and $\epsilon \in (0,1)$ that for any contiguous indices $1 \leq j_1, ..., j_n \leq k$, $\Phi_{j_1} ... \Phi_{j_n} \geq_{cp} (1 - \epsilon) \E$, and that $\Phi_j \E = \E \Phi_j = \E$ for each $j \in 1...k$. Then
\begin{equation*}
\begin{split}
& \bigg \| \sum_{m=0}^k \frac{(-1)^m}{k^m} \sum_{W \in WS(m+1,k)} \Big ( \Big ( \prod_{j=1}^m (\E \L_{W,j})
    - \prod_{j=1}^m (\Phi_{W(j)} \L_{W,j}) \Big ) \Phi_{W(m+1)} \Big ) \E \bigg \|
\\ & \leq \frac{g}{k} \Big ( 1 + \frac{1}{k} \Big ) \frac{1}{1-\epsilon} (\|\L\|^3 + 4 \|\L\|^2 + 2 \|\L\|) e^{\|\L\|}
\end{split}
\end{equation*}
where $\Phi_{W(j)} = \Phi_{W(j)[|W(j)|]} \circ ... \circ \Phi_{W(j)[1]}$, and each $\L_{W,j} \in (\L_m)$ is the Lindbladian appearing between the channel compositions $W(j)$ and $W(j+1)$.
\end{lemma}
\begin{proof}
For each $n \in 1...k$, let $\#(n, WS(m+1,k))$ denote the number of partitions in $WS(m+1,k)$ that contain at least one interval of length exactly $n$. To estimate $\#(n, WS(m+1,k))$, consider placing $m-1$ partition boundaries (creating $m$ partitions), then placing a final partition boundary that is exactly $n$ positions away from one of those already placed. To place the first $1...k$ boundaries has $(k+1 \text{ choose } m-1)$ choices when including placing before the 1st or after the $k$th as well as between any indices. To place the final boundary has 2 choices for each already placed boundary, although these terms are subsequently divided by a factor of 2 to account for exchange between the final and $n$-positions-away pre-existing boundary, plus another 2 choices for placement at either end of the interval, yielding a $m+1$ possibilities. The estimate $\#(n, WS(m+1,k)) \leq (k+1 \text{ choose } m-1) \times (m+1)$ is an overcount, since it may double-count if there was already a length-$n$ partition between the first $m-1$ boundaries, and since placing a boundary $n$ positions away from another may not create an $n$-length interval if another boundary is in between. To simplify, as long as $m > 1$,
\begin{equation} \label{eq:termnum}
    \#(n, WS(m+1,k)) \leq \binom{k+1}{m-1} (m+1) = \frac{(k+1) k! (m+1)}{(m-1)!(k - m + 2)!} \leq \frac{(k+1)k^{m-2} (m+1)}{(m-1)!} \pl.
\end{equation}
Futhermore, $\#(n, WS(m+1,k))$ overcounts the number of terms containing one partition of length $n$ and no smaller partitions.

Assuming that each $|W_j| \geq n$,
\begin{equation*}
\begin{split}
    & \Big ( \prod_{j=1}^m (\Phi_{W(j)} \L_{W,j}) \Big ) \Phi_{W(m+1)}
    \\ & = \prod_{j=1}^m (((1-\epsilon^{\lfloor n / g \rfloor}) \E + \epsilon^{\lfloor n / g \rfloor}  \Psi_{W(j)}) \L_{W,j}) ((1-\epsilon^{\lfloor n / g \rfloor}) \E + \epsilon^{\lfloor n / g \rfloor} \Psi_{W(m+1)} )
\end{split}
\end{equation*}
for some contractions $(\Psi_{W(j)})_{j = 1}^{m+1}$ such that $\E \Psi_{W(j)} = \Psi_{W(j)} \E = \E$. Hence
\begin{equation} \label{eq:singleterm}
\Big \| \Big ( \prod_{j=1}^m (\Phi_{W(j)} \L_{W,j}) \Big ) \Phi_{W(m+1)} \E - \prod_{j=1}^m (\E \L_{W,j}) \E \Big \|
    \leq \Big ( \prod_{j=1}^m \|\L_{W,j}\| \Big ) \big (1 - (1-\epsilon^{\lfloor n / g \rfloor})^{m} \big ) \pl,
\end{equation}
 assuming that the smallest partition in $W$ is of length at least $n$.
 Via Bernoulli's inequality, $(1 - (1-\epsilon^{\lfloor n / g \rfloor})^m \big ) \leq m \epsilon^{\lfloor n / g \rfloor}$ when $m > 0$.

To bound the full expression as in the Lemma,
\begin{equation*}
\begin{split}
& \bigg \| \sum_{m=0}^k \frac{(-1)^m}{k^m}
    \bigg (  \sum_{W \in WS(m+1,k)} \bigg (  \Big ( \prod_{j=1}^m (\Phi_{W(j)} \L_{W,j}) \Big ) \Phi_{W(m+1)} \E - \prod_{j=1}^m (\E \L_{W,j}) \E \big ) \bigg ) \bigg ) \bigg \|
\\ & \leq \sum_{m=0}^k \frac{(-1)^m}{k^m} \sum_{W \in WS(m+1,k)}
      \bigg \| \bigg (  \Big ( \prod_{j=1}^m (\Phi_{W(j)} \L_{W,j}) \Big ) \Phi_{W(m+1)} \E - \prod_{j=1}^m (\E \L_{W,j}) \E \big ) \bigg ) \bigg \|
\\ & \leq \sum_{m=0}^k \frac{(-1)^m}{k^m} \sum_{n=1}^k \sum_{W \in WS(m+1,k,n)}
      \bigg \| \bigg (  \Big ( \prod_{j=1}^m (\Phi_{W(j)} \L_{W,j}) \Big ) \Phi_{W(m+1)} \E - \prod_{j=1}^m (\E \L_{W,j}) \E \big ) \bigg ) \bigg \|
    \pl,
\end{split}
\end{equation*}
where $W \in WS(m+1,k,n)$ counts all partionings with minimum inveral size $n$. Combining with Equations \eqref{eq:termnum} and \eqref{eq:singleterm},
\begin{equation} \label{eq:zlemmamain}
\begin{split}
& \bigg \| \sum_{m=0}^k \frac{(-1)^m}{k^m} \bigg (  \sum_{W \in WS(m+1,k)} \bigg (  \Big ( \prod_{j=1}^m (\Phi_{W(j)} \L_{W,j}) \Big ) \Phi_{W(m+1)} \E - \prod_{j=1}^m (\E \L_{W,j}) \E \big ) \bigg ) \bigg ) \bigg \|
\\ & \leq \Big ( \frac{1}{k} + \frac{1}{k^2} \Big ) \Big ( \sum_{n=0}^k  \epsilon^{\lfloor n / g \rfloor} \Big ) \sum_{m=1}^k \|\L \|^m\frac{(m+1) m}{(m-1)!} \pl.
\end{split}
\end{equation}
Estimating the first sum,
\begin{equation} \label{eq:nsum}
    \sum_{n=0}^k  \epsilon^{\lfloor n / g \rfloor} \leq g \sum_{n=0}^{\lceil k/g \rceil } \epsilon^n \leq \frac{g}{1-\epsilon} \pl.
\end{equation}
For the second sum,
\begin{equation} \label{eq:msummain}
    \sum_{m=1}^k \|\L \|^m \frac{(m+1) m}{(m-1)!} = \sum_{m=1}^k \|\L \|^m \frac{m^2}{(m-1)!} + \sum_{m=1}^k \|\L \|^m \frac{m}{(m-1)!} \pl.
\end{equation}
Note that whenever a sum starting at $m=0$ contains a factor of $m$ in all terms, it equivalently starts at $m=1$. The first term in Equation \eqref{eq:msummain} decomposes as
\begin{equation} \label{eq:msumterm1}
\begin{split}
& \sum_{m=1}^k \|\L \|^m\frac{m^2}{(m-1)!} = \|\L\| \sum_{m=0}^\infty \frac{(m+1)^2}{m!} \|\L\|^m
    \\ & = \|\L\| \Big ( \sum_{m=1}^k \frac{m^2}{m!} \|\L\|^m + 2 \sum_{m=1}^k \frac{m}{m!} \| \L \|^m + \sum_{m=0}^k \frac{1}{m!} \| \L \|^m \Big ) \pl.
\end{split}
\end{equation}
The second term in Equation \eqref{eq:msummain} is equivalent to
\begin{equation} \label{eq:msumterm2}
\begin{split}
\sum_{m=1}^k \|\L \|^m \frac{m}{(m-1)!} & = \| \L \| \sum_{m=0}^k \frac{m+1}{m!} \| \L \|^m
         = \| \L \| \sum_{m=0}^k \frac{m}{m!} \| \L \|^m + \| \L \| \sum_{m=0}^k \frac{1}{m!} \| \L \|^m \pl,
\end{split}
\end{equation}
which also happens to equal the first term in parenthesis in the last line of Equation \eqref{eq:msumterm1}. Moreover,
\begin{equation} \label{eq:onem}
    \sum_{m=1}^k \frac{m}{m!} \|\L\|^m =  \sum_{m=1}^k \frac{1}{(m-1)!} \| \L\|^m
        \leq \|\L\| \sum_{m=0}^\infty \frac{1}{m!} \|\L\|^m = \|\L\| e^{\|\L\|} \pl.
\end{equation}
Therefore, the terms in Equation \eqref{eq:msumterm2} are equivalent to $(\| \L \|^2 + \| \L \|) e^{\| \L \|}$. Collecting terms from Equation \eqref{eq:msumterm1} yields $(\|\L\|^3 + 3 \| \L \|^2 + \|\L\|) e^{\|\L\|}$. Adding these expressions yields
\begin{equation} \label{eq:msumfinal}
    \sum_{m=1}^k \|\L \|^m\frac{m^2}{(m-1)!} \leq (\|\L\|^3 + 4 \|\L\|^2 + 2 \|\L\|) e^{\|\L\|} \pl.
\end{equation}
Equations \eqref{eq:nsum} and \eqref{eq:msumfinal} with \eqref{eq:zlemmamain} complete the Lemma.
\end{proof}
\begin{lemma} \label{lem:zenolemfinal}
Let $(\Phi_{m})_{m=1}^k$ be a family of norm contractions for any $k \in \NN$ all having fixed point projector $\E$. Let $(\L_m)_{m=1}^k$ be bounded Lindbladians such that $e^{-\L_m s}$ is contractive for all $s \in \RR^+$ and $m$. Let $\|\L\| := \max_m \|\L_m\|$. Assume for given $g \in \NN$, projection $\E$, and $\epsilon \in (0,1)$ that for any contiguous indices $1 \leq j_1, ..., j_n \leq k$, $\Phi_{j_1} ... \Phi_{j_n} \geq_{cp} (1 - \epsilon) \E$, and that $\Phi_j \E = \E \Phi_j = \E$ for each $j \in 1...k$. Then
\begin{equation*}
\begin{split}
& \Big \| \prod_{m=1}^k \big (\Phi_{m} \circ e^{- \L_m / k} \big ) \E - \prod_{m=1}^k e^{- \E \L_m \E / k} \E \Big \|
        \leq 2 k \epow{\| \L \| / k}{2} + \frac{g}{k} \Big ( 1 + \frac{1}{k} \Big ) \frac{\|\L\|^3 + 4 \|\L\|^2 + 2 \|\L\|}{1-\epsilon} e^{\|\L\|} \pl .
\end{split}
\end{equation*}
\end{lemma}
\begin{proof}
    Via Lemma \ref{lem:firstord},
\[ \bigg \| \prod_{m=1}^k \big (\Phi_{m} \circ e^{- \L_m / k} \big ) - \prod_{m=1}^k (\Phi_m (1 - \L_m / k)) \bigg \|
    \leq k \epow{\| \L \| / k}{2} \pl. \]
Then via the binomial expansion,
\begin{equation} \label{eq:phiktaylor}
\prod_{m=1}^k (\Phi_m (1 - \L_m / k)) = \sum_{m=0}^{k} \frac{(-1)^m}{k^m} \sum_{W \in WS(m+1,k)} \Big ( \prod_{j=1}^m (\Phi_{W(j)} \L_{W,j}) \Big ) \Phi_{W(m+1)} \pl,
\end{equation}
where $\Phi_{W(j)}$ and $\L_{W,j}$ are defined as in Lemma \ref{lem:zenolem}. Similarly,
\[ \bigg \| \prod_{m=1}^k \big (\Phi_{m} \circ e^{- \E \circ \L_m \circ \E / k} \big ) - \prod_{j=1}^m (1 - \E \L_m \E / k) \bigg \| \leq \sum_{m=1}^k \epow{\| \E \L_m \E \| / k}{2} \leq k \epow{\| \L \| / k}{2} \pl. \]
Again via the binomial expansion and using that $\E$ absorbs $\Phi_j$ for all $j$,
\begin{equation} \label{eq:phiktaylorE}
     \prod_{j=1}^m (1 - \E \L_m \E / k) = \sum_{m=0}^{k} \frac{(-1)^m}{k^m} \sum_{W \in WS(m+1,k)} \Big ( \prod_{j=1}^m (\E \L_{W,j} \E) \Big ) \pl.
\end{equation}
Via idempotence of $\E$,
\[ \prod_{j=1}^m (\E \L_{W,j} \E) = \Big ( \prod_{j=1}^m (\E \L_{W,j}) \Big ) \E \]
Hence the desired norm is at most $2 k \epow{\| \L \| / k}{2}$ plus
\[ \bigg \| \sum_{m=0}^{k} \frac{(-1)^m}{k^m} \sum_{W \in WS(m+1,k)} \Big ( \prod_{j=1}^m (\E \L_{W,j}) - \Big ( \prod_{j=1}^m (\Phi_{W(j)} \L_{W,j}) \Big ) \Phi_{W(m+1)} \Big ) \E \bigg \| \pl. \]
Since the final $\E$ is a contraction, the expression in Lemma \ref{lem:zenolem} is an overestimate, completing this Lemma.
\end{proof}
\begin{theorem} \label{thm:zeno}
Let $(\Phi_{m})_{m=1}^k$ be a family of norm contractions for any $k \in \NN$. Let $(\L_m)_{m=1}^k$ be bounded Lindbladians such that $e^{-\L_m s}$ is contractive for all $s \in \RR^+$ and each $m \in 1...k$. Assume for given $g \in \NN$, projection $\E$, and $\epsilon \in (0,1)$ that for any contiguous indices $1 \leq j_1, ..., j_n \leq k$, $\Phi_{j_1} ... \Phi_{j_n} \geq_{cp} (1 - \epsilon) \E$, and that $\Phi_j \E = \E \Phi_j = \E$ for each $j \in 1...k$. Then given any $q \in \NN$ and $(r_m)_{m=0}^q$ such that $r_0 = 0$, $r_q = k$, and each $k_m := r_m - r_{m-1}$,
\begin{equation*}
\begin{split}
 \Big \| \prod_{m=1}^k & \big (\Phi_{m} \circ e^{- \L_m / k} \big ) \E - \prod_{m=1}^k e^{- \E \L_m \E / k} \E \Big \| 
    \leq \sum_{m=1}^q \max_{j \in r_{m-1} ... r_m}  \bigg ( 
      \\ &  2 k_m \epow{\| \L_j \| / k}{2} +  \frac{g}{k_m} \Big ( 1 + \frac{1}{k_m} \Big ) \frac{(\|\L_j \| k_m / k)^3 + 4 (\|\L_j \| k_m / k )^2 + 2 \|\L_j \| k_m / k }{1-\epsilon} e^{\|\L_j \| k_m / k} \bigg ) .
\end{split}
\end{equation*}
\end{theorem}
\begin{proof}
Via Lemma \ref{lem:tunsum},
\begin{equation*}
\begin{split}
    & \prod_{m=1}^k \big (\Phi_{m} \circ e^{- \L_m / k} \big ) - \prod_{m=1}^k e^{- \E \L_m \E / k}
    \\ &    = \sum_{j=1}^q \Big (\prod_{m=1}^{l_j} \Phi_{m} \circ e^{- \L_m / k} \Big )
    \Big ( \prod_{m=l_j}^{r_j} \big (\Phi_{m} \circ e^{- \L_m / k} \big ) - \prod_{m=r_{j-1}}^{r_j} e^{- \E \L_m \E / k} \Big )
         \Big ( \prod_{m=r_j}^{k} e^{- \E \L_m \E / k} \Big ) \pl.
\end{split}
\end{equation*}
Therefore, using that the left and right products are of contractions, that the right map commutes with the projection $\E$, that the norm is assumed submultiplicative, and the triangle inequality,
\begin{equation*}
\begin{split}
    & \Big \| \prod_{m=1}^k \big (\Phi_{m} \circ e^{- \L_m / k} \big ) \E - \prod_{m=1}^k e^{- \E \L_m \E / k} \E \Big \|
    \leq \sum_{j=1}^q 
    \Big \| \prod_{m=l_j}^{r_j} \big (\Phi_{m} \circ e^{- \L_m / k} \big ) \E - \prod_{m=l_j}^{r_j} e^{- \E \L_m \E / k} \E \Big \| \pl.
\end{split}
\end{equation*}
Note that $\L_m / k$ = $(k_m / k) \L_m / k_m$. Lemma \ref{lem:zenolemfinal} completes the first inequality.
\end{proof}

\begin{rem}
    As shown in Corollary \ref{cor:ctsfinal}, Theorem \ref{thm:zeno} implies a strong damping inequality with linear dependence on the characteristic timescale of damping (inverse dependence on the strength), linear dependence on time, and quadratic dependence on the norm of the damped Lindbladian. Analogously, Theorem \ref{thm:zeno} implies a generalized Zeno bound for discrete interruption scaling as $O(1/k)$, and where the dependence on $\|\L\|$ can be made quadratic via appropriate choice of $q$.
\end{rem}

\begin{rem} \normalfont \label{rem:alternate}
    As an alternate approach to prove a bound similar to Theorem \ref{thm:zvscmlsi}, assuming that $(\Phi^t)$ has $\lambda$-CMLSI, for any input density $\rho$,
    \begin{align*}
        D(\Phi^t(\E_0(\rho)) \| \Phi^t(\E(\rho))) \leq e^{- \lambda t} D(\E_0 (\rho) \| \E (\rho)) \pl,
    \end{align*}
    so
    \begin{align*}
        D(\E_0 (\rho) \| \E (\rho)) - D(\Phi^t(\E_0(\rho)) \| \Phi^t(\E(\rho)))
            \geq (1 - e^{- \lambda t}) D(\E_0 (\rho) \| \E (\rho)) \pl.
    \end{align*}
    Using the bound from Lemma \ref{lem:wsb} and Equation \eqref{eq:pprelent}, there is an input density $\rho$ for which
    \begin{align*}
        \delta_t \log C(\E) + (1 + \delta_t) h \Big ( \frac{\delta_t}{1+\delta_t} \Big ) 
            \geq (1 - e^{- \lambda t}) D(\E_0 (\rho) \| \E (\rho)) \pl.
    \end{align*}
    Therefore,
    \begin{equation} \label{eq:zvc1}
        \lambda \leq - \frac{1}{t} \ln \bigg ( 1 - \frac{1}{\log  D(\E_0 (\rho) \| \E (\rho))} \Big ( \delta_t \log C(\E) + (1 + \delta_t) h \Big ( \frac{\delta_t}{1+\delta_t} \Big ) \Big ) \bigg ) \pl.
    \end{equation}
    To leading order in small $\delta_t$,
    \begin{align*}
        \lambda & \lesssim \frac{1}{t \log C(\E_0 \| \E)} \Big ( \delta_t \log C(\E) + (1 + \delta_t) h \Big ( \frac{\delta_t}{1+\delta_t} \Big ) \Big ) \bigg )
         \approx \frac{\tilde{\delta}}{\log C(\E_0 \| \E)} \Big (  \log C(\E) + \log (1 / t \tilde{\delta}) \Big ) \pl.
    \end{align*}
    Though the obtained bound is optimized by setting $t$ large, doing so re-introduces higher-order corrections in $\delta_t$. Hence while this bound requires small $\tilde{\delta}$ and therefore also contrains the relevant time range in terms of it, the bound does not yield a $1/\lambda_0$ scaling of $\lambda$ even for large $\lambda_0$.
\end{rem}

\end{document}